\documentclass[openacc]{rsproca_new}
\usepackage{amsthm}
\usepackage{comment}
\usepackage[nameinlink]{cleveref}
\usepackage{bm}
\usepackage{todonotes}
\usepackage[most]{tcolorbox}
\usepackage{mathtools}
\usepackage{pgfplots}
\usepackage{etoolbox}
\makeatletter
\AtBeginDocument{%
  \patchcmd{\@setref}{\bfseries ??}{\bfseries\textcolor{red}{??}}{}{}%
  \patchcmd{\@citex}{\reset@font\bfseries ?}{\reset@font\bfseries\textcolor{red}{?}}{}{}%
}
\makeatother
\newtoggle{showhfe}
\togglefalse{showhfe}
\newcommand{\hfe}[1]{\iftoggle{showhfe}{\begingroup\color{teal}#1\endgroup}{#1}}
\newtoggle{showcla}
\togglefalse{showcla}
\newcommand{\cla}[1]{\iftoggle{showcla}{\begingroup\color{magenta}#1\endgroup}{#1}}
\definecolor{uzcolor}{RGB}{123, 31, 162}
\newtoggle{showuz}
\togglefalse{showuz}
\newcommand{\uz}[1]{\iftoggle{showuz}{\begingroup\color{uzcolor}#1\endgroup}{#1}}

\newtoggle{showuzfix}
\togglefalse{showuzfix}
\newcommand{\uzfix}[1]{\iftoggle{showuzfix}{\begingroup\color{uzcolor}#1\endgroup}{#1}}

\definecolor{oscolor}{RGB}{6, 84, 101}
\newtoggle{showos}
\togglefalse{showos}
\newcommand{\os}[1]{\iftoggle{showos}{\begingroup\color{oscolor}#1\endgroup}{#1}}

\newtoggle{showfix}
\togglefalse{showfix}
\newcommand{\fix}[1]{\iftoggle{showfix}{\begingroup\color{red}#1\endgroup}{#1}}

\newcommand{\dotcross}{\mathbin{\ooalign{
  \hfil$\times$\hfil\cr
  \noalign{\kern-1.0ex}
  \hfil$\cdot$\hfil\cr
}}}

\newtheorem{theorem}{\bf Theorem}[section]

\newtheorem{corollary}{\bf Corollary}[section]
\definecolor{camred}{RGB}{163, 28, 48}
\definecolor{oxblue}{RGB}{0, 33, 71}
\definecolor{oxgold}{RGB}{190, 158, 90}
\definecolor{lightblue}{RGB}{230, 240, 250}
\definecolor{lightred}{RGB}{248, 238, 240} 
\newcommand{\mat}[1]{\underline{\underline{#1}}}
\newcommand{\abs}[1]{\left|#1\right|}
\newcommand{\dd}{\,\mathrm{d}}
\newcommand{\cchevrons}[1]{\left\langle \!\left \langle #1 \right\rangle \! \right\rangle}
\let\vec\bm

\newtcolorbox{oxbox}[1][]{
  colback=lightblue,
  colframe=oxblue,
  colbacktitle=oxblue,
  coltitle=white,
  fonttitle=\bfseries,
  boxrule=0.6pt,
  arc=3pt,
  left=8pt, right=8pt, top=6pt, bottom=6pt,
  title=#1
}
\newtcolorbox{cambox}[1][]{
  colback=lightred,
  colframe=camred!60!black,
  colbacktitle=camred,
  coltitle=white,
  fonttitle=\bfseries,
  boxrule=0.6pt,
  arc=3pt,
  left=8pt, right=8pt, top=6pt, bottom=6pt,
  title=#1
}
\newtheorem{lemma}{Lemma}
\newtheorem{proposition}{Proposition}
\newtheorem{remark}{Remark}

\usepgfplotslibrary{groupplots}
\pgfplotsset{compat=1.18}

\titlehead{Research}

\definecolor{seaborngreen}{rgb}{0.3333333333333333, 0.6588235294117647, 0.40784313725490196}  
\definecolor{seaborncyan}{rgb}{0.39215686274509803, 0.7098039215686275, 0.803921568627451}  
\definecolor{seabornblue}{rgb}{0.2980392156862745, 0.4470588235294118, 0.6901960784313725}  
\definecolor{seabornpurple}{rgb}{0.5058823529411764, 0.4470588235294118, 0.6980392156862745}  
\definecolor{seabornred}{rgb}{0.7686274509803922, 0.3058823529411765, 0.3215686274509804}  
\definecolor{seabornorange}{rgb}{0.958, 0.476, 0.206}  
\definecolor{seabornsand}{rgb}{0.8, 0.7254901960784313, 0.4549019607843137}  

\begin{document}

\title{Kinetic derivation of thermal viscous models for nematic liquid crystal dynamics}

\author{
P.~E.~Farrell$^{1,2}$, J.~M\'alek$^{2}$, O.~Sou\v{c}ek$^{2}$ and U.~Zerbinati$^{1}$}

\address{$^{1}$University of Oxford, Mathematical Institute, Andrew Wiles Building, Woodstock Road, Oxford OX2 6GG, UK\\
$^{2}$Mathematical Institute, Faculty of Mathematics and Physics, Charles University, Sokolovsk\'a 83, 186 75 Praha 8, Czech Republic
}
\subject{}


\corres{Umberto Zerbinati\\
\email{zerbinati@maths.ox.ac.uk}}


\begin{abstract}
We develop a macroscopic thermodynamic theory of nematic liquid crystals starting from a kinetic theory of ordered fluids with a collision operator of Bhatnagar--Gross--Krook (BGK) type. The kinetic description incorporates mean-field alignment interactions through a Vlasov potential and relies on a separation of time scales, with orientational relaxation occurring on a faster time scale than translational momentum relaxation. At the continuum level, we establish the balance equations for mass, linear and angular momentum, energy, and entropy. Using the zeroth- and first-order Chapman--Enskog expansions, we derive a constitutive equation for the Helmholtz free energy and identify the associated structural form of the entropy production rate. We then exploit additional information from the kinetic description to determine a constitutive relation for the entropy production rate itself. Finally, by applying the constrained maximisation procedure of Rajagopal and Srinivasa, we obtain constitutive equations for the Cauchy stress and couple-stress tensors, as well as for the energy and entropy fluxes. \uz{In this way we generalise the recent inviscid kinetic theory of Farrell, Russo, and Zerbinati to account for viscous, thermal, and spin-diffusive effects, using the simplest BGK-type approximation of the collision operator.} Both compressible and incompressible variants of the theory are presented.
\end{abstract}

\begin{fmtext} In memory of our inspiring mentor and teacher, Professor K.~R.~Rajagopal, \textborn~November 24, 1950 -- \textdied~March 20, 2025. \end{fmtext}

\maketitle
\section{Introduction}
In recent work \cite{farrell}, a compressible \uz{\emph{inviscid}} variant of the Leslie--Ericksen equations has been derived from kinetic considerations. The approach in \cite{farrell} builds upon the kinetic theory of non-spherical particles introduced by Curtiss' school \cite{curtissI, curtissII, curtissIII, curtissIV, curtissV}, employing a novel moment closure inspired by selection procedures in rational thermodynamics to account for anisotropy through the Ericksen stress tensor.

A key limitation of that approach is that the closure is performed around an already aligned state, meaning that collective orientational behaviour is effectively imposed rather than emerging from the dynamics. While this assumption is appropriate in the context of polyatomic gases under strong external fields, it is insufficient for liquid crystals, where the isotropic-to-nematic transition is a fundamental emergent phenomenon. Furthermore, the model derived in \cite{farrell} is strictly inviscid.

In this work, we overcome both limitations to derive a compressible \uz{\emph{viscous}} variant of the Leslie--Ericksen equations. To address the physics of emergent alignment, we take a new starting point: instead of the Curtiss theory, we begin from the kinetic theory of ordered fluids proposed in \cite{carrillo}. This framework naturally incorporates mean-field alignment interactions and provides an entropy structure that enables a consistent description of the transition from isotropic to aligned phases within the kinetic formulation.

To address the macroscopic viscous behaviour, we apply the thermodynamically consistent moment closure recently \cla{explored} for the classical Boltzmann equation in \cite{malek}\footnote{Conceptually, the work of \cla{\cite{malek}} falls within the broader programme initiated by Truesdell and Muncaster \cite{Truesdell}, which aims to recast the kinetic theory of Maxwellian molecules as a branch of rational thermodynamics.}. This hybrid Chapman--Enskog--Rajagopal--Srinivasa procedure combines a Chapman--Enskog expansion---used to determine the structure of entropy production---with the constrained maximisation principle of Rajagopal and Srinivasa \cite{rajagopalSrinivasa}.
By combining this kinetic theory on an order-parameter manifold with the Chapman--Enskog--Rajagopal--Srinivasa closure, we generalise the inviscid theory of \cite{farrell} to include viscous, heat-conducting, and spin-diffusive effects. We also show that augmenting the Rajagopal--Srinivasa maximisation with a Lagrange multiplier enforcing the divergence-free constraint delivers, with no further work, the corresponding \emph{incompressible} model.

Our approach does not recover the full Leslie viscous stress tensor with its five independent viscosities, because we use an isotropic collision operator of Bhatnagar--Gross--Krook (BGK) type to make the calculations tractable. Instead, it yields a viscous stress tensor governed by a single viscosity parameter that does not coincide with any single classical Leslie coefficient. We conjecture that recovering the complete Leslie viscous stress tensor is possible with an anisotropic BGK collision operator involving several relaxation timescales. Investigating this is the subject of ongoing work.


This work is organised as follows. In \Cref{sec:kt} we summarise the kinetic theory of ordered fluids proposed in \cite{carrillo} \uz{and, in \Cref{sec:bgk}, the classical single-relaxation BGK collision operator that approximates the full Boltzmann operator throughout the rest of the paper, preserving its conservation and entropy-production structure}.
In \Cref{sec:hyd} we present the classical derivation of the balance equations for mass, momentum, energy, and a generalised angular momentum, starting from the kinetic theory of ordered fluids.
In \Cref{sec:hthm} we discuss the H-theorem for the kinetic theory of ordered fluids and the mechanism behind the emergence of alignment at sufficiently low temperature under the effect of specific types of interaction potentials. In the same section we also introduce the Helmholtz free energy of the kinetic system and prove the kinetic identity relating the integrated Vlasov power to the time derivative of the mean-field interaction energy. We then use the resulting variational structure to derive the Gibbs equilibrium distribution as a constrained minimiser of the free energy.
In \Cref{sec:ce} we use the Chapman--Enskog expansion for the kinetic theory of ordered fluids to compute the entropy production, and in \cla{\Cref{sec:hybrid}} we combine constrained maximisation of the entropy production with ideas from \cite{farrell} to derive the constitutive relations for a compressible \uz{\emph{viscous}} variant of the Leslie--Ericksen equations\uz{. We recast the angular-momentum balance as an equation for the nematic director, summarise the resulting closed system, and derive its incompressible counterpart by constrained maximisation}.

\section{A kinetic theory of ordered fluids}
\label{sec:kt}
Calamitic fluids consist of rod-like, i.e.~calamitic, molecules that typically interact via a hard potential\footnote{We refer to an interaction as hard when the corresponding collision kernel increases with the relative velocity of the particles. In \cite{malek}, the authors discuss how to generalise the BGK approximation to more general Boltzmann-type kinetic models with hard interaction kernels.}. In his seminal work \cite{onsager}, Onsager showed that a system of molecules interacting via a steric potential can undergo a phase transition from an isotropic to a liquid crystalline state \cla{as the density of the system rises above a critical value}\footnote{Onsager's original work was based on a virial expansion of the free energy. Such expansions are a classical tool in the statistical mechanical treatment of gases and are rigorously justified primarily in dilute regimes, similarly to the Boltzmann--Grad limit adopted in \cite{carrillo}. A modern presentation of Onsager's theory can be found in \cite{virgaOnsager}, where the author also extends the range of validity of the theory.}.

The liquid crystalline state is characterised by the emergence of collective behaviour in the form of molecular alignment, as the system tends to maximise its entropy, much like matches in a box tend to align when shaken. The emergence of alignment suggests that this phase possesses orientational order. Accordingly, we refer to a calamitic fluid in this phase as an ordered fluid. In this section, we give a brief introduction to the kinetic theory of ordered fluids proposed in \cite{carrillo}, which describes fluids whose molecular constituents are endowed with a microstructure---here, orientation---that exhibits collective behaviour, such as alignment.

The kinetic theory of ordered fluids proposed in \cite{carrillo} is based on the hypothesis that the molecular interactions are of weak-order nature, see \cite[Section 3.3]{carrillo}, \uzfix{i.e.~that only the interactions between the molecular order parameters are long-ranged,} and that there is a separation of scales between the relaxation time of the microstructural degrees of freedom and the relaxation time of the translational degrees of freedom, i.e.~the microstructural degrees of freedom relax much faster than the translational ones. A detailed discussion of the validity of these assumptions can be found in \cite[Remark 3.7]{carrillo}. Under these assumptions, the resulting kinetic equation takes the form

\begin{equation}
	\label{eq:cfmz}
	\partial_t f + m^{-1} \vec{p} \cdot \nabla_{\vec{x}} f + \mat{B}(\nu)^{-1}\vec{\varsigma} \cdot \nabla_{\nu} f + \vec{\mathcal{V}}\cdot \nabla_{\vec{\varsigma}} f = \mathcal{Q}(f,f),
\end{equation}
where $f = f(\vec{x}, \vec{p}, \nu, \vec{\varsigma}, t)$ is the distribution function of particles at time $t$, position $\vec{x}$, with linear momentum $\vec{p}$, microstructural configuration $\nu$, and conjugate momentum $\vec{\varsigma}$ associated with the microstructural configuration, and where $\mat{B}(\nu)$ is a generalised inertia tensor that may depend on $\nu$. The operator $\mathcal{Q}$ is a Boltzmann-type collision operator encoding the interactions between particles, while $\vec{\mathcal{V}}$ is a mean-field torque (referred to as the Vlasov torque) generated by an interaction potential $\mathcal{W}$ that encodes the effect of the interactions on the microstructural configuration degrees of freedom and is defined as
\begin{equation}
	\label{eq:potential}
	\vec{\mathcal{V}}(\vec{x}, \nu, t) = -\int_{\mathbb{R}^3} \int_{\mathbb{R}^3} \int_{\mathcal{M}} \int_{\mathbb{R}^3} \nabla_{\nu} \mathcal{W}(\abs{\vec{x} - \vec{x}_*}, \nu, \nu_*) f(\vec{x}_*, \vec{p}_*, \nu_*, \vec{\varsigma}_*, t) \, d\vec{x}_* d\vec{p}_* d\nu_* d\vec{\varsigma}_*,
\end{equation}
where $\mathcal{W}$ is a kernel encoding the strength of the interaction between particles at distance $\abs{\vec{x} - \vec{x}_*}$ and with microstructural configurations $\nu$ and $\nu_*$. We require this kernel to be rotationally and translationally invariant, and \emph{symmetric under particle exchange}, i.e.
\begin{equation}
	\label{eq:Wsym}
	\mathcal{W}(\abs{\vec{x}-\vec{x}_*}, \nu, \nu_*) = \mathcal{W}(\abs{\vec{x}-\vec{x}_*}, \nu_*, \nu),
\end{equation}
for all $\vec{x},\vec{x}_*\in\mathbb{R}^3$ and $\nu,\nu_*\in\mathcal{M}$.

In particular, following Capriz's original work \cite{capriz}, the microstructural configuration $\nu$ is assumed to belong to a Riemannian manifold $\mathcal{M}$, which encodes the nature of the microstructure. Together with the manifold $\mathcal{M}$, we also need to study the Lie group action $\mathcal{A}:\mathrm{SO}(3)\times \mathcal{M} \to \mathcal{M}$, which describes the effect of a rotation on the microstructural configuration. We call the tuple $(\mathcal{M}, \mathcal{A})$ the order-parameter manifold. With these tools in place, rotational invariance of the interaction kernel $\mathcal{W}$ means that
\begin{equation}
	\label{eq:rotinv}
	\mathcal{W}(\abs{\vec{x} - \vec{x}_*}, \nu, \nu_*) = \mathcal{W}(\abs{\mat{Q}(\vec{x} - \vec{x}_*)}, \mathcal{A}(\mat{Q}, \nu), \mathcal{A}(\mat{Q}, \nu_*)),
\end{equation}
for any $\mat{Q} \in \mathrm{SO}(3)$, $\vec{x}, \vec{x}_* \in \mathbb{R}^3$, and $\nu, \nu_* \in \mathcal{M}$.
In \cite[Section 2.2]{carrillo} it is proved that the collision operator $\mathcal{Q}$ is such that
\begin{equation}
	\label{eq:collinv}
	\int_{\mathcal{M}} \int_{\mathbb{R}^3} \int_{\mathbb{R}^3} \mathcal{Q}(f,f) \psi \, d\vec{p} d\nu d\vec{\varsigma} = 0, \; \forall \psi \in \left\{m, \vec{p}, \frac{1}{2}\left(\frac{\abs{\vec{p}}^2}{m} + \vec{\varsigma}^T \mat{B}(\nu)^{-1} \vec{\varsigma}\right), A_{\vec{x}}(\mat{\Omega})\cdot \vec{p} + A_{\nu}(\mat{\Omega})\cdot \vec{\varsigma} \right\},
\end{equation}
where $\mat{\Omega} \in \mathfrak{so}(3)$ is an arbitrary element of the Lie algebra, $A_{\vec{x}}$ is the infinitesimal generator of the canonical group action of $\mathrm{SO}(3)$ on $\mathbb{R}^3$, i.e.\ $\mathcal{A}_{\mathrm{canonical}}(\mat{Q}, \vec{x}) = \mat{Q}\vec{x}$, and $A_\nu$ is the infinitesimal generator of the group action $\mathcal{A}$.
Furthermore, in \cite[Section 3.4]{carrillo} a Boltzmann-type inequality leading to an H-theorem is proved under the assumption that the group action $\mathcal{A}$ is transitive, i.e.\ for any $\nu, \nu_* \in \mathcal{M}$ there exists $\mat{Q} \in \mathrm{SO}(3)$ such that $\mathcal{A}(\mat{Q}, \nu) = \nu_*$.\footnote{In the context of the order-parameter manifold given by Euler angles, the H-theorem in \cite{carrillo} is a generalisation of the H-theorem proved in \cite{curtissI, curtissV}. In \cite{curtissI} it is erroneously stated that the H-theorem holds for any choice of molecular shape; this statement is corrected in \cite{curtissV}, where the authors show that their H-theorem holds only for top-symmetric molecules. By contrast, the H-theorem for the kinetic theory of ordered fluids proposed in \cite{carrillo} holds for any molecular shape.}
This type of result is the focus of the next section, where we discuss the H-theorem for the kinetic theory of ordered fluids and the mechanism behind the emergence of alignment.

We focus on the case of a uniaxial nematic liquid crystal, for which the most elegant order-parameter manifold would be the projective plane $\mathbb{R}\mathrm{P}^2$; see \cite[Examples C]{carrillo}. \cla{For simplicity, however, we use the unit sphere $\mathbb{S}^2$ as the order-parameter manifold. The use of $\mathbb{S}^2$ as an order-parameter manifold is common in the literature on polar fluids \cite{onsager, virgaOnsager}, and the additional head-to-tail symmetry of the nematic phase can be imposed through the interaction kernel $\mathcal{W}$.} We therefore work within the setting of \cite[Examples D]{carrillo}, where the order-parameter manifold is $\mathbb{S}^2$ and the group action $\mathcal{A}$ is the canonical group action of $\mathrm{SO}(3)$ on $\mathbb{S}^2$, i.e.\ $\mathcal{A}_{\mathrm{canonical}}(\mat{Q}, \nu) = \mat{Q}\vec{\nu}$, where $\vec{\nu} \in \mathbb{R}^3$ is the embedding of $\nu \in \mathbb{S}^2$ in $\mathbb{R}^3$.

\subsection{Vanishing-girth limit}
\label{sec:vanishing_girth}
Following \cite[Section 5]{farrell}, we model the molecular constituents as rigid rods of length $\ell$ and vanishing girth $\varepsilon$, so that the inertia tensor $\mathbb{I}(\vec{\nu})$ is given by
\begin{equation}
	\label{eq:inertia}
	\mathbb{I}(\vec{\nu}) = \frac{m \ell^2}{12} \left(\mat{I} - \vec{\nu} \otimes \vec{\nu}\right),
\end{equation}
where $m$ is the mass of the molecule and $\mat{I}$ is the identity matrix in $\mathbb{R}^{3\times 3}$.
Comparing the orientational kinetic energy $\frac{1}{2}\dot{\vec{\nu}}\cdot \mat{B}\dot{\vec{\nu}}$ with the physical rotational kinetic energy of a needle $\frac{1}{2}\vec{\omega}\cdot\mathbb{I}\vec{\omega} = m\frac{\ell^2}{24}\abs{\vec{\omega}_\perp}^2 = m\frac{\ell^2}{24}\abs{\dot{\vec{\nu}}}^2$\footnote{\cla{Notice that $\vec{\omega}\cdot \mathbb{I}\vec{\omega}=\frac{m\ell^2}{12}\left(\abs{\vec{\omega}}^2-(\vec{\omega}\cdot\vec{\nu})^2\right)=\frac{m\ell^2}{12}\abs{\vec{\omega}\times\vec{\nu}}^2$. Assuming that the co-rotational time derivative $\overset{\circ}{\vec{\nu}}=0$, this implies that the total rotational kinetic energy of the molecule is given by $\frac{m\ell^2}{24}\,\dot{\vec{\nu}}\cdot\dot{\vec{\nu}}$.}}, and using $\dot{\vec{\nu}} = \vec{\omega}\times\vec{\nu}$ and $\abs{\dot{\vec{\nu}}}^2 = \abs{\vec{\omega}_\perp}^2$ we identify
\begin{equation}\label{eq:B_calamitic}
	\mat{B} = m\frac{\ell^2}{12} \mat{I},
\end{equation}
where $\mat{I}$ is the identity on the tangent space $T_{\vec{\nu}}\mathbb{S}^2$.
The conjugate momentum is
\begin{equation}\label{eq:varsigma_def}
	\vec{\varsigma} = \frac{\partial \mathcal{L}}{\partial\dot{\vec{\nu}}} = \mat{B}\dot{\vec{\nu}} = m\frac{\ell^2}{12}\dot{\vec{\nu}} \in T_{\vec{\nu}}\mathbb{S}^2,
\end{equation}
where $\mathcal{L}$ is the Lagrangian of the system, and thus $\mat{B}^{-1}\vec{\varsigma} = \left(m\frac{\ell^2}{12}\right)^{-1}\vec{\varsigma} = \dot{\vec{\nu}}$. \cla{Strictly speaking the conjugate momentum $\vec{\varsigma}$ is an element of the cotangent space $T^*_{\vec{\nu}}\mathbb{S}^2$, but throughout this paper we implicitly identify $T^*_{\vec{\nu}}\mathbb{S}^2$ with the tangent space $T_{\vec{\nu}}\mathbb{S}^2$, so that $\vec{\varsigma}$ may be treated as a tangent vector at $\vec{\nu}$.} Note that this vector lies in a two-dimensional space, not in $\mathbb{R}^3$. The orientational contribution to the kinetic energy is
\begin{equation}
	\frac{1}{2}\vec{\varsigma}\cdot \mat{B}^{-1}\vec{\varsigma} = \frac{1}{2}\left(m\frac{\ell^2}{12}\right)^{-1}\abs{\vec{\varsigma}}^2 = m\frac{\ell^2}{24}\abs{\dot{\vec{\nu}}}^2.
\end{equation}

The physical angular velocity $\vec{\omega}\in\mathbb{R}^3$ and the conjugate momentum $\vec{\varsigma}\in T_{\vec{\nu}}\mathbb{S}^2$ are related by $\vec{\varsigma} = m\frac{\ell^2}{12}\dot{\vec{\nu}} = m\frac{\ell^2}{12}(\vec{\omega}\times\vec{\nu})$. Conversely, $\vec{\omega}_\perp = -\vec{\nu}\times(\vec{\nu}\times\vec{\omega}) = \vec{\nu}\times\dot{\vec{\nu}}$, and the component along $\vec{\nu}$, i.e.\ the spin about the molecular axis, does not contribute to the kinetic energy of a needle and is not captured by $\vec{\varsigma}$.
This explains why the inertia tensor \eqref{eq:inertia} has a null direction along $\vec{\nu}$ while the generalised inertia tensor $\mat{B} = m\frac{\ell^2}{12} \mat{I}$ acts in $T_{\vec{\nu}}\mathbb{S}^2$ and thus has no null direction because it acts on a vector space of one lower dimension. In this setting the last collision invariant in \eqref{eq:collinv}, i.e.\ $A_{\vec{x}} \mat{Q}\cdot \vec{p} + A_{\vec{\nu}} \mat{Q}\cdot \vec{\varsigma}$, reads $\vec{\nu}\times\vec{\varsigma} + \vec{x}\times\vec{p}$ and is the classical total angular momentum of a needle. Indeed, we can verify that 
\begin{equation}
	\label{eq:total_angular_momentum}
	\vec{\nu}\times\vec{\varsigma} + \vec{x}\times\vec{p} = \vec{\nu}\times(m\frac{\ell^2}{12}\dot{\vec{\nu}}) + \vec{x}\times(m\vec{v}) = \mathbb{I}\vec{\omega} + \vec{x}\times(m\vec{v}).
\end{equation}
Thus the last collision invariant in \eqref{eq:collinv} corresponds to the conservation of the total angular momentum of a needle:
\begin{equation}
	\int_{\mathcal{M}} \int_{\mathbb{R}^3} \int_{\mathbb{R}^3} \mathcal{Q}(f,f) (\vec{\nu}\times\vec{\varsigma} + \vec{x}\times\vec{p}) \, d\vec{p} d\vec{\nu} d\vec{\varsigma} = \int_{\mathcal{M}} \int_{\mathbb{R}^3} \int_{\mathbb{R}^3} \mathcal{Q}(f,f) (\mathbb{I}\vec{\omega} + \vec{x}\times(m\vec{v})) \, d\vec{p} d\vec{\nu} d\vec{\varsigma} = 0.
\end{equation}
From this point on, we write $C_{m,\ell} = m\frac{\ell^2}{12}$, so that
\begin{equation}
	\mat{B} = C_{m,\ell}\mat{I}, \quad \mat{B}^{-1} = C_{m,\ell}^{-1}\mat{I}.
\end{equation}

\subsection{Collision response model}
We assume that the collision operator $\mathcal{Q}$ represents the effect of the microscopic collisions described by the collision response model studied in \cite[Example 2.32]{carrillo}, which is a generalisation of the classical hard-sphere collision model to the case of rigid rods. Thus the binary collision rule reads
\begin{subequations}
	\label{eq:binary_rule}
\begin{alignat}{2}
	\vec{v}' &= \vec{v} - 2\frac{J}{m} \vec{n},\qquad && \vec{v}'_*  = \vec{v}_* + 2\frac{J}{m} \vec{n}, \\
	\fix{(\vec{\nu}\times\vec{\varsigma}')}  &\fix{= (\vec{\nu}\times\vec{\varsigma}) + 2J \mathbb{I}^{\dagger} (\vec{r}\times \vec{n})}, \qquad &&\fix{(\vec{\nu}_*\times\vec{\varsigma}'_*)  = (\vec{\nu}_*\times\vec{\varsigma}_*) - 2J \mathbb{I}^{\dagger}_* (\vec{r}_*\times \vec{n})},
\end{alignat}
where $\mathbb{I}^\dagger$ denotes the Moore--Penrose pseudo-inverse of the singular inertia tensor, $\vec{r}$ and $\vec{r}_*$ are respectively the vectors connecting the centre of mass of each colliding particle to the point of contact, $\vec{n}$ is the unit normal vector to the plane separating the two colliding particles and passing through the point of contact, and $J$ is the collision impulse given by
\end{subequations}

\begin{equation}
\label{eq:collision_impulse}
J = -\dfrac{(\vec{v}-\vec{v}_*) \cdot \vec{n}}{\frac{2}{m} + \left[ \mathbb{I}^{\dagger} (\vec{r}\times \vec{n})\times \vec{r} + \mathbb{I}^{\dagger}_* (\vec{r}_*\times \vec{n})\times \vec{r}_*\right]\cdot \vec{n}}.
\end{equation}
The above collision rule follows an instantaneous elastic collision model, the limitations of which are briefly discussed in \cite[Remark 3.8]{carrillo}. 

In particular, we consider the collision operator
\begin{equation}
	\label{eq:collision}
	\mathcal{Q}(f,f) = \int_{\mathbb{R}^3}\int_{\mathbb{S}^2}\int_{\mathbb{R}^3}\int_0^{2\pi}\int_0^\pi (f^\prime_* f^\prime-f_*f)(\vec{k}\cdot\vec{g})S(\vec{k}(\theta,\varphi))\,d\theta\,d\varphi\,d\vec{p}_*\,d\vec{\nu}_*\,d\vec{\varsigma}_*,
\end{equation}
where $\vec{k}=\vec{k}(\theta,\varphi)$ is the unit collision direction, $S$ is the Jacobian of the transformation from the excluded volume to the spherical coordinates $(\theta, \varphi)$, and $\vec{g}$ is the relative velocity defined as
\begin{equation}
	\label{eq:relative_velocity}
	\fix{\vec{g} = (\vec{v} - \vec{v}_*) - \frac{1}{C_{m,\ell}}\left((\vec{\varsigma}\cdot \vec{r})\vec{\nu} - (\vec{\varsigma}_*\cdot \vec{r}_*)\vec{\nu}_* - (\vec{\nu}\cdot \vec{r})\vec{\varsigma} + (\vec{\nu}_*\cdot \vec{r}_*)\vec{\varsigma}_*\right)}.
\end{equation}

\uz{\subsection{A simplified BGK relaxation model}\label{sec:bgk}
The collision integral \eqref{eq:collision} is intractable for explicit hydrodynamic calculations. To make the subsequent analysis tractable we replace it throughout the remainder of the paper with the classical single-relaxation BGK operator
\begin{equation}\label{eq:BGK}
	\mathcal{Q}[f] \;=\; -\frac{1}{\tau}\bigl(f - f^{(0)}\bigr),
\end{equation}
where $\tau>0$ is the relaxation time and $f^{(0)}$ is the Maxwellian
\begin{equation}
	\label{eq:maxwellian}
	f^{(0)}(\vec{x},\vec{p},\vec{\nu},\vec{\varsigma},t) = \frac{\rho(\vec{x},t)\, \lambda(\vec{\nu},t)}{m(2\pi m k_B \theta(\vec{x},t))^{3/2}\,(2\pi C_{m,\ell}k_B\theta(\vec{x},t))}\exp\left(-\frac{m\abs{\vec{V}}^2}{2k_B\theta(\vec{x},t)} - \frac{\abs{\vec{\Sigma}}^2}{2C_{m,\ell}k_B\theta(\vec{x},t)}\right),
\end{equation}
where $\vec{v}=\vec{p}/m$, $\vec{V}=\vec{v}-\vec{u}$, \fix{$\vec{\Sigma}=\vec{\varsigma}-C_{m,\ell}(\vec{\Omega}\times\vec{\nu})$ is the peculiar conjugate momentum measured from the rigid co-rotation that carries the bulk intrinsic angular momentum, the bulk angular velocity $\vec{\Omega}(\vec{x},t)$ being fixed by requiring $f^{(0)}$ to share the intrinsic angular momentum $\vec{\mu}=\cchevrons{\vec{\nu}\times\vec{\varsigma}}$ with $f$, which gives $\vec{\mu}=C_{m,\ell}(\mat{I}-\mat{S})\vec{\Omega}$ with $\mat{S}=\int_{\mathcal{M}}\lambda(\vec{\nu},t)\,\vec{\nu}\otimes\vec{\nu}\,d\vec{\nu}$ the orientational order tensor,} $n=\rho/m$, and $d\vec{\Xi}=d\vec{p}\,d\vec{\nu}\,d\vec{\varsigma}$. \fix{Centring the rotational Gaussian at the rigid co-rotation $\vec{\Omega}\times\vec{\nu}$, rather than at a uniform shift of $\vec{\varsigma}$, is what keeps $\log f^{(0)}$ affine in the collision invariants \eqref{eq:collinv}, since the conserved rotational quantity is the intrinsic angular momentum $\vec{\nu}\times\vec{\varsigma}$ and not $\vec{\varsigma}$ itself, and it is the equilibrium form derived in \cite[Section 3.4]{carrillo}. Under the micro--macro identification $\vec{\nu}\equiv\widehat{\vec{\nu}}(\vec{x},t)$ adopted in \Cref{sec:hyd}, the rigid-rod constraint $\vec{\varsigma}\cdot\vec{\nu}=0$ forces this drift to coincide with the uniform shift $\vec{\sigma}=\cchevrons{\vec{\varsigma}}$, so the simplified peculiar momentum $\vec{\Sigma}=\vec{\varsigma}-\vec{\sigma}$ used there is recovered with no change to any later result.} This distribution shares $(\rho,\vec u, e, \vec{\mu}, \lambda)$ with $f$, where $\lambda(\vec{\nu},t)$ is a non-negative function representing the distribution of molecular orientations and $\theta(\vec{x},t)$ is the local temperature. The Gaussian translational and rotational structure of \eqref{eq:maxwellian} is fixed kinematically as the unique distribution annihilating the full collision integral \eqref{eq:collision} on its collision invariants \eqref{eq:collinv}, and is derived self-consistently in \Cref{app:maxwellian_derivation}. The orientational profile $\lambda(\vec{\nu},t)$ remains a free parameter at this level and is determined thermodynamically as a constrained minimiser of the Helmholtz free energy in \Cref{sec:gibbs_variational}.

Since $\log f^{(0)}$ is affine in the collision invariants of \eqref{eq:collinv} and $f,f^{(0)}$ share those moments,
\begin{equation}\label{eq:BGK_conservation}
	\int \mathcal{Q}[f]\,\psi\,d\vec{p}\,d\vec{\nu}\,d\vec{\varsigma} = 0\qquad \text{for every collision invariant } \psi,
\end{equation}
so the mass, momentum, energy, and angular-momentum balance laws derived in the next section are insensitive to the replacement.

\begin{lemma}\label{lemma:BGK_entropy}
	Let $f$ be a non-negative distribution function and let $\xi$ be the entropy production associated with the operator \eqref{eq:BGK}. Then
	\begin{equation}\label{eq:BGK_entropy}
		\xi \;=\; -k_B\!\int \mathcal{Q}[f]\,\log(f)\,d\vec{\Xi} \;\geq\; 0,
	\end{equation}
	with equality if and only if $f = f^{(0)}$.
\end{lemma}
\begin{proof}
	Since $\log f^{(0)}$ is affine in the collision invariants and $f,f^{(0)}$ share those moments, $\int(f-f^{(0)})\log f^{(0)}\,d\vec{\Xi} = 0$, so $\xi = (k_B/\tau)\int(f-f^{(0)})\log(f/f^{(0)})\,d\vec{\Xi}$. Convexity of $x\log x$ then yields $\xi\geq 0$, with equality iff $f = f^{(0)}$.
\end{proof}

\uz{An anisotropic extension of \eqref{eq:BGK} with director-aligned translational relaxation rates and a translational--rotational cross-coupling, generating the complete Leslie--Ericksen viscous response and Parodi's relation as Onsager reciprocity, is the subject of ongoing work.}}

\section{Balance laws}
\label{sec:hyd}
Multiplying \eqref{eq:cfmz} by a collision invariant $\psi$ and integrating with respect to $\vec{\Xi} = (\vec{p}, \vec{\nu}, \vec{\varsigma})$ we obtain
\begin{equation}
	\label{eq:pre_enskog_moment_equation}
	\int\psi\,\partial_t f\dd\vec{\Xi} + \int\psi\,\frac{\vec{p}}{m}\cdot\nabla_{\vec{x}} f\dd\vec{\Xi} + \int\psi\,\frac{\vec{\varsigma}}{C_{m,\ell}}\cdot\nabla_{\vec{\nu}} f\dd\vec{\Xi} + \int\psi\,\vec{\mathcal{V}}\cdot\nabla_{\vec{\varsigma}} f\dd\vec{\Xi} = 0.
\end{equation}
Following the calculations presented in \cite[Section 3]{malek} for the classical Boltzmann equation and in \cite[Section 2]{farrell} for the Curtiss kinetic theory, we obtain
\begin{equation}
	\label{eq:enskog_moment_equation}
	\partial_t(\rho \cchevrons{\psi}) + \nabla_{\vec{x}}\cdot\bigl(\rho\cchevrons{\vec{v}\psi}\bigr) - \rho \cchevrons{\dot{\vec{\nu}}\cdot\nabla_{\vec{\nu}}\psi} - \rho\cchevrons{\vec{\mathcal{V}}\cdot\nabla_{\vec{\varsigma}}\psi} = 0,
\end{equation}
where $\rho$ is the mass density and $\cchevrons{\psi}$ is the expectation of $\psi$ with respect to the distribution function $f$, defined by
\begin{equation}
	\rho(\vec{x},t) = mn(\vec{x},t) = m\int f \, d\vec{p} d\vec{\nu} d\vec{\varsigma}, \qquad
	\cchevrons{\psi}(\vec{x},t) = \frac{1}{n}\int \psi f \, d\vec{\Xi}.
\end{equation}
The only difference from the derivations in \cite[Section 3]{malek} and \cite[Section 2]{farrell} is the additional Vlasov term; since $\vec{\mathcal{V}}$ is independent of $\vec{\varsigma}$ by \eqref{eq:potential}, this term can be simplified by integration by parts in $\vec{\varsigma}$, as for all other terms.

Furthermore, since the collision invariant $\psi = m$ is independent of $\vec{\nu}$ and $\vec{\varsigma}$, the same calculations as in \cite[Section 3]{malek} give the continuity equation
\begin{equation}\label{eq:continuity_eq}
	\partial_t\rho + \nabla_{\vec{x}}\cdot(\rho\vec{u}) = 0,
\end{equation}
where $\vec{u} = \cchevrons{\vec{v}}$ is the system bulk velocity.
Since the collision invariant $\psi_2 = \vec{p}$ does not depend on $\vec{\nu}$ and $\vec{\varsigma}$, the same calculations as in \cite[Section 3]{malek} give the linear-momentum balance,
\begin{equation}
	\label{eq:linear_momentum_balance_law}
	\rho\left[\partial_t\vec{u} + (\nabla_{\vec{x}}\vec{u})\vec{u}\right] - \nabla_{\vec{x}}\cdot\mathbb{T} = 0,
\end{equation}
where $\mathbb{T}$ is the Cauchy stress tensor defined in kinetic theory as
\begin{equation}
	\label{eq:stress_tensor}
	\mathbb{T} = -\rho \cchevrons{\vec{V}\otimes\vec{V}}, \qquad \vec{V} = \vec{v} - \vec{u}.
\end{equation}

Since both $\vec{p}$ and $\vec{\nu}\times\vec{\varsigma} + \vec{x}\times\vec{p}$ are collision invariants, $\vec{\nu}\times\vec{\varsigma}$ is also annihilated by the collision operator. We may therefore substitute $\psi = \vec{\nu}\times\vec{\varsigma}$ in \eqref{eq:enskog_moment_equation} to obtain a balance law for the bulk intrinsic angular momentum $\vec{\mu}$, defined by
\begin{equation}
	\label{eq:eta_def}
	\vec{\mu} = \cchevrons{\vec{\nu}\times\vec{\varsigma}} = \cchevrons{\mathbb{I}\vec{\omega}}.
\end{equation}
We also define the bulk conjugate momentum to $\vec{\nu}$ and its peculiar part by
\begin{equation}
	\vec{\sigma} = \cchevrons{\vec{\varsigma}}, \qquad \vec{\Sigma}= \vec{\varsigma} - \vec{\sigma}.
\end{equation}
\fix{Throughout the derivation of the angular-momentum and energy balance laws we adopt the \emph{micro--macro identification} of the orientation variable: every molecule carries the orientation of the local director field, $\vec{\nu} = \widehat{\vec{\nu}}(\vec{x},t)$, so that any function of $\vec{\nu}$ alone factors out of the averages $\cchevrons{\cdot}$. This is the same kinematic hypothesis under which the strongly aligned closure of the later sections operates, see \Cref{lemma:erot_oseen_frank}.}

With this notation, \eqref{eq:enskog_moment_equation} with $\psi = \vec{\nu}\times\vec{\varsigma}$ can be expressed in terms of the peculiar velocity $\vec{V}$ and the peculiar conjugate momentum $\vec{\Sigma}$ as
\begin{equation}
	\label{eq:angular_momentum_balance_law_intermediate}
	\partial_t(\rho{\vec{\mu}}) + \nabla_{\vec{x}}\cdot(\rho\vec{u}\otimes\vec{\mu}) + \nabla_{\vec{x}}\cdot(\rho\cchevrons{\vec{V}\otimes(\vec{\nu}\times\vec{\Sigma})}) - \rho\cchevrons{\dot{\vec{\nu}}\cdot\nabla_{\vec{\nu}}(\vec{\nu}\times\vec{\varsigma})} \fix{-} \rho\cchevrons{\vec{\mathcal{V}}\cdot\nabla_{\vec{\varsigma}}(\vec{\nu}\times\vec{\varsigma})} = 0,
\end{equation}
where we used $\cchevrons{\vec{V}} = 0$ \fix{together with $\cchevrons{\vec{V}\otimes(\vec{\nu}\times\vec{\sigma})} = \cchevrons{\vec{V}}\otimes(\widehat{\vec{\nu}}\times\vec{\sigma}) = \vec{0}$, with the orientation factoring out of the average under the micro--macro identification}.
To simplify the previous expression we define the \emph{couple stress tensor} as
\begin{equation}\label{eq:couple_stress}
	\mathbb{M} = -\rho \cchevrons{\vec{V}\otimes(\vec{\nu}\times\vec{\Sigma})},
\end{equation}
and rewrite \eqref{eq:angular_momentum_balance_law_intermediate} as
\begin{equation}
	\label{eq:angular_momentum_balance_law_almost_final}
	\partial_t(\rho{\vec{\mu}}) + \nabla_{\vec{x}}\cdot(\rho\vec{u}\otimes\vec{\mu}) - \nabla_{\vec{x}}\cdot\mathbb{M} - \rho\cchevrons{\dot{\vec{\nu}}\cdot\nabla_{\vec{\nu}}(\vec{\nu}\times\vec{\varsigma})} \fix{-} \rho\cchevrons{\vec{\mathcal{V}}\cdot\nabla_{\vec{\varsigma}}(\vec{\nu}\times\vec{\varsigma})} = 0.
\end{equation}
{
To make the following calculations easier to follow, we overload the cross-product notation when it acts between a vector and a tensor. In tensor notation, $\vec{a}\times\mat{A}$ denotes the tensor whose $(i,j)$-th component is $\varepsilon_{ikl}a_k A_{lj}$, where $\varepsilon_{ikl}$ is the Levi--Civita symbol and we use the Einstein summation convention. Symmetrically, $\mat{A}\times\vec{a}$ denotes the tensor whose $(i,j)$-th component is $\varepsilon_{ikl}A_{kj}a_l$. Both conventions act column by column, crossing $\vec{a}$ into each column of $\mat{A}$ on the appropriate side, so that the Leibniz rule $\nabla_{\vec{\nu}}(\vec{\nu}\times\vec{\varsigma}) = (\nabla_{\vec{\nu}}\vec{\nu})\times\vec{\varsigma} + \vec{\nu}\times(\nabla_{\vec{\nu}}\vec{\varsigma})$ holds componentwise.
With this notation, $\vec{\nu}\times\vec{\varsigma}\in\mathbb{R}^3$ can be differentiated in the ambient space. Explicitly, $\nabla_{\vec{\nu}}(\vec{\nu}\times\vec{\varsigma}) = (\nabla_{\vec{\nu}}\vec{\nu})\times\vec{\varsigma} + \vec{\nu}\times(\nabla_{\vec{\nu}}\vec{\varsigma})$. The first term gives $\dot{\vec{\nu}}\cdot[(\nabla_{\vec{\nu}}\vec{\nu})\times\vec{\varsigma}]$. Although $\nabla_{\vec{\nu}}\vec{\nu} = \mat{I} - \vec{\nu}\otimes\vec{\nu}$ is the projection onto the tangent space $T_{\vec{\nu}}\mathbb{S}^2$ rather than the full identity, the fact that $\dot{\vec{\nu}}\cdot\vec{\nu} = 0$ means it acts as the identity in this contraction, reducing the term to $\dot{\vec{\nu}}\times\vec{\varsigma}$.}

The second term vanishes because in the Hamiltonian setting $\vec{\varsigma}$ is independent of $\vec{\nu}$, so $\nabla_{\vec{\nu}}\vec{\varsigma} = 0$. Thus we can further simplify \eqref{eq:angular_momentum_balance_law_almost_final} as
\begin{equation}
	\label{eq:angular_momentum_balance_law_final}
	\partial_t(\rho{\vec{\mu}}) + \nabla_{\vec{x}}\cdot(\rho\vec{u}\otimes\vec{\mu}) - \nabla_{\vec{x}}\cdot \mathbb{M} \fix{-} \rho\cchevrons{\vec{\mathcal{V}}\cdot\nabla_{\vec{\varsigma}}(\vec{\nu}\times\vec{\varsigma})} = 0,
\end{equation}
where the term $\rho\cchevrons{\dot{\vec{\nu}}\times\vec{\varsigma}}$ has been dropped because $\dot{\vec{\nu}} = \mat{B}^{-1}\vec{\varsigma}= C_{m,\ell}^{-1}\;\vec{\varsigma}$, and hence $\dot{\vec{\nu}}\times\vec{\varsigma} = C_{m,\ell}^{-1}\;\vec{\varsigma}\times\vec{\varsigma} = 0$.
Similarly, since $\vec{\nu}$ is independent of $\vec{\varsigma}$ and $\nabla_{\vec{\varsigma}}\vec{\varsigma} = \mat{I} - \vec{\nu}\otimes\vec{\nu}$ on $T_{\vec{\nu}}\mathbb{S}^2$, we have
\begin{equation}
  \nabla_{\vec{\varsigma}}(\vec{\nu}\times\vec{\varsigma}) = \vec{\nu}\times(\nabla_{\vec{\varsigma}}\vec{\varsigma}) = \vec{\nu}\times(\mat{I} - \vec{\nu}\otimes\vec{\nu}) = \vec{\nu}\times\mat{I}.
\end{equation}
Therefore $\vec{\mathcal{V}}\cdot\nabla_{\vec{\varsigma}}(\vec{\nu}\times\vec{\varsigma}) = \vec{\nu}\times\vec{\mathcal{V}}$, and hence
\begin{equation}\label{eq:Vlasov_torque}
	\rho\cchevrons{\vec{\mathcal{V}}\cdot\nabla_{\vec{\varsigma}}(\vec{\nu}\times\vec{\varsigma})} = \rho\cchevrons{\vec{\nu}\times\vec{\mathcal{V}}}.
\end{equation}
Using \eqref{eq:continuity_eq}, we obtain the following balance law for the bulk intrinsic angular momentum $\vec{\mu}$:
\begin{equation}\label{eq:balance_law_angular_momentum}
	\rho\bigl[\partial_t\vec{\mu} + (\nabla_{\vec{x}}\vec{\mu})\vec{u}\bigr] \fix{-} \nabla_{\vec{x}}\cdot \mathbb{M} = \rho\cchevrons{\vec{\nu}\times\vec{\mathcal{V}}}.
\end{equation}
\fix{The sign of the flux term parallels the linear-momentum balance \eqref{eq:linear_momentum_balance_law}, as it must given the parallel definitions $\mathbb{T} = -\rho\cchevrons{\vec{V}\otimes\vec{V}}$ and $\mathbb{M} = -\rho\cchevrons{\vec{V}\otimes(\vec{\nu}\times\vec{\Sigma})}$.}

\fix{\begin{remark}[The peculiar intrinsic angular momentum]\label{rmk:peculiar_angular_momentum}
	Without the micro--macro identification, the peculiar field conjugate to $\vec{\mu}$ is
	\begin{equation}
		\label{eq:peculiar_angular_momentum}
		\vec{\Lambda} := \vec{\nu}\times\vec{\varsigma} - \vec{\mu},
	\end{equation}
	which differs from $\vec{\nu}\times\vec{\Sigma}$ by $\vec{\nu}\times\vec{\sigma} - \vec{\mu}$, and the exact couple-stress flux is $\cchevrons{\vec{V}\otimes\vec{\Lambda}} = \cchevrons{\vec{V}\otimes(\vec{\nu}\times\vec{\varsigma})}$. Since $\vec{\nu}\times\vec{\varsigma} = \mathbb{I}(\vec{\nu})\vec{\omega}$ depends on the orientation $\vec{\nu}$, which is correlated with the peculiar velocity $\vec{V}$, the orientation does not factor out of this average and $\vec{\Lambda}$, rather than $\vec{\nu}\times\vec{\Sigma}$, is the exact peculiar field. Throughout this paper, however, we work around the aligned state and adopt the micro--macro identification $\vec{\nu} = \widehat{\vec{\nu}}(\vec{x},t)$, so that every statement of this section may be written indifferently in either set of variables. The full theory, in which $\vec{\nu}$ remains distributed and $\vec{\Lambda}$ is the correct peculiar field, is briefly derived in \Cref{app:energy_balance_xi}.
\end{remark}}

\begin{remark}
\label{rmk:vlasov_torque}
The main difference between \eqref{eq:balance_law_angular_momentum} and \cite[eq.~2.22c]{farrell} is the presence of the body torque term $\rho\cchevrons{\vec{\nu}\times\vec{\mathcal{V}}}$ on the right-hand side\footnote{A similar term appears in Curtiss' original work; see \cite[eq.~2.51]{curtissI}, where external torques are considered, unlike in \cite{farrell}, where only internal torques are considered.}, which is the mean-field torque generated by the Vlasov torque $\vec{\mathcal{V}}$. In particular, the Vlasov torque $\vec{\mathcal{V}}$ generates a mean-field torque perpendicular to both the molecular orientation $\vec{\nu}$ and the Vlasov torque itself, suggesting that the molecular constituents rotate their molecular axis until either $\vec{\nu}$ and $\vec{\mathcal{V}}$ are parallel or $\vec{\mathcal{V}}$ vanishes. As discussed in greater detail in \Cref{sec:hthm}, this is the mechanism responsible for the emergence of alignment in the system, which is a key feature of liquid crystals and other complex fluids with orientational order.
\end{remark}

Substituting $\psi = \frac{1}{2m}|\vec{p}|^2 + \frac{1}{2C_{m,\ell}}|\vec{\varsigma}|^2$ into \eqref{eq:enskog_moment_equation}, we obtain
\begin{align}
	\label{eq:energy_balance_law_pre}
	\partial_t\left(\rho\cchevrons{\frac{1}{2m}|\vec{p}|^2 + \frac{1}{2C_{m,\ell}}|\vec{\varsigma}|^2}\right) &+ \nabla_{\vec{x}}\cdot\left(\rho\cchevrons{\vec{v}\left(\frac{1}{2m}|\vec{p}|^2 + \frac{1}{2C_{m,\ell}}|\vec{\varsigma}|^2\right)}\right) \\
	&-\;\;\; \rho\cchevrons{\dot{\vec{\nu}}\cdot\nabla_{\vec{\nu}}\left(\frac{1}{2m}|\vec{p}|^2 + \frac{1}{2C_{m,\ell}}|\vec{\varsigma}|^2\right)} \\
	&-\;\;\; \rho\cchevrons{\vec{\mathcal{V}}\cdot\nabla_{\vec{\varsigma}}\left(\frac{1}{2m}|\vec{p}|^2 + \frac{1}{2C_{m,\ell}}|\vec{\varsigma}|^2\right)} = 0.
\end{align}
\fix{We divide \eqref{eq:energy_balance_law_pre} by $m$ and work with specific (i.e.\ per unit mass) energies throughout, so that the translational and rotational components are measured in the same units.} Following the same calculations presented in \cite[Section 3]{malek} to treat the translational degrees of freedom, we can simplify the above expression as
\begin{subequations}
	\begin{align}
		\label{eq:energy_balance_law_intermediate}
		&\rho\bigl[\partial_t e_{\mathrm{tr}} + \vec{u}\cdot\nabla_{\vec{x}} e_{\mathrm{tr}}\bigr] - \mathbb{T}:\nabla_{\vec{x}}\vec{u} + \nabla_{\vec{x}}\cdot\rho\cchevrons{\frac{1}{2}\abs{\vec{V}}^2\vec{V}} \\
		&+\partial_t \left(\rho\cchevrons{\frac{1}{2\fix{m}C_{m,\ell}}|\vec{\varsigma}|^2}\right) + \nabla_{\vec{x}}\cdot\left(\rho\cchevrons{\vec{v}\frac{1}{2\fix{m}C_{m,\ell}}|\vec{\varsigma}|^2}\right) \label{eq:energy_balance_law_intermediate_ang}\\
		&- \rho\cchevrons{\dot{\vec{\nu}}\cdot\nabla_{\vec{\nu}}\left(\frac{1}{2\fix{m}C_{m,\ell}}|\vec{\varsigma}|^2\right)} - \rho\cchevrons{\vec{\mathcal{V}}\cdot\nabla_{\vec{\varsigma}}\left(\frac{1}{2\fix{m}C_{m,\ell}}|\vec{\varsigma}|^2\right)} = 0,
	\end{align}
\end{subequations}
where the translational internal energy is defined as
\begin{equation}
	e_{\mathrm{tr}} = \frac{1}{2}\cchevrons{|\vec{V}|^2} = \frac{1}{2}\cchevrons{\abs{\vec{v}-\vec{u}}^2}.
\end{equation}

\fix{The orientational degrees of freedom require more care than in \cite[Section 3]{farrell}: the balance \eqref{eq:balance_law_angular_momentum} carries the mean-field torque on its right-hand side, while its counterpart \cite[eq.~2.22c]{farrell} is source-free, so the work exerted by the torque must be tracked through the reduction. The complete computation, together with a term-by-term correspondence with \cite{farrell}, is reported in \Cref{app:energy_balance_xi}. Since $\vec{\varsigma}\perp\vec{\nu}$ we have the pointwise identity $\abs{\vec{\nu}\times\vec{\varsigma}}^2 = \abs{\vec{\varsigma}}^2$, and under the micro--macro identification $\abs{\vec{\mu}} = \abs{\widehat{\vec{\nu}}\times\vec{\sigma}} = \abs{\vec{\sigma}}$, so the rotational kinetic energy splits as
\begin{equation}
	\label{eq:rot_energy_split}
	\frac{1}{2mC_{m,\ell}}\cchevrons{|\vec{\varsigma}|^2}
	= \frac{1}{2mC_{m,\ell}}\abs{\vec{\mu}}^2 + e_{\mathrm{rot}},
	\qquad
	e_{\mathrm{rot}} := \frac{1}{2mC_{m,\ell}}\cchevrons{\abs{\vec{\Sigma}}^2} = \frac{C_{m,\ell}}{2m}\cchevrons{\abs{\dot{\vec{\nu}} - \cchevrons{\dot{\vec{\nu}}}}^2}.
\end{equation}
The bulk part $\abs{\vec{\mu}}^2/(2mC_{m,\ell})$ is a kinetic, not an internal, energy and its budget is obtained by dotting the angular-momentum balance \eqref{eq:balance_law_angular_momentum} with $\vec{\mu}/(mC_{m,\ell})$,
\begin{equation}
	\label{eq:bulk_spin_balance}
	\rho\left[\partial_t \frac{\abs{\vec{\mu}}^2}{2mC_{m,\ell}} + \vec{u}\cdot\nabla_{\vec{x}}\frac{\abs{\vec{\mu}}^2}{2mC_{m,\ell}}\right]
	- \frac{1}{mC_{m,\ell}}\nabla_{\vec{x}}\cdot(\mathbb{M}\vec{\mu})
	+ \frac{1}{mC_{m,\ell}}\mathbb{M}:\nabla_{\vec{x}}\vec{\mu}
	= \frac{\rho}{mC_{m,\ell}}\,\vec{\mu}\cdot\cchevrons{\vec{\nu}\times\vec{\mathcal{V}}},
\end{equation}
where $(\mathbb{M}\vec{\mu})_i = \mathbb{M}_{ij}\mu_j$ and $\mathbb{M}:\nabla_{\vec{x}}\vec{\mu} = \mathbb{M}_{ij}\partial_i\mu_j$. The source on the right-hand side is the work exerted by the mean-field torque on the bulk intrinsic angular momentum.}

It remains to treat the Vlasov term. Since $\nabla_{\vec{\varsigma}}\left(\frac{1}{2C_{m,\ell}}|\vec{\varsigma}|^2\right) = \frac{1}{C_{m,\ell}}\vec{\varsigma} = \dot{\vec{\nu}}$, \fix{since $\vec{\varsigma}\perp\vec{\nu}$ implies the triple-product identity $\vec{\mathcal{V}}\cdot\vec{\varsigma} = (\vec{\nu}\times\vec{\varsigma})\cdot(\vec{\nu}\times\vec{\mathcal{V}})$, and since the orientation factors out of the average,} we have
\begin{equation}
	\label{eq:Vlasov_power}
	\frac{\rho}{\fix{m}}\cchevrons{\vec{\mathcal{V}}\cdot\dot{\vec{\nu}}} \fix{= \frac{\rho}{mC_{m,\ell}}\cchevrons{(\vec{\nu}\times\vec{\varsigma})\cdot(\vec{\nu}\times\vec{\mathcal{V}})}
	= \frac{\rho}{mC_{m,\ell}}\,\vec{\mu}\cdot\cchevrons{\vec{\nu}\times\vec{\mathcal{V}}}}.
\end{equation}
\fix{The full Vlasov power therefore coincides with the source of \eqref{eq:bulk_spin_balance}: the mean field works on the bulk intrinsic angular momentum, not on the peculiar fluctuations. Subtracting \eqref{eq:bulk_spin_balance} from the rotational part of \eqref{eq:energy_balance_law_intermediate} removes the Vlasov power entirely and yields the source-free balance of the rotational internal energy,
\begin{equation}
	\label{eq:energy_balance_law_final}
	\rho\bigl[\partial_t e_{\mathrm{rot}} + \vec{u}\cdot\nabla_{\vec{x}} e_{\mathrm{rot}}\bigr] - \frac{1}{mC_{m,\ell}}\mathbb{M}:\nabla_{\vec{x}}\vec{\mu} + \nabla_{\vec{x}}\cdot \rho\cchevrons{\vec{V} \frac{1}{2mC_{m,\ell}}|\vec{\Sigma}|^2} = 0.
\end{equation}}
Substituting \eqref{eq:Vlasov_power} and \eqref{eq:energy_balance_law_final} into \eqref{eq:energy_balance_law_intermediate}, we obtain the energy balance law
\begin{equation}
	\label{eq:energy_balance_law}
	\rho\bigl[\partial_t \widetilde{e} + \vec{u}\cdot\nabla_{\vec{x}} \widetilde{e}\bigr] - \mathbb{T}:\nabla_{\vec{x}}\vec{u} - \fix{\frac{1}{mC_{m,\ell}}}\mathbb{M}:\nabla_{\vec{x}}\vec{\mu} + \nabla_{\vec{x}}\cdot\vec{Q} = \fix{0},
\end{equation}
where the kinetic internal energy is the sum of its translational and rotational components,
\begin{equation}
	\widetilde{e} = e_{\mathrm{tr}} + e_{\mathrm{rot}} = \frac{1}{2}\cchevrons{|\vec{V}|^2} + \frac{1}{2\fix{m}C_{m,\ell}} \cchevrons{\abs{\vec{\Sigma}}^2},
\end{equation}
and we have defined the heat flux vector as
\begin{equation}
	\vec{Q} = \rho\cchevrons{\vec{V} \frac{1}{2}|\vec{V}|^2} + \rho\cchevrons{\vec{V} \frac{1}{2\fix{m}C_{m,\ell}}|\vec{\Sigma}|^2}.
\end{equation}

\fix{\begin{remark}
	\label{rmk:mu_source_in_energy}
	Contrary to the angular-velocity setting of \cite{farrell}, the balance \eqref{eq:balance_law_angular_momentum} carries the source $\rho\cchevrons{\vec{\nu}\times\vec{\mathcal{V}}}$. Reducing the rotational kinetic energy to its internal part multiplies \eqref{eq:balance_law_angular_momentum} by $\vec{\mu}/(mC_{m,\ell})$, and the resulting power $\rho\,\vec{\mu}\cdot\cchevrons{\vec{\nu}\times\vec{\mathcal{V}}}/(mC_{m,\ell})$ acts as a source for the kinetic energy $\abs{\vec{\mu}}^2/(2mC_{m,\ell})$ of the bulk intrinsic angular momentum. By \eqref{eq:Vlasov_power} this work coincides with the full Vlasov power, which is why no Vlasov source survives in \eqref{eq:energy_balance_law}: the mean field exchanges energy with the bulk intrinsic angular momentum, not with the internal energy. Without the micro--macro identification a peculiar remainder $\rho\cchevrons{\vec{\Lambda}\cdot(\vec{\nu}\times\vec{\mathcal{V}})}/(mC_{m,\ell})$ sources the internal energy, see \Cref{rmk:peculiar_angular_momentum} and \Cref{app:energy_balance_xi}, and it vanishes at every aligned Gibbs equilibrium, where $\vec{\mathcal{V}}\parallel\widehat{\vec{\nu}}$.
\end{remark}}

It remains to compute the balance law for the entropy density. This case differs from the previous ones because $\psi = -k_B \log(f) - 1$ is not a collision invariant\footnote{It does not make sense to evaluate the logarithm of a quantity with units. Thus, every time we write $\log(f)$, it should be understood as $\log(f/f_{\mathrm{ref}})$, where the reference is used only to nondimensionalise the argument of the logarithm. It is often natural to choose $f_{\mathrm{ref}}$ to be the Maxwellian distribution; for this choice, the notion of entropy is equivalent to relative entropy.}. Hence the right-hand side of \eqref{eq:enskog_moment_equation} does not vanish, and we have
\begin{equation}
	\label{eq:entropy_balance_law_pre}
	\partial_t(\rho \psi) + \nabla_{\vec{x}}\cdot\bigl(\rho\cchevrons{\vec{v}\psi}\bigr) - \rho \cchevrons{\dot{\vec{\nu}}\cdot\nabla_{\vec{\nu}}\psi} - \rho\cchevrons{\vec{\mathcal{V}}\cdot\nabla_{\vec{\varsigma}}\psi} = - k_B\int \mathcal{Q}(f,f)\log(f)\,d\vec{\Xi}.
\end{equation}
Following the same calculations presented in \cite[Section 3]{malek} to treat the translational degrees of freedom, and defining the macroscopic entropy density $\eta$ by
\begin{equation}
	\rho\eta = -k_B\int f \log(f)\,d\vec{\Xi},
\end{equation}
we obtain, via the chain rule, the evolution equation for the entropy density
\begin{subequations}
\begin{align}
	\label{eq:entropy_balance_law_intermediate}
	\rho \bigl[\partial_t \eta + \vec{u}\cdot\nabla_{\vec{x}} \eta\bigr] - \nabla_{\vec{x}}\cdot\frac{\rho k_B}{m}\cchevrons{\vec{V}\log(f)} - \frac{\rho k_B}{m} \cchevrons{\dot{\vec{\nu}}\cdot\nabla_{\vec{\nu}}\log(f)} &- \frac{\rho k_B}{m}\cchevrons{\vec{\mathcal{V}}\cdot\nabla_{\vec{\varsigma}}\log(f)} \\
																												&= -k_B\int \mathcal{Q}(f,f)\log(f)\,d\vec{\Xi}.
\end{align}
\end{subequations}
Since we are in a Lagrangian setting, $\dot{\vec{\nu}}$ is independent of $\vec{\nu}$. Furthermore, $\vec{\mathcal{V}}$ is independent of $\vec{\varsigma}$ by definition, see \eqref{eq:potential}. We can therefore simplify the previous expression by integrating by parts in $\vec{\nu}$ and $\vec{\varsigma}$ to obtain
\begin{equation}
	\label{eq:entropy_balance_law}
	\rho \bigl[\partial_t \eta + \vec{u}\cdot\nabla_{\vec{x}} \eta\bigr] + \nabla_{\vec{x}}\cdot \vec{\Phi} = \xi,
\end{equation}
where the entropy flux vector and the entropy production are defined, respectively, by
\begin{equation}
	\label{defn:entropy_flux_and_production}
	\vec{\Phi} = -\frac{\rho k_B}{m}\cchevrons{\vec{V}\log(f)}, \qquad \xi = - k_B\int \mathcal{Q}(f,f)\log(f)\,d\vec{\Xi}.
\end{equation}

\begin{lemma}[Vlasov power conservation]
	\label{lemma:vlasov_power}
	Let $f$ be a solution of \eqref{eq:cfmz} that decays, together with all its $\vec{x}$-derivatives, sufficiently fast at infinity. \uz{Assume that the kernel $\mathcal{W}$ is translation-invariant in space, $\mathcal{W}(\vec{x},\vec{\nu},\vec{\nu}_*,\vec{x}_*) = \widehat{\mathcal{W}}(\abs{\vec{x}-\vec{x}_*},\vec{\nu},\vec{\nu}_*)$, and additionally $\vec{x}$-independent, so that $\widehat{\mathcal{W}}(\abs{\vec{x}-\vec{x}_*},\vec{\nu},\vec{\nu}_*) = \widehat{\mathcal{W}}(\vec{\nu},\vec{\nu}_*)$;} assume further that it satisfies the exchange symmetry \eqref{eq:Wsym}, and that the binary collision rule \eqref{eq:binary_rule} preserves the orientations $\vec{\nu}$, $\vec{\nu}_*$ of the colliding particles during the collision\footnote{The orientation-preserving hypothesis is verified by the collision response model \eqref{eq:binary_rule} since only the linear momenta $\vec{p}$, $\vec{p}_*$ and the conjugate orientational momenta $\vec{\varsigma}$, $\vec{\varsigma}_*$ are exchanged in a binary encounter, while $\vec{\nu}$ and $\vec{\nu}_*$ enter the rule only through the geometric factor $\vec{r}\times\vec{n}$. The $\vec{x}$-independence of $\mathcal{W}$ is the regime in which every explicit calculation of the manuscript is performed; we will later use for example \cla{attractive} Kuramoto-type kernels of the form $\cla{\mathcal{W}_K}(\vec{\nu},\vec{\nu}_*) = \cla{-}\,\vec{\nu}\cdot\vec{\nu}_*$.}. Define
	\begin{equation}
		\label{eq:Uint}
		\mathcal{U}_{\mathrm{int}}(t) = \frac{\fix{1}}{2}\int\!\!\int \widehat{\mathcal{W}}(\vec{\nu},\vec{\nu}_*)\, f(\vec{x},\vec{p},\vec{\nu},\vec{\varsigma},t)\, f(\vec{x}_*,\vec{p}_*,\vec{\nu}_*,\vec{\varsigma}_*,t)\, d\vec{\Xi}\, d\vec{\Xi}_*\, d\vec{x}\, d\vec{x}_*,
	\end{equation}
	where $d\vec{\Xi}=d\vec{p}\, d\vec{\nu}\, d\vec{\varsigma}$ and $\rho\cchevrons{A} = m\!\int\! A\, f\, d\vec{\Xi}$. Then
	\begin{equation}
		\label{eq:vlasov_power_conservation}
		\int \frac{\rho}{\fix{m}}\cchevrons{\vec{\mathcal{V}}\cdot\dot{\vec{\nu}}}\, d\vec{x} \;+\; \frac{d}{dt}\mathcal{U}_{\mathrm{int}}(t) \;=\; 0.
	\end{equation}
\end{lemma}
\begin{proof}
We use the shorthand $f := f(\vec{x},\vec{p},\vec{\nu},\vec{\varsigma},t)$ and $f_* := f(\vec{x}_*,\vec{p}_*,\vec{\nu}_*,\vec{\varsigma}_*,t)$, and write all integrals over the full phase space in each of the unstarred and starred variables, with $d\vec{\Xi}=d\vec{p}\,d\vec{\nu}\,d\vec{\varsigma}$.

Differentiating \eqref{eq:Uint} in time produces two summands by the product rule:
\begin{equation}
	\label{eq:vp_step1}
	\frac{d\mathcal{U}_{\mathrm{int}}}{dt} = \frac{\fix{1}}{2}\int\!\!\int \widehat{\mathcal{W}}(\vec{\nu},\vec{\nu}_*)\,\bigl[(\partial_t f)\, f_* \,+\, f\, (\partial_t f_*)\bigr]\, d\vec{\Xi}\, d\vec{\Xi}_*\, d\vec{x}\, d\vec{x}_*.
\end{equation}

The second summand can be brought to the form of the first by the change of variables
\begin{equation}
	\label{eq:vp_swap}
	\quad (\vec{x},\vec{p},\vec{\nu},\vec{\varsigma}) \,\longleftrightarrow\, (\vec{x}_*,\vec{p}_*,\vec{\nu}_*,\vec{\varsigma}_*),
\end{equation}
which has unit Jacobian and leaves the integration domain invariant. After this transformation, $\partial_t f_*$ becomes $\partial_t f$, $f$ becomes $f_*$, and $\widehat{\mathcal{W}}(\vec{\nu},\vec{\nu}_*)$ becomes $\widehat{\mathcal{W}}(\vec{\nu}_*,\vec{\nu}) = \widehat{\mathcal{W}}(\vec{\nu},\vec{\nu}_*)$, the last equality being the exchange symmetry \eqref{eq:Wsym}. The two summands are then identical, and
\begin{equation}
	\label{eq:vp_step2}
		\frac{d\mathcal{U}_{\mathrm{int}}}{dt} = \fix{\int\!\!\int} \widehat{\mathcal{W}}(\vec{\nu},\vec{\nu}_*)\, f_*\, \partial_t f\, d\vec{\Xi}\, d\vec{\Xi}_*\, d\vec{x}\, d\vec{x}_*.
\end{equation}
Here and below, $\partial_t f$ refers to the time derivative of $f$ at the unstarred phase point $(\vec{x},\vec{p},\vec{\nu},\vec{\varsigma})$. The starred variables enter the integrand only through $f_*$ and $\widehat{\mathcal{W}}(\vec{\nu},\vec{\nu}_*)$. We now substitute $\partial_t f$ from the kinetic equation \eqref{eq:cfmz},
\begin{equation}
	\label{eq:vp_kinetic}
	\partial_t f = -\frac{1}{m}\vec{p}\cdot\nabla_{\vec{x}} f \;-\; \dot{\vec{\nu}}\cdot\nabla_{\vec{\nu}} f \;-\; \vec{\mathcal{V}}\cdot\nabla_{\vec{\varsigma}} f \;+\; \mathcal{Q}(f,f),
\end{equation}
with $\dot{\vec{\nu}} = \mat{B}^{-1}\vec{\varsigma}$, into \eqref{eq:vp_step2} and inspect the four resulting contributions in turn. The $\vec{\varsigma}$-streaming contribution vanishes because $\widehat{\mathcal{W}}(\vec{\nu},\vec{\nu}_*)$ and $f_*$ are independent of $\vec{\varsigma}$, and $\vec{\mathcal{V}}(\vec{x},\vec{\nu},t)$ is independent of $\vec{\varsigma}$ by definition \eqref{eq:potential}, so integration by parts in $\vec{\varsigma}$ produces $\nabla_{\vec{\varsigma}}\!\cdot\vec{\mathcal{V}}=0$.

The $\vec{x}$-streaming contribution vanishes because $\widehat{\mathcal{W}}$ depends only on $(\vec{\nu},\vec{\nu}_*)$, $f_*$ only on the starred variables, and $\vec{p}$ is independent of $\vec{x}$, so $\vec{p}\cdot\nabla_{\vec{x}}f = \nabla_{\vec{x}}\!\cdot(\vec{p}\,f)$ is a pure $\vec{x}$-divergence and integrates to zero by decay at infinity. The collision contribution vanishes because the binary rule \eqref{eq:binary_rule} leaves $\vec{\nu}$ and $\vec{\nu}_*$ unchanged during the collision: at fixed $(\vec{\nu},\vec{\nu}_*)$ the rule is a measure-preserving map of $(\vec{p},\vec{\varsigma},\vec{p}_*,\vec{\varsigma}_*)$ alone, so the standard reciprocity argument of \cite[Section 2.2]{carrillo} runs pointwise in $\vec{\nu}$ and gives
\begin{equation}
	\label{eq:vp_collision_zero}
	\int \mathcal{Q}(f,f)(\vec{x},\vec{p},\vec{\nu},\vec{\varsigma})\, d\vec{p}\, d\vec{\varsigma} = 0 \qquad\text{at each fixed }(\vec{x},\vec{\nu}),
\end{equation}
and pulling the $(\vec{p},\vec{\varsigma})$-independent factor $\widehat{\mathcal{W}}(\vec{\nu},\vec{\nu}_*)\,f_*$ outside the inner $(\vec{p},\vec{\varsigma})$-integral then makes the collision contribution vanish.

Equivalently, the orientation-preservation of \eqref{eq:binary_rule} extends the list of collision invariants \eqref{eq:collinv} by the entire algebra of functions of $\vec{\nu}$ alone: for any $\psi(\vec{\nu})$, $\int\mathcal{Q}\,\psi(\vec{\nu})\,d\vec{p}\,d\vec{\nu}\,d\vec{\varsigma}=0$. The factor $\widehat{\mathcal{W}}(\vec{\nu},\vec{\nu}_*)\,f_*$, viewed as a function of $\vec{\nu}$ at fixed starred variables, falls in this extended class, and that is why $\int\mathcal{Q}\,d\vec{p}\,d\vec{\varsigma}=0$ holds already at each fixed $(\vec{x},\vec{\nu})$, rather than only after integration over $\vec{\nu}$.

Only the $\vec{\nu}$-streaming contribution survives. The phase-space flow $\dot{\vec{\nu}} = \mat{B}^{-1}\vec{\varsigma}$ is divergence-free on the sphere ($\nabla_{\vec{\nu}}\cdot\dot{\vec{\nu}} = 0$), so integration by parts in $\vec{\nu}$, with boundary integral over $\partial\mathcal{M}$ vanishing either because $\mathcal{M}$ is closed or by decay if it is unbounded, gives
\begin{equation}
	\label{eq:vp_Tnu_step1}
	\frac{d\mathcal{U}_{\mathrm{int}}}{dt} = \fix{\int\!\!\int} (\nabla_{\vec{\nu}} \widehat{\mathcal{W}})\, f_*\,\cdot\,\dot{\vec{\nu}}\, f\, d\vec{\Xi}\, d\vec{\Xi}_*\, d\vec{x}\, d\vec{x}_*.
\end{equation}
The integral of $(\nabla_{\vec{\nu}}\widehat{\mathcal{W}})\,f_*$ over the starred variables equals $-\vec{\mathcal{V}}(\vec{x},\vec{\nu},t)$ by the definition \eqref{eq:potential} of the Vlasov torque, and substituting this into \eqref{eq:vp_Tnu_step1} together with $\rho\cchevrons{A} = m\!\int A f\,d\vec{\Xi}$ yields
\begin{equation}
	\frac{d\mathcal{U}_{\mathrm{int}}}{dt} = -\!\int \vec{\mathcal{V}}\cdot\dot{\vec{\nu}}\, f\, d\vec{\Xi}\, d\vec{x} = -\!\int \frac{\rho}{\fix{m}}\cchevrons{\vec{\mathcal{V}}\cdot\dot{\vec{\nu}}}\, d\vec{x},
\end{equation}
which is \eqref{eq:vlasov_power_conservation}.
\end{proof}

\begin{remark}\label{rmk:swap_vs_collision_map}
\cla{The change of dummy variables in \eqref{eq:vp_swap} used in the proof above is a relabelling of integration variables on a symmetric integration domain. Its Jacobian is unity by construction. It should not be conflated with the \emph{collision map} $(\vec{p},\vec{\varsigma},\vec{p}_*,\vec{\varsigma}_*)\mapsto(\vec{p}',\vec{\varsigma}',\vec{p}'_*,\vec{\varsigma}'_*)$ that connects pre- and post-collisional momenta in \eqref{eq:binary_rule}, which is a genuine phase-space transformation. The collision map has unit Jacobian for our top-symmetric rigid-rod elastic rule (the same fact that yields detailed balance for the molecular constituents at hand, cf.~\cite[Theorem 3.12]{carrillo} and the derivation of the Maxwellian in \Cref{app:maxwellian_derivation}), but can fail to be measure-preserving for asymmetric microstructured molecules. This is the structural reason detailed balance generally fails in microstructured fluids \cite{carrillo}. The relabelling \eqref{eq:vp_swap} used in our proof has nothing to do with this issue: it would have unit Jacobian regardless of detailed balance.}
\end{remark}

\fix{The local interaction-energy density associated with the mean-field potential is
\begin{equation}
	\label{eq:eint_local}
	\rho(\vec{x},t)\,e_{\mathrm{int}}(\vec{x},t) := \frac{1}{2}\int \widehat{\mathcal{W}}(\vec{\nu},\vec{\nu}_*)\, f(\vec{x},\vec{p},\vec{\nu},\vec{\varsigma},t)\, f(\vec{x}_*,\vec{p}_*,\vec{\nu}_*,\vec{\varsigma}_*,t)\, d\vec{\Xi}\, d\vec{\Xi}_*\, d\vec{x}_*,
\end{equation}
the local form of the mean-field interaction energy $\mathcal{U}_{\mathrm{int}}$ of \eqref{eq:Uint}. It enters the equilibrium Helmholtz free energy of \Cref{sec:helmholtz}, whose minimisation fixes the Gibbs alignment, but it is excluded from the constitutive free energy \eqref{eq:psi_decomposition} used below. Adjoining $e_{\mathrm{int}}$ to the kinetic internal energy would reintroduce the Vlasov power as a source in the energy balance, and hence in the entropy production \eqref{eq:entropy_production_constitutive}. This source is not of flux--affinity form, so the Rajagopal--Srinivasa constrained maximisation of \Cref{sec:hybrid} would not close. Ignoring this term is physically benign. Alignment emerges as a gradient flow relaxing the orientational free energy toward its minimiser, the aligned Gibbs state, so in the strongly aligned regime $s\to 1$ in which the closure is performed, $e_{\mathrm{int}}$ sits at a near-stationary value and the Vlasov torque $\cchevrons{\vec{\nu}\times\vec{\mathcal{V}}}$ vanishes at the aligned equilibrium, as we will later prove in \Cref{corollary:kuramoto,corollary:onsager}. Furthermore, as $e_{\mathrm{int}}$ enters the balances only under a total time derivative, this near-stationarity makes its contribution vanish to leading order.}

We conclude this section by summarising the balance laws derived above for the mass density $\rho$, the bulk velocity $\vec{u}$, the bulk intrinsic angular momentum $\vec{\mu}$, the \fix{total internal energy $\widetilde{e} = e_{\mathrm{tr}} + e_{\mathrm{rot}}$}, and the entropy density $\eta$, i.e.
\begin{subequations}
	\label{eq:balance_laws}
	\begin{align}
		&\partial_t\rho + \nabla_{\vec{x}}\cdot(\rho\vec{u}) = 0, &&\\
		&\rho\bigl[\partial_t\vec{u} + (\nabla_{\vec{x}}\vec{u})\vec{u}\bigr] - \nabla_{\vec{x}}\cdot\mathbb{T} = 0, && \mathbb{T} = -\rho\cchevrons{\vec{V}\otimes\vec{V}}, \\
		&\rho\bigl[\partial_t\vec{\mu} + (\nabla_{\vec{x}}\vec{\mu})\vec{u}\bigr] \fix{-} \nabla_{\vec{x}}\cdot\mathbb{M} = \rho\cchevrons{\vec{\nu}\times\vec{\mathcal{V}}}, && \mathbb{M} = -\rho\cchevrons{\vec{V}\otimes(\vec{\nu}\times\vec{\Sigma})}, \\
		&\fix{\rho\bigl[\partial_t \widetilde{e} + \vec{u}\cdot\nabla_{\vec{x}} \widetilde{e}\bigr] - \mathbb{T}:\nabla_{\vec{x}}\vec{u} - \tfrac{1}{mC_{m,\ell}}\mathbb{M}:\nabla_{\vec{x}}\vec{\mu} + \nabla_{\vec{x}}\cdot\vec{Q} = 0}, &&\\
		&\rho\bigl[\partial_t \eta + \vec{u}\cdot\nabla_{\vec{x}} \eta\bigr] + \nabla_{\vec{x}}\cdot \vec{\Phi} = \xi, && \vec{\Phi} = -\frac{\rho k_B}{m}\cchevrons{\vec{V}\log(f)},
	\end{align}
\end{subequations}
where the \fix{energy balance is written for the kinetic internal energy $\widetilde{e}$ and is source-free under the micro--macro identification, by which the entire Vlasov power is absorbed into the budget of the bulk intrinsic angular momentum, see \Cref{rmk:mu_source_in_energy}}. The Vlasov torque $\vec{\mathcal{V}}$, the bare heat flux $\vec{Q}$ and the entropy production $\xi$ are defined as
\begin{align}
	\vec{\mathcal{V}}(\vec{x},\vec{\nu},t) &= -\!\int_{\mathbb{R}^3}\!\int_{\mathbb{R}^3}\!\int_{\mathcal{M}}\!\int_{\mathbb{R}^3} \nabla_{\vec{\nu}} \mathcal{W}(\vec{x},\vec{\nu},\vec{\nu}_*,\vec{x}_*)\, f(\vec{x}_*,\vec{p}_*,\vec{\nu}_*,\vec{\varsigma}_*,t) \, d\vec{x}_*\, d\vec{p}_*\, d\vec{\nu}_*\, d\vec{\varsigma}_*,\\
	\vec{Q} &= \rho\cchevrons{\vec{V}\!\left(\tfrac{1}{2}\abs{\vec{V}}^2 + \tfrac{\abs{\vec{\Sigma}}^2}{2\fix{m}C_{m,\ell}}\right)}, \\
	\xi &= -k_B\!\int \mathcal{Q}(f,f)\log(f)\,d\vec{\Xi},
\end{align}
and the macroscopic state variables read
\begin{align}
	\rho(\vec{x},t) &= m\!\int f \, d\vec{p}\, d\vec{\nu}\, d\vec{\varsigma}, &
	\vec{u}(\vec{x},t) &= \cchevrons{\vec{v}} = \tfrac{1}{n}\!\int \vec{v} f \, d\vec{\Xi}, \\
	\vec{\mu}(\vec{x},t) &= \cchevrons{\vec{\nu}\times\vec{\varsigma}} = \tfrac{1}{n}\!\int (\vec{\nu}\times\vec{\varsigma}) f \, d\vec{\Xi}, &
	\fix{\widetilde{e}(\vec{x},t)} &= \tfrac{1}{2}\cchevrons{\abs{\vec{V}}^2} + \tfrac{1}{2\fix{m}C_{m,\ell}} \cchevrons{\abs{\vec{\Sigma}}^2}, \\
	\rho(\vec{x},t)\eta(\vec{x},t) &= -k_B\!\int f \log(f)\,d\vec{\Xi}.
\end{align}

\section{Second law of thermodynamics and emergence of alignment}
\label{sec:hthm}
As briefly mentioned in \Cref{sec:kt}, one of the crucial features of the kinetic theory developed in \cite{carrillo} is the existence of an H-theorem-type result under the assumption that the group action $\mathcal{A}$ is transitive \cite[Theorem 3.13]{carrillo}\footnote{The problem of proving this when the group action is not transitive remains open, and we conjecture that a form of the H-theorem remains valid even when $\mathcal{A}$ is not transitive.}.
The importance of an H-theorem-type result is twofold. On the one hand, it guarantees that the kinetic theory is in some sense\footnote{The issue of devising a precise mathematical statement of the second law of thermodynamics remains open. Many, including a subset of the authors, would not agree with identifying the H-theorem with the second law of thermodynamics. We refer the reader interested in this issue to \cite[Remark 5]{malek} and \cite[Chapter I, Section 5]{Truesdell}.} consistent with the second law of thermodynamics. On the other hand, it establishes that the kinetic equation admits a monotone entropy functional, which is a crucial step in the study of the long-time behaviour of the system and of its relaxation to equilibrium.
\uz{Replacing the full Boltzmann operator with the BGK operator \eqref{eq:BGK} of \Cref{sec:bgk} turns the H-theorem into an immediate consequence of \Cref{lemma:BGK_entropy}: the local entropy production $\xi$ is non-negative pointwise, with equality if and only if $f$ coincides with the Maxwellian \eqref{eq:maxwellian}, and integrating the entropy balance \eqref{eq:entropy_balance_law} over the spatial domain upgrades this to a global statement. 

\begin{proposition}
	\label{prop:H_theorem}
	Let $f$ be a solution of \eqref{eq:cfmz} with the BGK collision operator \eqref{eq:BGK} of \Cref{sec:bgk}, and let $\xi$ be the entropy production of \eqref{eq:entropy_balance_law}. Then under the assumption that the entropy flux $\vec{\Phi}$ decays sufficiently fast at infinity the entropy of the system is non-decreasing in time, i.e.
	\begin{equation}
		\frac{d}{dt}\mathcal{H}(t) \geq 0, \qquad \mathcal{H}(t) = \int \rho(\vec{x},t)\eta(\vec{x},t) \, d\vec{x},
	\end{equation}
	with equality if and only if $f$ is the Maxwellian \eqref{eq:maxwellian}.
\end{proposition}
\begin{proof}
	Differentiating $\mathcal{H}(t)$ and using \eqref{eq:continuity_eq}, \eqref{eq:entropy_balance_law} and the divergence theorem,
	\begin{equation*}
		\frac{d}{dt}\mathcal{H}(t) = \int \rho\bigl[\partial_t \eta + \vec u\cdot\nabla_{\vec x}\eta\bigr]\, d\vec{x} = -\int \nabla_{\vec{x}}\!\cdot\!\vec{\Phi}\, d\vec{x} + \int \xi\, d\vec{x} = \int \xi\, d\vec{x} \geq 0,
	\end{equation*}
	the decay hypothesis on $\vec{\Phi}$ discarding the divergence integral and the last inequality being \Cref{lemma:BGK_entropy}.
\end{proof}
}

The H-theorem of \Cref{prop:H_theorem} controls the relaxation of the distribution function $f$ in the $(\vec{p},\vec{\varsigma})$ variables, but it gives no information about the orientational distribution $\lambda(\vec{\nu},t)$, which enters the Maxwellian \eqref{eq:maxwellian} as a free parameter. To determine $\lambda$, we must keep track of the energy stored in the mean-field potential, which the entropy balance \eqref{eq:entropy_balance_law} cannot see. This is analogous to keeping track of the $L^2$ norm of the electric field when studying waves in plasmas \cite{villani, landau}. We will do so in the next section. 

\subsection{Helmholtz free energy}\label{sec:helmholtz}
Given the form of the Maxwellian \eqref{eq:maxwellian}, determined by the macroscopic quantities $\rho, \lambda, \theta, \vec{u}, \vec{\mu}$, we adopt the Helmholtz free energy as the thermodynamic potential,
\begin{equation}
	\rho(\vec{x},t)\psi(\vec{x},t) = \rho(\vec{x},t)\Big(e(\vec{x},t)-\theta(\vec{x},t)\eta(\vec{x},t)\Big).
\end{equation}
It is illuminating to unpack the equilibrium value of the free-energy density $\rho \psi^{(0)}$ and exhibit its decomposition into a standard polyatomic-gas contribution and a purely orientational contribution. At equilibrium, $e^{(0)}_{\mathrm{tr}}$ and $e^{(0)}_{\mathrm{rot}}$ can be computed from \eqref{eq:maxwellian}:
\begin{equation}
	\label{eq:equipartition_easy}
	e^{(0)}_{\mathrm{tr}} = \frac{3}{2}\frac{k_B}{m}\theta, \qquad e^{(0)}_{\mathrm{rot}} = \frac{k_B}{m}\theta,
\end{equation}
where the prefactor of $e^{(0)}_{\mathrm{rot}}$ reflects the two-dimensional orientational momentum space $T_{\vec{\nu}}\mathbb{S}^2$ together with the reduction $C_{m,\ell}$ from \eqref{eq:maxwellian}. A direct calculation of the kinetic entropy density $\rho\eta^{(0)} = -k_B\!\int f^{(0)}\log f^{(0)}\,d\vec{\Xi}$ then yields, up to an additive constant independent of the thermodynamic state,
\begin{equation}
	\label{eq:entropy_density_maxwellian}
	\rho\eta^{(0)} = \frac{k_B\rho}{m}\!\left[\frac{5}{2}\bigl(1+\log(2\pi k_B\theta)\bigr) - \log\frac{\rho}{m} - \int_{\mathcal{M}} \lambda(\vec{\nu})\log\lambda(\vec{\nu})\, d\vec{\nu}\right],
\end{equation}
and combining \eqref{eq:equipartition_easy}--\eqref{eq:entropy_density_maxwellian} with the local interaction energy \eqref{eq:eint_local} produces the local Helmholtz free-energy density at equilibrium,
\begin{equation}
	\label{eq:helmholtz_density_maxwellian}
	\rho\psi^{(0)} = \rho\bigl(e^{(0)}_{\mathrm{tr}} + e^{(0)}_{\mathrm{rot}} + e_{\mathrm{int}}^{(0)}\bigr) - \theta\rho\eta^{(0)}.
\end{equation}
\fix{The interaction energy $e_{\mathrm{int}}^{(0)}$ enters here, in the equilibrium free energy whose minimisation fixes the Gibbs alignment of \Cref{sec:gibbs_variational}, but not in the constitutive free energy \eqref{eq:psi_decomposition} that drives the Chapman--Enskog closure, for the reasons given after \eqref{eq:eint_local}.}
Collecting in \eqref{eq:helmholtz_density_maxwellian} all terms that depend only on $(\rho,\theta)$ and not on the orientational profile $\lambda$ defines the \emph{standard} part of the free-energy density,
\begin{equation}
	\label{eq:psi_std}
	\rho\psi^{(0)}_{\mathrm{std}}(\rho,\theta) := -\frac{k_B\theta\rho}{m}\!\left[\frac{5}{2}\log(2\pi k_B\theta) - \log\frac{\rho}{m}\right] + \text{const.},
\end{equation}
which coincides with the Helmholtz free-energy density of an ideal polyatomic gas with three translational and two rotational degrees of freedom \cite{landau}, while the residue
\begin{equation}
	\label{eq:psi_ori}
	\rho\psi^{(0)}_{\mathrm{ori}}[\lambda;\rho,\theta] := \rho\, e_{\mathrm{int}}[\lambda;\rho] + \frac{k_B\theta\rho}{m}\int_{\mathcal{M}} \lambda(\vec{\nu})\log\lambda(\vec{\nu})\, d\vec{\nu}
\end{equation}
is purely orientational and isolates the part of the free energy responsible for alignment dynamics. \uz{The total Helmholtz free energy obtained by integrating $\rho\psi^{(0)}$ over the spatial domain,
\begin{equation}
	\label{eq:helmholtz_free_energy}
	\mathfrak{A}(t) := \int_\Omega \rho(\vec{x},t)\,\psi^{(0)}(\vec{x},t)\, d\vec{x},
\end{equation}
plays the role of a thermodynamic potential whose dissipation rate is studied in \Cref{sec:hybrid}.} Thus the Helmholtz free energy can be decomposed as
\begin{equation}
	\label{eq:helmholtz_decomposition}
	\psi^{(0)} = \psi^{(0)}_{\mathrm{std}}(\rho,\theta) + \psi^{(0)}_{\mathrm{ori}}[\lambda;\rho,\theta].
\end{equation}

\subsection{Orientational Gibbs distribution}\label{sec:gibbs_variational}
Henceforth we assume that $\lambda$ has the form of a Gibbs distribution \cite{onsager, virgaOnsager}:
	\begin{equation}
		\label{eq:gibbs_variational}
		\uz{\lambda(\vec{x}, \vec{\nu}) = \frac{1}{Z(\vec{x})}\exp\!\left(-\frac{1}{m k_B\theta}\int_\Omega \int_{\mathbb{S}^2} \widehat{\mathcal{W}}(\abs{\vec{x}-\vec{x}_*}, \vec{\nu},\vec{\nu}_*)\, \rho(\vec{x}_*)\,\lambda(\vec{x}_*,\vec{\nu}_*)\, d\vec{x}_*\, d\vec{\nu}_*\right),}
	\end{equation}
	\uz{where $Z(\vec{x})$ is chosen so that $\int_{\mathbb{S}^2}\lambda(\vec{x},\vec{\nu})\,d\vec{\nu} = 1$ at every $\vec{x}$.}
To illustrate the emergence of different orientational correlations, we now focus on the space-homogeneous setting where every molecule interacts with every other molecule regardless of distance:
	\begin{equation}
		\label{eq:gibbs_variational_hom}
		\uz{\lambda(\vec{\nu}) = \frac{1}{Z}\exp\!\left(-\frac{1}{k_B\theta}\int_{\mathbb{S}^2} \widehat{\mathcal{W}}(\vec{\nu},\vec{\nu}_*)\,\lambda(\vec{\nu}_*)\, d\vec{\nu}_*\right),}
	\end{equation}
	with $Z$ chosen so that $\int_{\mathbb{S}^2}\lambda\,d\vec{\nu} = 1$.
The remainder of this section studies specific forms of interaction potentials that often appear in the context of complex fluids, namely the Kuramoto potential for molecules without head-to-tail symmetry and the Maier--Saupe potential for molecules with head-to-tail symmetry.
\subsubsection{Kuramoto-type potentials}

We first consider the \cla{attractive} Kuramoto-type kernel $\cla{\mathcal{W}_K}(\vec{\nu},\vec{\nu}_*) = \cla{-}\,\vec{\nu}\cdot\vec{\nu}_*$. 
For this kernel, \eqref{eq:gibbs_variational} reduces to the Gibbs distribution
\begin{equation}
	\label{eq:gibbs}
	\lambda(\vec{\nu}) = \frac{1}{Z} \exp\left(-\frac{1}{k_B\theta}\int_{\cla{\mathbb{S}^2}}\cla{\widehat{\mathcal{W}}}(\vec{\nu},\vec{\nu}_*)\lambda(\vec{\nu}_*)\,d\vec{\nu}_*\right),
\end{equation}
where $Z$ is a normalisation constant ensuring $\int_{\cla{\mathbb{S}^2}}\lambda\,d\vec{\nu}=1$ and the dimensionless coupling has been absorbed into a redefinition of the temperature scale.
\begin{proposition}
	\label{corollary:kuramoto}
		Under the hypothesis that the equilibrium distribution of molecular orientations $\lambda$ is given by \eqref{eq:gibbs}, with a Kuramoto-type mean-field potential 
	\begin{equation}
		\cla{\mathcal{W}_K}(\vec{\nu}, \vec{\nu}_*) = \cla{-}\,\vec{\nu}\cdot\vec{\nu}_*,
	\end{equation}
	\fix{the order parameter $\vec{\gamma} = \int_{\mathbb{S}^2}\vec{\nu}\,\lambda(\vec{\nu})\,d\vec{\nu}$ satisfies the self-consistent equation
	\begin{equation}\label{eq:gibbs_self_consistent}
		\vec{\gamma} = \frac{1}{Z}\int_{\mathbb{S}^2}\vec{\nu}\,\exp\!\left(\tfrac{1}{k_B\theta}\,\vec{\nu}\cdot\vec{\gamma}\right)d\vec{\nu},\qquad Z = \int_{\mathbb{S}^2}\exp\!\left(\tfrac{1}{k_B\theta}\,\vec{\nu}\cdot\vec{\gamma}\right)d\vec{\nu},
	\end{equation}
	and} the following statements hold:
	\begin{enumerate}
		\item The isotropic distribution $\lambda_{\mathrm{iso}}(\vec{\nu}) = \frac{1}{4\pi}$ is compatible with the Gibbs distribution and the Vlasov torque $\vec{\mathcal{V}}$ vanishes at equilibrium.
		\item If $k_B\theta \cla{>} \frac{1}{3}$ the isotropic distribution is the only solution of \eqref{eq:gibbs_self_consistent}, while if $k_B\theta \cla{<} \frac{1}{3}$ non-zero aligned solutions of \eqref{eq:gibbs_self_consistent} emerge.
		\item The non-zero aligned solutions of \eqref{eq:gibbs_self_consistent} are such that the torque induced by the Vlasov torque $\vec{\mathcal{V}}$ is zero, i.e.\ $\cchevrons{\vec{\nu}\times\vec{\mathcal{V}}} = \vec{0}$.
	\end{enumerate}
\end{proposition}
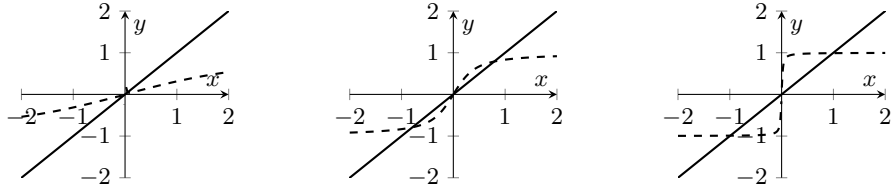
\begin{figure}
\begin{tikzpicture}
	\begin{groupplot}[
	    group style={
		group size=3 by 1,
		horizontal sep=1.6cm
	    },
	    width=0.32\textwidth,
	    height=0.28\textwidth,
	    axis lines=middle,
	    xlabel={$x$},
	    ylabel={$y$},
	    xmin=-2, xmax=2,
	    ymin=-2, ymax=2,
	    samples=300,
	    unbounded coords=jump,
	]

	\nextgroupplot[]
	\addplot[thick] {x};
	\addplot[thick,dashed,domain=-2:-0.02] {cosh(x)/sinh(x) - 1/x};
	\addplot[thick,dashed,domain=0.02:2]  {cosh(x)/sinh(x) - 1/x};

	\nextgroupplot[]
	\addplot[thick] {x};
	\addplot[thick,dashed,domain=-2:-0.01] {cosh(6*x)/sinh(6*x) - 1/(6*x)};
	\addplot[thick,dashed,domain=0.01:2]   {cosh(6*x)/sinh(6*x) - 1/(6*x)};

	\nextgroupplot[]
	\addplot[thick] {x};
	\addplot[thick,dashed,domain=-2:-0.005] {cosh(100*x)/sinh(100*x) - 1/(100*x)};
	\addplot[thick,dashed,domain=0.005:2]   {cosh(100*x)/sinh(100*x) - 1/(100*x)};

	\end{groupplot}
\end{tikzpicture}
\caption{Graphical representation of the self-consistent equation for the order parameter $s$ in the case of a Kuramoto-type mean-field potential. The dashed lines represent the function $I(s) = \coth\left(\frac{s}{k_B\theta}\right) - \frac{k_B\theta}{s}$ while the solid line represents the function $s$. Each plot corresponds to a different value of the temperature $\theta$; from left to right, the temperature $\theta$ decreases.}
\end{figure}

The bifurcation discussed in \Cref{corollary:kuramoto} is a pitchfork bifurcation of \hfe{$\mathfrak{A}[\lambda]$}: \cla{above the critical temperature ($k_B\theta > 1/3$) the isotropic branch is the unique minimiser, while below it ($k_B\theta < 1/3$)} the isotropic branch becomes a saddle and the ordered branches are global minimisers of \hfe{$\mathfrak{A}$}. This analysis constitutes a structural improvement over the inviscid theory of \cite{farrell}, which closes the moment hierarchy around an already-aligned state and is therefore unable to capture the transition.

\begin{proof}
	We define the order parameter $\vec{\gamma}$ related to the state of alignment of the system, i.e.
	\begin{equation}
		\vec{\gamma} = \int_{\mathcal{M}}\vec{\nu}\lambda(\vec{\nu})\,d\vec{\nu} \quad \Rightarrow \quad \vec{\nu}\cdot\vec{\gamma} = \int_{\mathcal{M}}\vec{\nu}\cdot\vec{\nu}_*\lambda(\vec{\nu}_*)\,d\vec{\nu}_*,
	\end{equation}
		Notice that $\vec{\gamma}$ is no longer constrained to $\mathbb{S}^2$ and that, \cla{using $\mathcal{W} = -\vec{\nu}\cdot\vec{\nu}_*$}, we can rewrite the Gibbs distribution as
	\begin{equation}
		\lambda(\vec{\nu}) = \frac{1}{Z} \exp\!\left(\cla{\frac{1}{k_B\theta}}\,\vec{\nu}\cdot\vec{\gamma}\right).
	\end{equation}
	Substituting this form of $\lambda(\vec{\nu})$ in the definition of $\vec{\gamma}$ recovers the self-consistent equation \eqref{eq:gibbs_self_consistent} for the order parameter $\vec{\gamma}$, namely
	\begin{equation}
		\vec{\gamma} = \int_{\mathcal{M}}\vec{\nu}\lambda(\vec{\nu})\,d\vec{\nu} = \frac{1}{Z} \int_{\mathcal{M}}\vec{\nu}\exp\!\left(\cla{\frac{1}{k_B\theta}}\,\vec{\nu}\cdot\vec{\gamma}\right)\,d\vec{\nu}.
	\end{equation}
	We begin by observing that the isotropic state, i.e.\ $\lambda_{\mathrm{iso}}(\vec{\nu}) = \frac{1}{4\pi}$, is always a solution of the previous equation, since
	\begin{equation}
		\vec{\gamma}=\int_{\mathcal{M}}\vec{\nu}\lambda_{\mathrm{iso}}(\vec{\nu})\,d\vec{\nu} = \frac{1}{4\pi} \int_{\mathcal{M}}\vec{\nu}\,d\vec{\nu} = 0\; \text{and} \; \lambda_{\mathrm{iso}}(\vec{\nu}) = \frac{1}{Z} \exp\!\left(\cla{\frac{1}{k_B\theta}}\,\vec{\nu}\cdot\vec{\gamma}\right) = \frac{1}{Z} \exp(0) = \frac{1}{Z},
	\end{equation}
	and $Z$ is equal to $4\pi$ since the area of the unit sphere $\mathbb{S}^2$ is $4\pi$.
	Furthermore, the isotropic state is compatible with the condition that the Vlasov torque $\vec{\mathcal{V}}$ vanishes at equilibrium, since \cla{with $\mathcal{W} = -\vec{\nu}\cdot\vec{\nu}_*$ we have $-\nabla_{\vec{\nu}}\mathcal{W} = \vec{\nu}_*$, hence}
	\begin{equation}
		\label{eq:Vlasov_self_consistent}
			\vec{\mathcal{V}}(\vec{\nu}) = -\int \nabla_{\vec{\nu}}\cla{[-\vec{\nu}\cdot\vec{\nu}_*]}\,\lambda_{\mathrm{iso}}(\vec{\nu}_*)f^{(0)}_{\vec{p}_*,\vec{\varsigma}_*} \, d\vec{\Xi}_* = \int \vec{\nu}_* \lambda_{\mathrm{iso}}(\vec{\nu}_*)f^{(0)}_{\vec{p}_*,\vec{\varsigma}_*} \, d\vec{\Xi}_* = \vec{\gamma}_{\mathrm{iso}} = \vec{0}.
		\end{equation}
		More generally, we can rewrite $\vec{\gamma}$ as
	\begin{equation}
		\vec{\gamma} = s\widehat{\vec{\nu}} \label{eq:nematic_dir_ericksen}
	\end{equation}
	where $s$ is the scalar order parameter and $\widehat{\vec{\nu}}\in \mathbb{S}^2$ is the aligned director of the system. Under this rewriting the self-consistent equation for the order parameter $\vec{\gamma}$ can be rewritten\footnote{The idea originated in the context of the Heisenberg model. To see the rewriting it is sufficient to rewrite the Gibbs distribution in terms of $\alpha=\frac{s}{k_B \theta}$ as $\lambda(\vec{\nu})=\frac{\alpha}{4\pi \sinh(\alpha)}\exp(\alpha\vec{\nu}\cdot\widehat{\vec{\nu}})$ and evaluate its first moment to be $(\coth(\alpha)-\frac{1}{\alpha})\widehat{\vec{\nu}}$ and substitute it on the right-hand side of the self-consistency equation.} as a self-consistent equation for the scalar order parameter $s$ \cite{Kirkpatrick}, i.e.
	\begin{equation}
		s = I(s), \qquad I (s) = \coth\left(\frac{s}{k_B\theta}\right) - \frac{k_B\theta}{s}.
	\end{equation}
		Taylor-expanding $I(s)$ near the critical point $s = 0$, which we have previously shown to be a solution of the self-consistent equation corresponding to the isotropic state, we obtain
	\begin{equation}
		s = \frac{s}{3k_B\theta} + O(s^3).
	\end{equation}
		Therefore, classical results in bifurcation theory suggest that if $k_B\theta \cla{>} \frac{1}{3}$ the isotropic state is the only solution of the self-consistent equation, while if $k_B\theta \cla{<} \frac{1}{3}$ a pair of non-zero scalar branches $s=\pm s_*$ emerges for a fixed polar axis, equivalently an aligned branch with arbitrary orientation when $s$ is taken to be non-negative.
		Lastly, we prove that for the non-zero aligned solutions of the self-consistent equation for the orientational distribution $\lambda$, the mean-field torque $\cchevrons{\vec{\nu}\times\vec{\mathcal{V}}}$ is zero.
	\begin{equation}
		\cchevrons{\vec{\nu}\times\vec{\mathcal{V}}} = \cchevrons{\vec{\nu}\times\vec{\gamma}} = \int_{\mathbb{S}^2} \vec{\nu}\times \vec{\gamma}\, \lambda(\vec{\nu})\,d\vec{\nu} = \left(\int_{\mathbb{S}^2} \vec{\nu}\, \lambda(\vec{\nu})\,d\vec{\nu}\right)\times\vec{\gamma} = \vec{\gamma}\times\vec{\gamma} = \vec{0}.
	\end{equation}
\end{proof}
As discussed in \Cref{rmk:vlasov_torque}, the mean-field torque $\rho\cchevrons{\vec{\nu}\times\vec{\mathcal{V}}}$ is the only term in the balance law for the bulk intrinsic angular momentum $\vec{\mu}$ that can potentially drive a change in the state of alignment. For any self-consistent solution of the Gibbs distribution, this torque vanishes; such states are therefore equilibrium states of the system.

\subsubsection{Maier--Saupe-type potentials}
\cla{The steric Onsager kernel $\mathcal{W}_{\mathrm{O}}(\vec{\nu},\vec{\nu}_*) = c\,|\sin\angle(\vec{\nu},\vec{\nu}_*)|$, derived as the excluded volume between two thin rods of relative orientation $\angle(\vec{\nu},\vec{\nu}_*)$, is non-differentiable at $\vec{\nu}_*=\vec{\nu}$ and is therefore not amenable to a direct Taylor expansion at aligned configurations. We replace it by the Maier--Saupe kernel
\begin{equation}
	\label{eq:maierSaupe}
	\mathcal{W}_{\mathrm{MS}}(\vec{\nu},\vec{\nu}_*) = -(\vec{\nu}\cdot\vec{\nu}_*)^2,
\end{equation}
which preserves the head-to-tail symmetry $(\vec{\nu},\vec{\nu}_*)\mapsto(-\vec{\nu},-\vec{\nu}_*)$ of the nematic phase and is the standard mean-field proxy for steric nematogens used in the liquid-crystal literature.}

\cla{\begin{proposition}\label{corollary:onsager}
	Assuming the orientational distribution at equilibrium is given by the Gibbs distribution \eqref{eq:gibbs} with the Maier--Saupe-type mean-field potential \eqref{eq:maierSaupe}, the following statements hold:
	\begin{enumerate}
		\item The isotropic distribution $\lambda_{\mathrm{iso}}(\vec{\nu}) = \frac{1}{4\pi}$ is compatible with the Gibbs distribution and the induced torque $\vec{\nu}\times\vec{\mathcal{V}}$ vanishes at the isotropic equilibrium.
		\item Because $\mathcal{W}_{\mathrm{MS}}$ is invariant under $\vec{\nu}\mapsto -\vec{\nu}$, the first moment $\int\vec{\nu}\,\lambda\,d\vec{\nu}$ vanishes at every equilibrium and the natural order parameter is the symmetric traceless Q-tensor
		\begin{equation}
			\label{eq:Qtensor}
			\mat{Q} \;=\; \int_{\mathbb{S}^2} \Bigl(\vec{\nu}\otimes\vec{\nu} - \tfrac{1}{3}\mat{I}\Bigr) \lambda(\vec{\nu})\,d\vec{\nu} \fix{\;=\; S\Bigl(\widehat{\vec{\nu}}\otimes\widehat{\vec{\nu}} - \tfrac{1}{3}\mat{I}\Bigr)},
		\end{equation}
		\fix{which is uniaxial about the director $\widehat{\vec{\nu}}$ at equilibrium, with scalar order parameter $S = \tfrac{3}{2}\,\widehat{\vec{\nu}}\cdot\mat{Q}\,\widehat{\vec{\nu}} = \cchevrons{\tfrac{1}{2}\bigl(3(\vec{\nu}\cdot\widehat{\vec{\nu}})^2-1\bigr)} \in [-\tfrac{1}{2}, 1]$.}
		\item The self-consistent equation for $S$ takes the form $S = \mathcal{I}_{\mathrm{MS}}(S)$ with
		\begin{equation}
			\label{eq:MS_self_consistent}
			\mathcal{I}_{\mathrm{MS}}(S) = \frac{\displaystyle\int_{-1}^{1} \tfrac{3x^2-1}{2}\,e^{\,S(3x^2-1)/(3k_B\theta)}\,dx}{\displaystyle\int_{-1}^{1} e^{\,S(3x^2-1)/(3k_B\theta)}\,dx}.
		\end{equation}
		The Taylor expansion of $\mathcal{I}_{\mathrm{MS}}$ near $S=0$ gives $\mathcal{I}_{\mathrm{MS}}(S) = (2/(15\,k_B\theta))\,S + O(S^2)$. The isotropic branch is therefore locally stable for $k_B\theta > 2/15$ and locally unstable for $k_B\theta < 2/15$.
		\item On the aligned branch, the Vlasov torque $\vec{\mathcal{V}}(\vec{\nu}) = 2\,\mat{A}\,\vec{\nu}$ with $\mat{A} = \tfrac{1}{3}\mat{I}+\mat{Q}$ is uniaxial about $\widehat{\vec{\nu}}$, and the induced torque $\cchevrons{\vec{\nu}\times\vec{\mathcal{V}}}$ vanishes.
	\end{enumerate}
\end{proposition}}

\cla{\begin{proof}
	(i) The isotropic check is identical to \Cref{corollary:kuramoto}: $\lambda_{\mathrm{iso}}(\vec{\nu})=\tfrac{1}{4\pi}$ has $\int\vec{\nu}\otimes\vec{\nu}\,\lambda_{\mathrm{iso}}\,d\vec{\nu} = \tfrac{1}{3}\mat{I}$, hence $\mat{Q}_{\mathrm{iso}}=0$, and
	\begin{equation}
		\vec{\mathcal{V}}_{\mathrm{iso}}(\vec{\nu}) = -\int_{\mathbb{S}^2}\nabla_{\vec{\nu}}\bigl[-(\vec{\nu}\cdot\vec{\nu}_*)^2\bigr]\,\lambda_{\mathrm{iso}}(\vec{\nu}_*)\,d\vec{\nu}_* = 2\bigl(\tfrac{1}{3}\mat{I}\bigr)\vec{\nu} = \tfrac{2}{3}\vec{\nu},
	\end{equation}
	so $\vec{\mathcal{V}}_{\mathrm{iso}}$ is parallel to $\vec{\nu}$ and therefore has zero torque against $\vec{\nu}$.

	(ii) For a general Gibbs solution, substituting \eqref{eq:maierSaupe} into the Vlasov potential gives
	\begin{equation}
		\int \mathcal{W}_{\mathrm{MS}}(\vec{\nu},\vec{\nu}_*)\,\lambda(\vec{\nu}_*)\,d\vec{\nu}_* = -\vec{\nu}^T\!\left(\int \vec{\nu}_*\otimes\vec{\nu}_*\,\lambda(\vec{\nu}_*)\,d\vec{\nu}_*\right)\vec{\nu} = -\vec{\nu}^T \mat{A}\,\vec{\nu},
	\end{equation}
	with $\mat{A} = \tfrac{1}{3}\mat{I}+\mat{Q}$. The Gibbs distribution therefore takes the Bingham form
	\begin{equation}
		\label{eq:bingham}
		\lambda(\vec{\nu}) = \frac{1}{Z}\exp\!\left(\frac{1}{k_B\theta}\,\vec{\nu}^T\mat{Q}\,\vec{\nu}\right),
	\end{equation}
	after absorbing the trace into $Z$. At uniaxial states
	\begin{equation}
		\mat{Q} = S(\widehat{\vec{\nu}}\otimes\widehat{\vec{\nu}} - \tfrac{1}{3}\mat{I}), \label{eq:nematic_dir_degennes}
	\end{equation}
	hence $\vec{\nu}^T\mat{Q}\,\vec{\nu} = (S/3)\bigl(3(\vec{\nu}\cdot\widehat{\vec{\nu}})^2 - 1\bigr)$ and
	\begin{equation}
		\lambda(\vec{\nu}) = \frac{1}{Z}\exp\!\left(\frac{S\bigl(3(\vec{\nu}\cdot\widehat{\vec{\nu}})^2-1\bigr)}{3k_B\theta}\right).
	\end{equation}
	Computing the second moment of $\lambda$ in spherical coordinates with $\widehat{\vec{\nu}}$ as the polar axis gives the self-consistent equation $S = \tfrac{1}{2}\cchevrons{3(\vec{\nu}\cdot\widehat{\vec{\nu}})^2 - 1}_\lambda$, which is precisely \eqref{eq:MS_self_consistent}.

	(iii) Taylor-expanding the exponential in \eqref{eq:MS_self_consistent}, the leading non-trivial term comes from $\int_{-1}^{1} \bigl((3x^2-1)/2\bigr)^2\,dx = 2/5$, while $\int_{-1}^{1} (3x^2-1)/2\,dx = 0$ removes the constant term. Hence
	\begin{equation}
		\mathcal{I}_{\mathrm{MS}}(S) = \frac{\bigl(2S/(3k_B\theta)\bigr)(2/5)}{2} + O(S^2) = \frac{2 S}{15\,k_B\theta} + O(S^2).
	\end{equation}
	The linearised self-consistency $S = (2/(15 k_B\theta))\,S + O(S^2)$ admits a non-trivial root near $S=0$ exactly when $2/(15 k_B\theta) > 1$, i.e., $k_B\theta < 2/15$.

	(iv) From $\mathcal{W}_{\mathrm{MS}}(\vec{\nu},\vec{\nu}_*) = -(\vec{\nu}\cdot\vec{\nu}_*)^2$, one computes
	\begin{equation}
		\vec{\mathcal{V}}(\vec{\nu}) = -\int \nabla_{\vec{\nu}}\mathcal{W}_{\mathrm{MS}}\,\lambda(\vec{\nu}_*)\,d\vec{\nu}_* = 2\,\mat{A}\,\vec{\nu},
	\end{equation}
	with $\mat{A} = \tfrac{1}{3}\mat{I}+\mat{Q}$. The induced torque is
	\begin{equation}
		\cchevrons{\vec{\nu}\times\vec{\mathcal{V}}} = 2\int_{\mathbb{S}^2}\vec{\nu}\times(\mat{A}\vec{\nu})\,\lambda(\vec{\nu})\,d\vec{\nu} = 2\int_{\mathbb{S}^2}\vec{\nu}\times(\mat{Q}\vec{\nu})\,\lambda(\vec{\nu})\,d\vec{\nu},
	\end{equation}
	since $\vec{\nu}\times(\tfrac{1}{3}\mat{I}\vec{\nu}) = \tfrac{1}{3}\vec{\nu}\times\vec{\nu}=\vec{0}$. At uniaxial states, $\mat{Q}\vec{\nu} = S\bigl((\widehat{\vec{\nu}}\cdot\vec{\nu})\widehat{\vec{\nu}} - \tfrac{1}{3}\vec{\nu}\bigr)$, so $\vec{\nu}\times(\mat{Q}\vec{\nu}) = S(\widehat{\vec{\nu}}\cdot\vec{\nu})(\vec{\nu}\times\widehat{\vec{\nu}})$. The integrand is odd under the $\mathrm{SO}(2)$-rotation of $\vec{\nu}$ about $\widehat{\vec{\nu}}$, and $\lambda$ is $\mathrm{SO}(2)$-invariant about $\widehat{\vec{\nu}}$, so the integral vanishes.
\end{proof}

\Cref{corollary:onsager} is the nematic counterpart of \Cref{corollary:kuramoto}, with the polar order parameter $\vec{\gamma}\in\mathbb{R}^3$ replaced by the axial order parameter $\mat{Q}\in\mathrm{Sym}_0(\mathbb{R}^3)$. The same orthogonality mechanism that eliminates the induced torque in the Kuramoto case operates here as well, with the polar director $\widehat{\vec{\nu}}$ replaced by its head-to-tail-symmetric axial counterpart.}

The analysis presented in the previous subsections constitutes a structural improvement over the inviscid theory of \cite{farrell}, which neglects the Vlasov torque producing alignment, and hence is not capable of explaining why a strongly aligned state is compatible with the microscopic dynamics of nematic fluids. 

\subsection{One-constant Oseen--Frank energy}\label{sec:oseen_frank_emergence}
In modern variational and dynamic theories of liquid crystals it is common to express the internal energy associated with the average configuration of the rod-like molecules via the nematic director field, which in our case corresponds to $\widehat{\vec{\nu}}(\vec{x}, t)\in \mathbb{S}^2$ defined either as in \eqref{eq:nematic_dir_ericksen} or as in \eqref{eq:nematic_dir_degennes} \cite{virgaCV, ericksen1960,ericksen1961, ericksen1962, leslie1968}. A typical modelling assumption is that the internal energy of nematic fluids takes the form
	\begin{equation}
		I[\widehat{\vec{\nu}}] = \int_\Omega \frac{K}{2} \abs{\nabla_{\vec{x}}\widehat{\vec{\nu}}}^2\, d\vec{x}, \label{eq:one-constant}
	\end{equation}
	which corresponds to the intuitive idea that all distortions of the nematic state ``cost'' equally. Despite being an extreme simplification, to the point that de~Gennes \& Prost referred to it as the `poor man's' model of liquid crystals \cite{degennes} and Stephen \& Straley called it the theoretician's energy \cite{theo}, such an approximation is still capable of predicting many quantitative aspects of liquid crystal phenomena. This section is devoted to deriving a form of \eqref{eq:one-constant} from the kinetic model considered here, under the assumption $s \to 1$. \fix{The one-constant elastic energy emerges from the rotational kinetic energy at strong alignment in \Cref{lemma:erot_oseen_frank}, with Frank modulus $K = K_{\mathrm{kin}} = C_{m,\ell}\,p/m$.}

\section{Equations of state and thermodynamic constitutive relations}\label{sec:eos}
In this section we study the relations among the four relevant thermodynamic variables of the system: the mass density $\rho$, the \fix{total internal energy $\widetilde{e} = e_{\mathrm{tr}} + e_{\mathrm{rot}}$}, the entropy density $\eta$, and the temperature $\theta$. The dependence of $\eta$ on the other thermodynamic variables is often called the \textit{constitutive equation for the entropy density}. The relation between $\theta$ and $e$ is the \textit{caloric equation of state}; here it is given by the \textit{equipartition of energy}, i.e.\ energy is shared equally among the translational and rotational degrees of freedom independently of the state of the system. The relations between $\theta$ and $\rho$ are the \textit{thermal equations of state}; here they reduce to the \textit{ideal-gas law}, i.e.\ the pressure is proportional to the product of the density and the temperature of the system. \cla{Throughout this section all constitutive relations are evaluated at the Maxwellian $f^{(0)}$ of \eqref{eq:maxwellian}, which is the local equilibrium distribution for the BGK kinetic model \eqref{eq:BGK} of \Cref{sec:bgk}.}
\uz{\subsection{Translational caloric and thermal equations of state}\label{sec:eos:caloric_thermal}
At equilibrium the rigid-rod translational equipartition of $f^{(0)}$ identifies the thermodynamic temperature $\theta$ through the translational caloric equation of state
\begin{equation}
	\label{eq:equipartition_energy}
	e_{\mathrm{tr}}(\theta) \;=\; \tfrac{3}{2}\,R_s\,\theta,
\end{equation}
the $3/2$ factor reflecting the three translational degrees of freedom of the rigid-rod molecules of \Cref{sec:vanishing_girth}; equivalently
\begin{equation}
	\label{def:temperature}
	\theta \;=\; \frac{2}{3\,R_s}\,e_{\mathrm{tr}}.
\end{equation}
The thermodynamic (gas) pressure $p$ arises from the translational momentum flux of $f^{(0)}$, $\mathbb{T}^{(0)} = -\tfrac{2}{3}\rho\,e_{\mathrm{tr}}\,\mat{I} = -p\mat{I}$, and is conjugate to the specific volume in the translational component,
\begin{equation}
	\label{def:pressure}
	p \;=\; \rho^2\,\partial_{\rho}\,e_{\mathrm{tr}}\bigl|_\eta.
\end{equation}
Combining \eqref{eq:equipartition_energy} with \eqref{def:pressure} recovers the ideal-gas law
\begin{equation}
	\label{eq:ideal_gas_law}
	p \;=\; R_s\,\rho\,\theta.
\end{equation}
\fix{The rotational contribution to $\psi$ is supplied by $e_{\mathrm{rot}}(\theta,\nabla_{\vec{x}}\widehat{\vec{\nu}})$ of \Cref{lemma:erot_oseen_frank}.}
}

\cla{\begin{lemma}[Oseen--Frank form of $e_{\mathrm{rot}}$ at strong alignment]\label{lemma:erot_oseen_frank}
Let $f$ be a distribution sharing the macroscopic density $\rho$, bulk velocity $\vec{u}$, and orientational distribution $\lambda(\vec{\nu},t)$ with the strongly-aligned Maxwellian \eqref{eq:maxwellian} (uniaxial about the director $\widehat{\vec{\nu}}(\vec{x})\in\mathbb{S}^2$, $s\to 1$ in the order-parameter notation of \Cref{corollary:kuramoto}). Then the rotational internal energy admits the Oseen--Frank form
\begin{equation}
	\label{eq:oseen_frank}
	e_{\mathrm{rot}} \;=\; \frac{C_{m,\ell}}{2\fix{m}}\,\cchevrons{\bigl|(\nabla_{\vec{x}}\widehat{\vec{\nu}})\vec{V}\bigr|^2} \;=\; -\frac{C_{m,\ell}}{2\fix{m}\rho}\,\mathrm{tr}\bigl((\nabla_{\vec{x}}\widehat{\vec{\nu}})^T\,\mathbb{T}\,\nabla_{\vec{x}}\widehat{\vec{\nu}}\bigr),
\end{equation}
where $\vec{V}=\vec{v}-\vec{u}$ is the peculiar velocity and $\mathbb{T} = -\rho\cchevrons{\vec{V}\otimes\vec{V}}$ is the Cauchy stress.
\end{lemma}
\begin{proof}
	\uz{In the strongly aligned regime, the orientation variable has a dual microscopic and macroscopic nature $\cchevrons{\vec{\nu}} = s\,\widehat{\vec{\nu}}$ with $s\to 1$, so at leading order the orientation of each molecule coincides with the value of the director field at the molecule's position, i.e.\ $\vec{\nu}(t) = \widehat{\vec{\nu}}(\vec{x}(t),t)$. Under the assumption that the co-rotational time derivative of $\vec{\nu}$ vanishes, the total time derivative of the nematic director from a microscopic viewpoint reads
\begin{equation}
	\dot{\vec{\nu}} \;=\; \partial_t\widehat{\vec{\nu}} + (\nabla_{\vec{x}}\widehat{\vec{\nu}})\,\vec{v},
\end{equation}
yet, since $\widehat{\vec{\nu}}$ is a macroscopic field, it can be moved outside the average, so
\begin{equation}
	\cchevrons{\dot{\vec{\nu}}} \;=\; \partial_t\widehat{\vec{\nu}} + (\nabla_{\vec{x}}\widehat{\vec{\nu}})\,\cchevrons{\vec{v}} \;=\; \partial_t\widehat{\vec{\nu}} + (\nabla_{\vec{x}}\widehat{\vec{\nu}})\,\vec{u},
	\qquad\text{hence}\qquad
	\dot{\vec{\nu}} - \cchevrons{\dot{\vec{\nu}}} \;=\; (\nabla_{\vec{x}}\widehat{\vec{\nu}})\,\vec{V}.
\end{equation}
By definition of the peculiar conjugate momentum $\vec{\Sigma}=\vec{\varsigma}-\vec{\sigma}$, \fix{which under the micro--macro identification carries the same magnitude as the peculiar intrinsic angular momentum, $\abs{\vec{\Lambda}} = \abs{\vec{\Sigma}}$ by \Cref{rmk:peculiar_angular_momentum},} the rotational energy is the variance of $\dot{\vec{\nu}}$, so
\begin{equation}
	e_{\mathrm{rot}} \;=\; \frac{1}{2 \fix{m} C_{m,\ell}}\,\cchevrons{\abs{\vec{\Sigma}}^2} \;=\; \frac{C_{m,\ell}}{2\fix{m}}\,\cchevrons{\abs{\dot{\vec{\nu}} - \cchevrons{\dot{\vec{\nu}}}}^2} \;=\; \frac{C_{m,\ell}}{2\fix{m}}\,\cchevrons{\bigl|(\nabla_{\vec{x}}\widehat{\vec{\nu}})\vec{V}\bigr|^2},
\end{equation}
which is the first equality of \eqref{eq:oseen_frank}. The second identity follows from $\cchevrons{|(\nabla_{\vec{x}}\widehat{\vec{\nu}})\vec{V}|^2} = -\rho^{-1}\,\mathrm{tr}((\nabla_{\vec{x}}\widehat{\vec{\nu}})^T\,\mathbb{T}\,\nabla_{\vec{x}}\widehat{\vec{\nu}})$ via $\mathbb{T} = -\rho\cchevrons{\vec{V}\otimes\vec{V}}$ (cf.~\cite[Section~5]{farrell}).}
\end{proof}
We assume henceforth that any dependence of $e_{\mathrm{rot}}$ on $(\rho,\mathbb{T})$ enters via the translational energy $e_{\mathrm{tr}}$.}

\os{
\subsection{Entropy balance and rate of entropy production inferred from the structure of the Helmholtz free energy and the balance laws}
\label{sec:entropy_balance_and_rate}

In this section we infer the entropy balance and the structure of the entropy production rate from the fundamental thermodynamic relation for the Helmholtz free energy obtained in the previous sections. \fix{The macroscopic Helmholtz free energy is the sum of the kinetic internal energy $\widetilde{e} = e_{\mathrm{tr}} + e_{\mathrm{rot}}$ and the temperature-weighted entropy,}
\begin{equation}
	\label{eq:psi_decomposition}
	\psi(\theta,\rho,\nabla_{\vec{x}}\widehat{\vec{\nu}}) \;=\; e_{\mathrm{tr}}(\theta) \,+\, e_{\mathrm{rot}}(\theta,\nabla_{\vec{x}}\widehat{\vec{\nu}}) \,-\, \theta\,\eta(\theta,\rho),
\end{equation}
\fix{with the director-gradient dependence carried by the rotational specific energy $e_{\mathrm{rot}}$ of \Cref{lemma:erot_oseen_frank}.}
This form of the Helmholtz free energy, treated as a fundamental thermodynamic relation, allows us to infer the local entropy balance. Recalling the definition of the Helmholtz free energy, internal energy, and entropy,
\begin{equation}
	\psi = e-\theta\eta,
\end{equation}
together with the constitutive relation \eqref{eq:helmholtz_decomposition} and \fix{the rotational Oseen--Frank energy of \Cref{lemma:erot_oseen_frank}},
\begin{equation}
	\psi = \psi(\theta,\rho, \nabla_{\vec{x}}\widehat{\vec{\nu}} ),
\end{equation}
applying the material time derivative to both expressions and equating them yields
\begin{align}
\partial_\theta{\psi}\,\frac{d\theta}{dt} + \partial_\rho{\psi}\, \frac{d\rho}{dt} + \bigl(\partial_{\nabla_{\vec{x}}\widehat{\vec{\nu}}}\psi\bigr):\frac{d}{dt}\nabla_{\vec{x}}\widehat{\vec{\nu}} = \frac{d e}{dt} - \frac{d\theta}{dt}\eta - \theta\frac{d\eta}{dt}.
\end{align}
We employ the thermodynamic identity $\partial_\theta{{\psi}} = -\eta$, which follows from considering the Helmholtz free energy as the Legendre transform of the energy with respect to the entropy, and we define the thermodynamic pressure in the usual manner as
\begin{equation}
\uz{\label{eq:thermodynamic_pressure}}
p = \rho^2 {\partial_\rho{\psi} }.
\end{equation}
\uz{Substituting} the time derivatives of density and energy from the mass and energy balances \uz{\eqref{eq:continuity_eq} and} \fix{\eqref{eq:energy_balance_law}, the latter source-free under the micro--macro identification by \Cref{rmk:mu_source_in_energy},} yields, after multiplying by $\rho$, the following identity:
	\begin{equation}
		\rho\,\theta\,\frac{d\eta}{dt} \;=\; (\mathbb{T}+p\mat{I})\!:\!\nabla_{\vec{x}}\vec{u} \,+\, \fix{\frac{1}{mC_{m,\ell}}}\,\mathbb{M}\!:\!\nabla_{\vec{x}}\vec{\mu} \,-\, \nabla_{\vec{x}}\!\cdot\!\fix{\vec{Q}} \,-\, \rho\,\bigl(\partial_{\nabla_{\vec{x}}\widehat{\vec{\nu}}}\psi\bigr):\frac{d}{dt}\nabla_{\vec{x}}\widehat{\vec{\nu}}.
	\end{equation}
Using $d(\nabla_{\vec{x}}\widehat{\vec{\nu}})/dt = \nabla_{\vec{x}}(d\widehat{\vec{\nu}}/dt) - (\nabla_{\vec{x}}\widehat{\vec{\nu}})\,\nabla_{\vec{x}}\vec{u}$ 
	yields
\begin{align}
	\bigl(\partial_{\nabla_{\vec{x}}\widehat{\vec{\nu}}}\psi\bigr):\frac{d}{dt}\nabla_{\vec{x}}\widehat{\vec{\nu}} = -(\nabla_{\vec{x}}\widehat{\vec{\nu}})^T\left(\partial_{\nabla_{\vec{x}}\widehat{\vec{\nu}}}\psi\right) :\nabla_{\vec{x}}\vec{u} + \partial_{\nabla_{\vec{x}}\widehat{\vec{\nu}}}\psi:\nabla_{\vec{x}}\dot{\widehat{\vec{\nu}}}\label{eq:intermediate_exp_e_rot}
\end{align}	
\uz{The molecules under consideration have vanishing co-rotational time derivative, i.e.\ $\overset{\circ}{\vec{\nu}} = \vec{0}$, the same microscopic kinematic identity invoked in the proof of \Cref{lemma:erot_oseen_frank}. Since each orientation then co-rotates rigidly with the same local field, the director inherits this rigid rotation:
\begin{equation}
\dot{\widehat{\vec{\nu}}} = \partial_t\widehat{\vec{\nu}} + (\nabla_{\vec{x}}\widehat{\vec{\nu}})\,\vec{u} = \dot{\cchevrons{\vec{\nu}}} = \cchevrons{\vec{\omega}\times \vec{\nu}} = \cchevrons{\vec{\omega}} \times \widehat{\vec{\nu}},
\end{equation}
where the second identity follows from the same discussion as in \Cref{lemma:erot_oseen_frank}.
}
Thus, we can rewrite the last term of \eqref{eq:intermediate_exp_e_rot} as

\begin{align}
\partial_{\nabla_{\vec{x}}\widehat{\vec{\nu}}}\psi:\nabla_{\vec{x}}\dot{\widehat{\vec{\nu}}} &= \left(\widehat{\vec{\nu}}\times \partial_{\nabla_{\vec{x}}\widehat{\vec{\nu}}}\psi\right):\nabla_{\vec{x}}\uz{\cchevrons{\vec{\omega}}} +
\uz{\cchevrons{\vec{\omega}}}\cdot\left( \nabla_{\vec{x}}\widehat{\vec{\nu}}\dotcross \partial_{\nabla_{\vec{x}}\widehat{\vec{\nu}}}\psi \right),
\end{align}
where we define the cross product between a vector and a tensor by $(\vec a\times\mathbb{B})_{ij} = \varepsilon_{ikl}\vec{a}_k\mathbb{B}_{lj}$ and the dotcross product between two tensors by $(\mathbb{A}\dotcross\mathbb{B})_i = \varepsilon_{ijk}\mathbb{A}_{jl}\mathbb{B}_{kl}$.
For the specific form of the dependence of $\psi$ on $\nabla_{\vec{x}}\widehat{\vec{\nu}}$ given by \uz{\eqref{eq:oseen_frank}}, the dotcross product is identically zero.

\uz{It remains to relate the bulk angular velocity $\cchevrons{\vec{\omega}}$ appearing above to the macroscopic state variables. By definition \eqref{eq:eta_def} the bulk intrinsic angular momentum is $\vec{\mu} = \cchevrons{\vec{\nu}\times\vec{\varsigma}}$, and $\vec{\varsigma} = \mat{B}\dot{\vec{\nu}} = C_{m,\ell}\,\dot{\vec{\nu}}$, so that using $\dot{\vec{\nu}} = \vec{\omega}\times\vec{\nu}$ we obtain, molecule by molecule,
\begin{equation}
	\vec{\nu}\times\vec{\varsigma} \;=\; C_{m,\ell}\,\vec{\nu}\times(\vec{\omega}\times\vec{\nu}) \;=\; C_{m,\ell}\bigl(\vec{\omega} - \vec{\nu}(\vec{\nu}\cdot\vec{\omega})\bigr) \;=\; C_{m,\ell}\,\vec{\omega},
\end{equation}
where the last equality holds because the axial spin $\vec{\nu}\cdot\vec{\omega}$ is not a degree of freedom for the vanishing-girth molecules under consideration, since the inertia tensor \eqref{eq:inertia} has a null direction along $\vec{\nu}$, so we may normalise $\vec{\omega}\cdot\vec{\nu} = 0$. Averaging then gives the desired relationship,
\begin{equation}
	\label{eq:mu_omega0}
	\vec{\mu} \;=\; C_{m,\ell}\,\cchevrons{\vec{\omega}},
	\qquad\text{hence}\qquad
	\nabla_{\vec{x}}\cchevrons{\vec{\omega}} \;=\; \frac{1}{C_{m,\ell}}\,\nabla_{\vec{x}}\vec{\mu},
\end{equation}
since $C_{m,\ell}$ is a constant fixed by the molecular geometry.

As a consequence, we obtain}

	\begin{align}
		\label{eq:entropy_production_constitutive}
		\rho\,\theta\,\frac{d\eta}{dt} &= \left[\mathbb{T} + p\mat{I} + \rho \,(\nabla_{\vec{x}}\widehat{\vec{\nu}})^T \partial_{\nabla_{\vec{x}}\widehat{\vec{\nu}}} \psi\right]\!:\!\nabla_{\vec{x}}\vec{u} + \fix{\frac{1}{mC_{m,\ell}}}\left[\mathbb{M}-\fix{m\rho}\!\left(\widehat{\vec{\nu}}\times \partial_{\nabla_{\vec{x}}\widehat{\vec{\nu}}} \psi\right)\right]\!:\!\nabla_{\vec{x}}\vec{\mu}\\						   &- \nabla_{\vec{x}}\!\cdot\!\fix{\vec{Q}}.\nonumber
	\end{align}
Dividing the above identity by $\theta$ and rearranging yields the following form of the entropy balance:
\begin{align}
		\rho\,\frac{d\eta}{dt} &= \frac{1}{\theta}\left\{ \left[\mathbb{T} + p\mat{I} + \rho \,(\nabla_{\vec{x}}\widehat{\vec{\nu}})^T \partial_{\nabla_{\vec{x}}\widehat{\vec{\nu}}} \psi\right]\!:\!\nabla_{\vec{x}}\vec{u} + \fix{\frac{1}{mC_{m,\ell}}}\left[\mathbb{M}-\fix{m\rho}\!\left(\widehat{\vec{\nu}}\times \partial_{\nabla_{\vec{x}}\widehat{\vec{\nu}}} \psi\right)\right]\!:\!\nabla_{\vec{x}}\vec{\mu}\right\} \\ & -\fix{\vec{Q}}
		\cdot\frac{\nabla_{\vec{x}}\theta}{\theta^2} - \nabla_{\vec{x}}\!\cdot\!\left(\frac{\fix{\vec{Q}}}{\theta}\right).\nonumber
\end{align}
%

\uz{Since $e_{\mathrm{tr}}$ and $\eta$ are $\nabla_{\vec{x}}\widehat{\vec{\nu}}$-independent, the decomposition \eqref{eq:psi_decomposition} gives \fix{$\partial_{\nabla_{\vec{x}}\widehat{\vec{\nu}}}\psi = \partial_{\nabla_{\vec{x}}\widehat{\vec{\nu}}} e_{\mathrm{rot}}$, which carries the Frank modulus $K = K_{\mathrm{kin}} = C_{m,\ell}\,p/m$ via \Cref{lemma:erot_oseen_frank}, recovered at zeroth order in \Cref{sec:zero}.}}

}

\section{Zeroth-order hydrodynamics}
\label{sec:zero}
\cla{We now consider the hydrodynamic regime close to equilibrium. The standard approach for the Boltzmann equation would be to perform a zeroth-order Chapman--Enskog expansion around the Maxwellian distribution, which would yield the compressible Euler equations as the hydrodynamic limit \cite{farrell, curtissI, curtissV}. Instead, we evaluate only the quantities appearing in \Cref{sec:eos} using the Maxwellian distribution \eqref{eq:maxwellian} and proceed with a closure inspired by linear irreversible thermodynamics, similarly to \cite{malek}.}

We begin by computing $e_{\mathrm{tr}}$ and $e_{\mathrm{rot}}$ using the Maxwellian distribution \eqref{eq:maxwellian}:
\begin{align}
	\uz{\mathbb{T}^{(0)}} &= \uz{-}\int m(\vec{V}\otimes \vec{V}) \uz{f^{(0)}} \,d\vec{\Xi} = \uz{-}\frac{2}{3} \rho e_{\mathrm{tr}}\uz{^{(0)}} \mat{I} = \uz{-p_{\mathrm{tr,rot}}} \mat{I},\\
	e^{(0)}_{\mathrm{rot}} &= - \frac{C_{m,\ell}}{2\fix{m}\rho} \mathrm{tr}\left((\nabla_{\vec{x}}\widehat{\vec{\nu}})^T \uz{\mathbb{T}^{(0)}} \,\nabla_{\vec{x}}\widehat{\vec{\nu}}\right) = \frac{C_{m,\ell}}{2\fix{m}}\frac{\uz{p_{\mathrm{tr,rot}}}}{\rho} \mathrm{tr}\left((\nabla_{\vec{x}}\widehat{\vec{\nu}})^T \nabla_{\vec{x}}\widehat{\vec{\nu}}\right) = \frac{C_{m,\ell}}{3\fix{m}}e_{\mathrm{tr}}\uz{^{(0)}}\, \nabla_{\vec{x}}\widehat{\vec{\nu}}\,:\, \nabla_{\vec{x}}\widehat{\vec{\nu}},
\end{align}
where $p_{\mathrm{tr,rot}}=R_s\rho\theta$ is the translational--rotational thermal pressure.

\uz{\begin{corollary}[Entropy production at the Maxwellian]
	\label{cor:zeroth_entropy_production}
	Substituting the entropy balance \eqref{eq:entropy_balance_law}, $\rho\,d\eta/dt = \xi_\tau - \nabla_{\vec{x}}\!\cdot\!\vec{\Phi}$, into \uz{the entropy-production identity \eqref{eq:entropy_production_constitutive}} gives
	\begin{equation}
		\label{eq:entropy_production_explicit}
		\theta\,\xi_\tau \;=\; \left[\mathbb{T} + p\mat{I} + \uz{\rho \,(\nabla_{\vec{x}}\widehat{\vec{\nu}})^T\partial_{\nabla_{\vec{x}}\widehat{\vec{\nu}}} \psi}\right]\!:\!\nabla_{\vec{x}}\vec{u} + \fix{\frac{1}{mC_{m,\ell}}}\left[\mathbb{M}-\fix{m\rho}\!\left(\widehat{\vec{\nu}}\times \partial_{\nabla_{\vec{x}}\widehat{\vec{\nu}}} \psi\right)\right]\!:\!\nabla_{\vec{x}}\vec{\mu} - \nabla_{\vec{x}}\!\cdot\!\fix{\vec{Q}} + \theta\,\nabla_{\vec{x}}\!\cdot\!\vec{\Phi}.
	\end{equation}
	The entropy flux $\vec{\Phi} = -\frac{\rho k_B}{m}\,\cchevrons{\vec{V}\log f}$ of \eqref{defn:entropy_flux_and_production} vanishes at the Maxwellian $f^{(0)}$ by parity, since $\vec{V}\log f^{(0)}$ is odd in $\vec{V}$, so $\vec{\Phi}^{(0)} = \vec{0}$. The BGK operator \eqref{eq:BGK} vanishes at $f^{(0)}$ as well, hence
	\begin{equation}\label{eq:zeroth_entropy_production_zero}
		\xi_\tau^{(0)} \;=\; 0,
	\end{equation}
	and \eqref{eq:entropy_production_explicit} reduces at zeroth order, together with \fix{$\vec{Q}^{(0)} = \vec{0}$}, to
	\begin{equation}
		\label{eq:closure_zeroth_constraint}
		0 \;=\; \left[\mathbb{T} + p\mat{I} + \uz{\rho \,(\nabla_{\vec{x}}\widehat{\vec{\nu}})^T\partial_{\nabla_{\vec{x}}\widehat{\vec{\nu}}} \psi}\right]\!:\!\nabla_{\vec{x}}\vec{u} + \fix{\frac{1}{mC_{m,\ell}}}\left[\mathbb{M}-\fix{m\rho}\!\left(\widehat{\vec{\nu}}\times \partial_{\nabla_{\vec{x}}\widehat{\vec{\nu}}} \psi\right)\right]\!:\!\nabla_{\vec{x}}\vec{\mu}.
	\end{equation}
\end{corollary}}

Following \cite[Section 6]{farrell}, we observe that a valid choice of constitutive relations that satisfies \eqref{eq:closure_zeroth_constraint} identically is
\begin{align}
	\label{eq:closure_T_zeroth}
	\mathbb{T}^{(0)} &= -p\mat{I} - \uz{\rho\, (\nabla_{\vec{x}}\widehat{\vec{\nu}})^T\partial_{\nabla_{\vec{x}}\widehat{\vec{\nu}}} \psi}, \\
	\label{eq:closure_M_zeroth}
	\mathbb{M}^{(0)} &= \fix{m\rho} \left(\widehat{\vec{\nu}}\times \partial_{\nabla_{\vec{x}}\widehat{\vec{\nu}}} \uz{\psi}\right).
\end{align}
\uz{The closures \eqref{eq:closure_T_zeroth}--\eqref{eq:closure_M_zeroth} are precisely those that make the two square brackets in \eqref{eq:closure_zeroth_constraint} of \Cref{cor:zeroth_entropy_production} vanish identically.}

\uz{Since $e_{\mathrm{tr}}$ and $\eta$ are $\nabla_{\vec{x}}\widehat{\vec{\nu}}$-independent, only $e_{\mathrm{rot}}$ contributes to $\partial_{\nabla_{\vec{x}}\widehat{\vec{\nu}}}\psi$. At the strongly-aligned Maxwellian, the $\nabla_{\vec{x}}\widehat{\vec{\nu}}$-dependent part of $\rho\,e_{\mathrm{rot}}^{(0)}$ assembles into
\begin{equation}
		\label{eq:tilde_erot_zeroth}
	\rho\,e_{\mathrm{rot}}^{(0)} \;=\; \frac{K_{\mathrm{kin}}}{2}\,\abs{\nabla_{\vec{x}}\widehat{\vec{\nu}}}^2 \;+\; (\text{terms independent of }\nabla_{\vec{x}}\widehat{\vec{\nu}}),
\end{equation}
with \uz{$K_{\mathrm{kin}} = C_{m,\ell}\,p_{\mathrm{tr,rot}}/m$} from \eqref{eq:oseen_frank} of \Cref{lemma:erot_oseen_frank} evaluated at \uz{the Maxwellian translational moment $-\rho\cchevrons{\vec{V}\otimes\vec{V}}^{(0)} = -p_{\mathrm{tr,rot}}\mat{I}$} (see \Cref{rmk:Kkin_emerges} below). Substituting \eqref{eq:tilde_erot_zeroth} into \eqref{eq:closure_T_zeroth}--\eqref{eq:closure_M_zeroth} gives}
\begin{align}
	\label{eq:closures_explicit_K}
	\mathbb{T}^{(0)} \;=\; -p\mat{I} \;-\; \fix{K_{\mathrm{kin}}}\,\uz{\bigl(\nabla_{\vec{x}}\widehat{\vec{\nu}}\odot\nabla_{\vec{x}}\widehat{\vec{\nu}}\bigr)}, \qquad
	\mathbb{M}^{(0)} \;=\; \fix{m\,K_{\mathrm{kin}}}\,\bigl(\widehat{\vec{\nu}}\times\nabla_{\vec{x}}\widehat{\vec{\nu}}\bigr),\fix{\footnote{In components $\mathbb{M}^{(0)}_{ij} = m K\,\varepsilon_{jkl}\,\widehat{\nu}_k\,\partial_i\widehat{\nu}_l$, the first index being the flux index, consistently with \eqref{eq:couple_stress}, the divergence $(\nabla_{\vec{x}}\cdot\mathbb{M})_j = \partial_i\mathbb{M}_{ij}$ and the contraction $\mathbb{M}:\nabla_{\vec{x}}\vec{\mu} = \mathbb{M}_{ij}\partial_i\mu_j$.}}
\end{align}
\uz{where $(\nabla_{\vec{x}}\widehat{\vec{\nu}}\odot\nabla_{\vec{x}}\widehat{\vec{\nu}})_{ij} := \partial_i\widehat{\vec{\nu}}_k\,\partial_j\widehat{\vec{\nu}}_k = \bigl((\nabla_{\vec{x}}\widehat{\vec{\nu}})^T\nabla_{\vec{x}}\widehat{\vec{\nu}}\bigr)_{ij}$,} so the Ericksen stress, \uz{now in its classical tensorial form,} and the couple stress at zeroth order carry the \fix{Frank modulus $K = K_{\mathrm{kin}}$}.

We can now substitute the closures \eqref{eq:closure_T_zeroth} and \eqref{eq:closure_M_zeroth} together with \eqref{eq:tilde_erot_zeroth} into \eqref{eq:balance_laws} to obtain the macroscopic balance laws at zeroth order, which are given by
\begin{subequations}
	\label{eq:balance_laws_zeroth}
	\begin{align}
		&\partial_t\rho + \nabla_{\vec{x}}\cdot(\rho\vec{u}) = 0, &&\\
		&\rho\left[\partial_t\vec{u} + (\nabla_{\vec{x}}\vec{u})\vec{u}\right] - \nabla_{\vec{x}}\cdot\mathbb{T}\uz{^{(0)}} = 0, && \mathbb{T}\uz{^{(0)}} = -p\mat{I} - \uz{K\!\left(\nabla_{\vec{x}}\widehat{\vec{\nu}} \odot \nabla_{\vec{x}}\widehat{\vec{\nu}}\right)}, \\
		&\rho\bigl[\partial_t\vec{\mu} + (\nabla_{\vec{x}}\vec{\mu})\vec{u}\bigr] \fix{-} \nabla_{\vec{x}}\cdot\mathbb{M}\uz{^{(0)}} = 0, && \vec{\mu} = C_{m,\ell}\,\widehat{\vec{\nu}}\times \dot{\widehat{\vec{\nu}}},\quad \mathbb{M}\uz{^{(0)}} = \fix{m K}\!\left(\widehat{\vec{\nu}}\times\nabla_{\vec{x}}\widehat{\vec{\nu}}\right),\\
		&\fix{\rho\bigl[\partial_t \widetilde{e} + \vec{u}\cdot\nabla_{\vec{x}} \widetilde{e}\bigr]} - \mathbb{T}\uz{^{(0)}}:\nabla_{\vec{x}}\vec{u} - \fix{\tfrac{1}{mC_{m,\ell}}}\mathbb{M}\uz{^{(0)}}:\nabla_{\vec{x}}\vec{\mu} = 0, && \fix{\vec{Q}^{(0)} = \vec{0}},\quad \fix{p = R_s\,\rho\,\theta},\quad \uz{\theta = \tfrac{2}{3\,R_s}\,e_{\mathrm{tr}}},\\
		&\rho\bigl[\partial_t \eta + \vec{u}\cdot\nabla_{\vec{x}} \eta\bigr] = 0, && \cla{\xi_\tau^{(0)}} = 0.
	\end{align}
\end{subequations}
\fix{Here $p = \rho^2\partial_\rho\psi = R_s\rho\theta$ is the thermodynamic pressure, coinciding with the translational--rotational (thermal) ideal-gas pressure $p_{\mathrm{tr,rot}} = R_s\rho\theta$ because the constitutive free energy \eqref{eq:psi_decomposition} carries no density-dependent interaction term.}

We refer to \eqref{eq:balance_laws_zeroth} as the \emph{Euler--Ericksen system}.

\fix{\begin{remark}[Frank modulus $K = K_{\mathrm{kin}}$]\label{rmk:Kkin_emerges}
The kinematic contribution $K_{\mathrm{kin}} = C_{m,\ell}\,p/m$, identified by \cite{farrell}, emerges from $\rho e_{\mathrm{rot}}^{(0)}$ via \Cref{lemma:erot_oseen_frank}. At the zeroth-order stress $\mathbb{T}^{(0)} = -p\mat{I}$ we have $\rho e_{\mathrm{rot}}^{(0)} = (C_{m,\ell}p/2m)|\nabla_{\vec{x}}\widehat{\vec{\nu}}|^2$, so the prefactor of $|\nabla\widehat{\vec{\nu}}|^2$ in $\rho e_{\mathrm{rot}}^{(0)}$ is $K_{\mathrm{kin}}/2$, and the Ericksen and couple stresses at zeroth order carry $K = K_{\mathrm{kin}}$, \eqref{eq:closures_explicit_K}. This is the natural consequence of working with the Helmholtz decomposition \eqref{eq:psi_decomposition} throughout \eqref{eq:entropy_production_constitutive}--\eqref{eq:balance_laws_zeroth}.
\end{remark}}

\section{First-order hydrodynamics}\label{sec:ce}
\cla{The zeroth-order hydrodynamics derived in the previous section, and in particular the balance laws~\eqref{eq:balance_laws_zeroth}, capture only the inviscid, non-dissipative response of the system.} \uz{To recover viscous, heat-conducting and orientational dissipation we perform a first-order Chapman--Enskog expansion of the kinetic equation~\eqref{eq:cfmz} with the BGK collision operator~\eqref{eq:BGK} of~\Cref{sec:bgk}. The anisotropic extension is the subject of ongoing work.} \cla{The crucial point that distinguishes the present calculation from a standard Chapman--Enskog procedure is that the time derivatives appearing in the linearised problem must be substituted using the zeroth-order balance laws~\eqref{eq:balance_laws_zeroth}, not the compressible Euler equations: the Ericksen-stress and couple-stress contributions to $\mathbb{T}^{(0)}$ and $\mathbb{M}^{(0)}$, carrying the \fix{Frank modulus $K = K_{\mathrm{kin}}$} of~\Cref{rmk:Kkin_emerges}, propagate through the substitution rule and produce structurally new terms in $f^{(1)}$. The structure of the calculation parallels the Chapman--Enskog procedure for the BGK operator \cite[Appendix~A]{farrell}, but is generalised to the kinetic theory of ordered fluids, in which the orientational distribution $\lambda(\vec{\nu},t)$ and the Vlasov mean-field force introduce new challenges and torque contributions.} 
\cla{We non-dimensionalise the kinetic equation~\eqref{eq:cfmz} by introducing the Knudsen number $\epsilon = \mathrm{Kn} = \ell_{\mathrm{mfp}}/\ell_{\mathrm{macro}}$ as the formal small parameter and rescale the BGK relaxation time as
\begin{equation}
	\label{eq:ce_rescaling}
	\tau \;=\; \epsilon\,\bar\tau,
\end{equation}
with $\bar\tau = \mathcal{O}(1)$. The kinetic equation~\eqref{eq:cfmz} together with~\eqref{eq:BGK} then reads
\begin{equation}
	\label{eq:ce_kinetic_rescaled}
	\mathcal{D}f \;:=\; \partial_t f + m^{-1}\vec{p}\cdot\nabla_{\vec{x}}f + \mat{B}^{-1}\vec{\varsigma}\cdot\nabla_{\vec{\nu}}f + \vec{\mathcal{V}}\cdot\nabla_{\vec{\varsigma}}f \;=\; -\frac{1}{\epsilon\,\bar\tau}\bigl(f - f^{(0)}\bigr),
\end{equation}
where $\mathcal{D}$ is the streaming operator. We introduce a multi-scale time decomposition,}
\begin{equation}
	\label{eq:ce_multiscale}
	\partial_t = \partial_{t_0} + \epsilon\,\partial_{t_1} + \epsilon^2\,\partial_{t_2} + \cdots,
\end{equation}
and an asymptotic expansion of the distribution,
\begin{equation}
	\label{eq:ce_expansion}
	f = f^{(0)} + \epsilon\, f^{(1)} + \epsilon^2\, f^{(2)} + \cdots.
\end{equation}
The constraint that the macroscopic moments $(\rho, \vec u, e, \vec\mu)$, together with the orientational distribution $\lambda(\vec{\nu}, t)$, are entirely captured by the local Maxwellian $f^{(0)}$ amounts to the solvability condition that, for every $k\geq 1$,
\begin{equation}
	\label{eq:ce_sol_conditions}
	\int f^{(k)}\,\psi\,\dd\vec{\Xi} = 0\quad \text{for every collision invariant } \psi \text{ in~\eqref{eq:collinv}},
\end{equation}
together with the orthogonality of $f^{(k)}$ with respect to the orientational distribution $\lambda$, namely $\int f^{(k)}\,\dd\vec V\,\dd\vec\Sigma = 0$ pointwise in $\vec{\nu}$. These conditions guarantee that the corrections $f^{(k)}$ for $k\geq 1$ encode purely off-equilibrium fluxes and do not redefine the macroscopic state. Substituting~\eqref{eq:ce_multiscale} and~\eqref{eq:ce_expansion} into~\eqref{eq:ce_kinetic_rescaled} and matching powers of $\epsilon$ produces an infinite hierarchy of equations, the first two of which are analysed below.

At leading order in $\epsilon$, the right-hand side of~\eqref{eq:ce_kinetic_rescaled} forces $f^{(0)} - f^{(0)} = 0$, which is satisfied trivially. \cla{The slow evolution of the macroscopic fields $(\rho, \vec u, e, \vec\mu)$ on the time scale $t_0$ is then governed by the zeroth-order balance laws~\eqref{eq:balance_laws_zeroth}:
\begin{subequations}
	\label{eq:ce_zeroth_balance}
	\begin{align}
		&\partial_t\rho + \nabla_{\vec{x}}\cdot(\rho\vec{u}) = 0, &&\\
		&\rho\left[\partial_t\vec{u} + (\nabla_{\vec{x}}\vec{u})\vec{u}\right] - \nabla_{\vec{x}}\cdot\mathbb{T}\uz{^{(0)}} = 0, && \mathbb{T}\uz{^{(0)}} = -p\mat{I} - \uz{K\!\left(\nabla_{\vec{x}}\widehat{\vec{\nu}} \odot \nabla_{\vec{x}}\widehat{\vec{\nu}}\right)}, \\
		&\rho\bigl[\partial_t\vec{\mu} + (\nabla_{\vec{x}}\vec{\mu})\vec{u}\bigr] \fix{-} \nabla_{\vec{x}}\cdot\mathbb{M}\uz{^{(0)}} = 0, && \vec{\mu} = C_{m,\ell}\,\widehat{\vec{\nu}}\times \dot{\widehat{\vec{\nu}}},\quad \mathbb{M}\uz{^{(0)}} = \fix{m K}\!\left(\widehat{\vec{\nu}}\times\nabla_{\vec{x}}\widehat{\vec{\nu}}\right),\\
		&\fix{\rho\bigl[\partial_t \widetilde{e} + \vec{u}\cdot\nabla_{\vec{x}} \widetilde{e}\bigr]} - \mathbb{T}\uz{^{(0)}}:\nabla_{\vec{x}}\vec{u} - \fix{\tfrac{1}{mC_{m,\ell}}}\mathbb{M}\uz{^{(0)}}:\nabla_{\vec{x}}\vec{\mu} = 0, && \fix{\vec{Q}^{(0)} = \vec{0}},\quad \fix{p = R_s\,\rho\,\theta},\quad \theta = \frac{2}{5\,R_s}\,\widetilde{e},\\
		&\rho\bigl[\partial_t \eta + \vec{u}\cdot\nabla_{\vec{x}} \eta\bigr] = 0, && \xi_{\tau}^{(0)} = 0.
	\end{align}
\end{subequations}
We emphasise that~\eqref{eq:ce_zeroth_balance} differs from the standard compressible Euler equations precisely by the Ericksen-stress and couple-stress contributions inherited from the closures~\eqref{eq:closure_T_zeroth}--\eqref{eq:closure_M_zeroth} with the rotational specific energy of \uz{\eqref{eq:entropy_production_constitutive}, obtained by following the path of choosing the Helmholtz free energy as our primary thermodynamic potential}. These are the contributions that will propagate through the time-derivative substitution rule used at the next order, and ultimately produce structurally new terms in $f^{(1)}$.

Collecting the order-$\epsilon^{1}$ terms in~\eqref{eq:ce_kinetic_rescaled} we obtain
\begin{equation}
	\label{eq:ce_first_order_linearised}
	f^{(1)} \;=\; -\bar\tau\,\mathcal{D}f^{(0)}.
\end{equation}
The remainder of this section is devoted to the explicit computation of the streaming derivative $\mathcal{D}f^{(0)}$, the resulting form of $f^{(1)}$, and its first-order moments.}

\uz{\subsection{Expansion independence of the equation of state}\label{sec:ce_expansion_independence}
The solvability conditions~\eqref{eq:ce_sol_conditions} have an immediate consequence for the thermodynamic structure of the model, which we record before turning to the explicit computation of $f^{(1)}$. Following \cite[Section~5]{malek}, we call a macroscopic observable $\varphi[f]$ \emph{expansion independent} when
\begin{equation}
	\label{eq:expansion_independence}
	\varphi[f] = \varphi\Bigl[\textstyle\sum_{i=0}^{k}\epsilon^i f^{(i)}\Bigr]\qquad\text{for every } k\geq 0,
\end{equation}
that is, when the observable is unchanged on every truncation of the Chapman--Enskog expansion~\eqref{eq:ce_expansion}. We write $\varphi^{(k)}$ for the observable evaluated on the truncation $\sum_{i=0}^k\epsilon^i f^{(i)}$. The three results below are, for the present model, the analogues of \cite[Propositions~1--3]{malek}, the only new ingredients being the rotational component and the orientational distribution $\lambda$.

\begin{proposition}[Expansion independence of the conserved fields]\label{prop:expansion_independence_fields}
Under the solvability conditions~\eqref{eq:ce_sol_conditions} the mass density $\rho$, the bulk velocity $\vec u$, the total internal energy $\widetilde{e}$, and the bulk intrinsic angular momentum $\vec\mu$ are expansion independent.
\end{proposition}
\begin{proof}
Each of these fields is a moment of $f$ against a collision invariant of~\eqref{eq:collinv}. For the density,
\begin{equation*}
	\rho^{(k)} = \int m\sum_{i=0}^k\epsilon^i f^{(i)}\,\dd\vec\Xi = \int m f^{(0)}\,\dd\vec\Xi + \sum_{i=1}^k\epsilon^i\int m f^{(i)}\,\dd\vec\Xi = \rho,
\end{equation*}
since every term with $i\geq 1$ vanishes because $m$ is a collision invariant. The same computation with the momentum $\vec p$, with the total energy $\tfrac12(\abs{\vec p}^2/m + \vec\varsigma^T\mat{B}^{-1}\vec\varsigma)$, and with the angular momentum of~\eqref{eq:collinv} shows that $\rho\vec u$, the total energy density $\rho(\tfrac12\abs{\vec u}^2 + \widetilde{e})$ and $\vec\mu$ coincide with their zeroth-order values. As $\rho$ and $\rho\vec u$ are expansion independent, so is $\vec u$, and as the total energy density is then expansion independent, so is the internal energy $\widetilde{e}$.
\end{proof}

\begin{proposition}[Expansion independence of the entropy]\label{prop:expansion_independence_entropy}
Under the solvability conditions~\eqref{eq:ce_sol_conditions} the specific entropy satisfies $\eta^{(1)} = \eta^{(0)} + O(\epsilon^2)$.
\end{proposition}
\begin{proof}
Truncating~\eqref{eq:ce_expansion} at first order and expanding the logarithm gives $\log(f^{(0)} + \epsilon f^{(1)}) = \log f^{(0)} + \epsilon\,f^{(1)}/f^{(0)} + O(\epsilon^2)$, so the entropy density $\rho\eta = -k_B\int f\log f\,\dd\vec\Xi$ of~\eqref{defn:entropy_flux_and_production} reads
\begin{equation*}
	\rho\eta^{(1)} = -k_B\int f^{(0)}\log f^{(0)}\,\dd\vec\Xi - \epsilon\,k_B\int f^{(1)}\log f^{(0)}\,\dd\vec\Xi + O(\epsilon^2).
\end{equation*}
The first integral is $\rho\eta^{(0)}$ of~\eqref{eq:entropy_density_maxwellian}. In the second, $\log f^{(0)}$ of~\eqref{eq:ce_log_maxwellian} is an affine combination of the constant, the peculiar total energy $\tfrac12(m\abs{\vec V}^2 + \abs{\vec\Sigma}^2/C_{m,\ell})$ and the orientational term $\log\lambda(\vec{\nu})$. The constant multiplies $\int f^{(1)}\,\dd\vec\Xi = 0$ by mass solvability, the peculiar total energy multiplies $\int \tfrac12(m\abs{\vec V}^2 + \abs{\vec\Sigma}^2/C_{m,\ell})\,f^{(1)}\,\dd\vec\Xi = 0$ by the expansion independence of the internal energy of \Cref{prop:expansion_independence_fields}, and the orientational term multiplies $\int f^{(1)}\,\dd\vec V\,\dd\vec\Sigma = 0$ at fixed $\vec{\nu}$. Hence $\rho\eta^{(1)} = \rho\eta^{(0)} + O(\epsilon^2)$.
\end{proof}

These two propositions transfer the equation of state of \Cref{sec:eos} from equilibrium to the dissipative first-order state.

\begin{corollary}[Equation of state at first order]\label{cor:eos_first_order}
Under the solvability conditions~\eqref{eq:ce_sol_conditions} the caloric equation of state~\eqref{eq:equipartition_energy} and the thermal equation of state~\eqref{eq:ideal_gas_law} hold on the first-order truncation, and the thermodynamic temperature $\theta$ and the pressure $p$ are expansion independent.
\end{corollary}
\begin{proof}
By \Cref{prop:expansion_independence_fields} the density and the total internal energy are expansion independent, and by \Cref{prop:expansion_independence_entropy} so is the entropy up to $O(\epsilon^2)$. The equilibrium relation $\widetilde{e} = \widetilde{e}(\rho,\eta)$ of \Cref{sec:eos} therefore holds on the first-order truncation, so the temperature $\theta = \partial_\eta \widetilde{e}$ of~\eqref{def:temperature} \uz{and the thermodynamic pressure $p = \rho^2\partial_\rho\psi$ of \eqref{eq:thermodynamic_pressure} retain their zeroth-order values.} \fix{By the Helmholtz decomposition \eqref{eq:psi_decomposition} the pressure coincides with the translational--rotational ideal-gas pressure, $p = p_{\mathrm{tr,rot}} = R_s\rho\theta$, which together with the equipartition law $\theta = \tfrac{2}{5 R_s}\widetilde{e}$ holds not only at equilibrium but on the dissipative first-order state.}
\end{proof}

A comment on the meaning of the order superscripts in the Helmholtz potential is in order. When writing $\psi^{(1)}$ we mean
\begin{equation}
	\label{eq:psi_first_order_convention}
	\psi^{(1)} \;:=\; e_{\mathrm{tr}}^{(1)} \,+\, e_{\mathrm{rot}}^{(1)} \,-\, \theta\,\eta^{(1)}.
\end{equation}
The kinetic components $e_{\mathrm{tr}}$, $e_{\mathrm{rot}}$ and $\eta$ are evaluated on the truncation $f^{(0)} + \epsilon f^{(1)}$\fix{; the constitutive free energy carries no interaction component}. The practical gain of \eqref{eq:psi_first_order_convention} is considerable.
}

\subsection{Streaming derivative}\label{sec:ce_streaming}
We compute the streaming derivative $\mathcal{D}f^{(0)}$ via the chain rule, $\mathcal{D}f^{(0)} = f^{(0)}\,\mathcal{D}\log f^{(0)}$, where, from~\eqref{eq:maxwellian},
\begin{equation}
	\label{eq:ce_log_maxwellian}
	\log f^{(0)} = \log\rho - \frac{5}{2}\log\theta - \frac{m\,\abs{\vec{V}}^2}{2k_B\theta} - \frac{\abs{\vec{\Sigma}}^2}{2C_{m,\ell}k_B\theta} + \log\lambda(\vec{\nu}) + \mathrm{C},
\end{equation}
where $\vec{V} = \vec{v} - \vec{u}$, $\vec{\Sigma} = \vec{\varsigma} - \vec{\sigma}$, $\vec{\sigma} = \cchevrons{\vec{\varsigma}}$, and $\mathrm{C}$ is a constant independent of any physical quantity. We collect the contribution of each term to $\mathcal{D}\log f^{(0)}$. Throughout, we substitute the time derivatives of the macroscopic fields using the zeroth-order balance laws~\eqref{eq:ce_zeroth_balance}.

Acting with $\mathcal{D}$ on $\log\rho$ and using the continuity equation in~\eqref{eq:ce_zeroth_balance},
\begin{align}
	\mathcal{D}\log\rho &= \rho^{-1}\bigl(\partial_{t_0}\rho + \vec{v}\cdot\nabla_{\vec{x}}\rho\bigr) = \rho^{-1}\bigl(-\nabla_{\vec{x}}\cdot(\rho\vec{u}) + \vec{v}\cdot\nabla_{\vec{x}}\rho\bigr) = -\nabla_{\vec{x}}\cdot\vec{u} + \rho^{-1}\vec{V}\cdot\nabla_{\vec{x}}\rho. \label{eq:ce_streaming_rho}
\end{align}

\uz{Acting with $\mathcal{D}$ on $\abs{\vec V}^2 = (\vec v - \vec u)\cdot(\vec v - \vec u)$ and using the momentum balance in~\eqref{eq:ce_zeroth_balance}, that is, the Euler--Ericksen system derived in \Cref{sec:zero} with the Ericksen stress and \fix{the pressure $p = p_{\mathrm{tr,rot}} = R_s\rho\theta$},
\begin{equation}
	\partial_{t_0}\vec{u} = -(\nabla_{\vec{x}}\vec{u})\vec{u} - \rho^{-1}\nabla_{\vec{x}}p - \rho^{-1}\nabla_{\vec{x}}\!\cdot\!\bigl[K\,(\nabla_{\vec{x}}\widehat{\vec{\nu}}\odot\nabla_{\vec{x}}\widehat{\vec{\nu}})\bigr],
\end{equation}
gives
\begin{align}
	\mathcal{D}\abs{\vec{V}}^2 &= -2\vec{V}\cdot\bigl(\partial_{t_0}\vec{u} + (\nabla_{\vec{x}}\vec{u})\vec{v}\bigr) = -2\,\vec V\cdot(\nabla_{\vec{x}}\vec u)\vec V + \frac{2}{\rho}\vec V\cdot\nabla_{\vec{x}}p + \frac{2}{\rho}\vec V\cdot\nabla_{\vec{x}}\!\cdot\!\bigl[K\,(\nabla_{\vec{x}}\widehat{\vec{\nu}}\odot\nabla_{\vec{x}}\widehat{\vec{\nu}})\bigr]. \label{eq:ce_streaming_V2}
\end{align}
The first term is the standard Newtonian shear contribution, the second the full-pressure-gradient contribution, and the third the Ericksen-stress contribution. The thermal part of the pressure gradient will cancel classically against the density-gradient and temperature-gradient sources, as we will see in \Cref{thm:ce_first_order} below. The interaction part combines with the Ericksen divergence into the vector
\begin{equation}
	\label{eq:ce_W_def}
	\vec{W} \;:=\; \nabla_{\vec x}\!\cdot\!\bigl[K\,(\nabla_{\vec x}\widehat{\vec{\nu}}\odot\nabla_{\vec x}\widehat{\vec{\nu}})\bigr]\fix{,}
\end{equation}
\fix{with $K = K_{\mathrm{kin}}$,} which produces in $\mathcal{D}\log f^{(0)}$ the $\vec{V}$-linear source $-(m/(k_B\theta\rho))\,\vec{V}\cdot\vec{W}$. \fix{The following remark shows that $\vec{W}$ is removed by the solvability condition of the first-order problem.}

\fix{\begin{remark}[$\vec{W}$ and the solvability condition]\label{rmk:xiE_solvability}
The source $-(m/(k_B\theta\rho))\,\vec{V}\cdot\vec{W}$ is, at each point, a linear combination of the \emph{momentum collision invariants} of \eqref{eq:collinv}, which the linearised BGK operator annihilates, so it is not inverted into $f^{(1)}$. The first-order solvability condition \eqref{eq:ce_sol_conditions} requires $f^{(1)}$ to carry no hydrodynamic momentum, $\int f^{(1)}\,m\vec{V}\,d\vec{\Xi} = \vec{0}$, and it is this condition that governs the component: were $\vec{W}$ retained, $f^{(1)}$ would carry the spurious momentum $\int f^{(1)}\,m\vec{V}\,d\vec{\Xi} = \bar\tau\,\vec{W}$. With the constitutive free energy reduced to $\psi = \widetilde{e} - \theta\eta$ of \eqref{eq:psi_decomposition} the interaction pressure is absent, and
\begin{equation}
	\label{eq:ce_W_constraint}
	\vec{W} \;=\; \nabla_{\vec x}\!\cdot\!\bigl[K_{\mathrm{kin}}\,(\nabla_{\vec x}\widehat{\vec{\nu}}\odot\nabla_{\vec x}\widehat{\vec{\nu}})\bigr]
\end{equation}
is purely the divergence of the elastic Ericksen stress carried by the zeroth-order momentum balance. It is therefore removed from $f^{(1)}$ by the solvability condition, contributing to neither the first-order heat flux nor the entropy production, and imposing no constraint on the state.
\end{remark}}

\cla{For the temperature, the rigid-rod caloric law $\widetilde{e} = (5/2)\,R_s\,\theta$ of \eqref{eq:equipartition_energy} gives $\partial_{t_0}\theta = (2/(5R_s))\,\partial_{t_0}\widetilde{e}$.} \fix{The energy balance in~\eqref{eq:ce_zeroth_balance} is the balance for the kinetic internal energy $\widetilde{e}$, so
\begin{align}
	\rho\,\partial_{t_0}\widetilde{e} &= -\rho\,\vec u\cdot\nabla_{\vec x}\widetilde{e} + \mathbb{T}^{(0)}\!:\!\nabla_{\vec{x}}\vec{u} + \frac{1}{mC_{m,\ell}}\mathbb{M}^{(0)}\!:\!\nabla_{\vec{x}}\vec{\mu}\nonumber\\
	&= -\rho\,\vec u\cdot\nabla_{\vec x}\widetilde{e} - p_{\mathrm{tr,rot}}\,\nabla_{\vec{x}}\cdot\vec{u} - K\,(\nabla_{\vec{x}}\widehat{\vec{\nu}}\odot\nabla_{\vec{x}}\widehat{\vec{\nu}})\!:\!\nabla_{\vec{x}}\vec{u} + \frac{K}{C_{m,\ell}}\,(\widehat{\vec{\nu}}\times\nabla_{\vec{x}}\widehat{\vec{\nu}})\!:\!\nabla_{\vec{x}}\vec{\mu},
\end{align}
using $\mathbb{T}^{(0)} = -p_{\mathrm{tr,rot}}\mat{I} - K(\nabla_{\vec{x}}\widehat{\vec{\nu}}\odot\nabla_{\vec{x}}\widehat{\vec{\nu}})$ with $p_{\mathrm{tr,rot}} = R_s\rho\theta$,}
\begin{equation}
	\label{eq:ce_streaming_theta}
	\begin{aligned}
	\mathcal{D}\log\theta = -\frac{2}{5}\Bigl[\nabla_{\vec{x}}\!\cdot\!\vec{u} + (K/p_{\mathrm{tr,rot}})\,(\nabla_{\vec{x}}\widehat{\vec{\nu}}\odot\nabla_{\vec{x}}\widehat{\vec{\nu}})\!:\!\nabla_{\vec{x}}\vec{u}\Bigr]
	\fix{+} \frac{2K}{5\,C_{m,\ell}\,R_s\,\rho\,\theta}\,(\widehat{\vec{\nu}}\times\nabla_{\vec{x}}\widehat{\vec{\nu}})\!:\!\nabla_{\vec{x}}\vec{\mu}\\
	+ \theta^{-1}\vec V\cdot\nabla_{\vec x}\theta.
	\end{aligned}
\end{equation}
The Ericksen correction to the adiabatic compression rate of the temperature is now the full tensorial contraction $(K/p_{\mathrm{tr,rot}})(\nabla_{\vec x}\widehat{\vec{\nu}}\odot\nabla_{\vec x}\widehat{\vec{\nu}})\!:\!\nabla_{\vec x}\vec{u}$, sensitive to the orientation of the compression relative to the director texture. In the limit $\nabla_{\vec x}\widehat{\vec{\nu}}\to 0$ it reduces to the classical $-\frac{2}{5}\nabla\cdot\vec u$ of a polyatomic gas with the diatomic adiabatic exponent $\gamma = 7/5$ \footnote{The work \cite{farrell} discusses experimental observations of the deviation from the diatomic adiabatic exponent and conjectures that the Ericksen correction would resolve it.}.}

Acting with $\mathcal{D}$ on $\abs{\vec\Sigma}^2 = (\vec\varsigma - \vec\sigma)\cdot(\vec\varsigma - \vec\sigma)$ and using the angular-momentum balance in~\eqref{eq:ce_zeroth_balance} together with the rigid-rod identification $\vec\sigma = \cchevrons{\vec\varsigma}$ and the aligned regime $s\to 1$, we obtain
\begin{equation}
	\label{eq:ce_streaming_Sigma2}
	\mathcal{D}\abs{\vec\Sigma}^2 = -2\,\vec\Sigma\cdot\bigl[\partial_{t_0}\vec\sigma + (\nabla_{\vec{x}}\vec\sigma)\vec v\bigr] + 2(\vec{\mathcal V}\cdot\vec\Sigma).
\end{equation}
By the alignment identity $\vec\gamma\cdot\dot{\widehat{\vec{\nu}}} = s(\widehat{\vec{\nu}}\cdot\dot{\widehat{\vec{\nu}}}) = 0$\uz{, cf.~\Cref{rmk:vlasov_torque} and \Cref{corollary:kuramoto},} which holds at $s\to 1$, the Vlasov contribution $\vec{\mathcal V}\cdot\vec\Sigma$ projects to zero against any test function depending only on $\abs{\vec V}^2 + \abs{\vec\Sigma}^2/C_{m,\ell}$.

Lastly, since $\lambda$ depends neither on $\vec V$ nor on $\vec\Sigma$,
\begin{equation}
	\label{eq:ce_streaming_lambda}
	\mathcal{D}\log\lambda = \partial_{t_0}\log\lambda + \vec v\cdot\nabla_{\vec{x}}\log\lambda + (\mat{B}^{-1}\vec\varsigma)\cdot\nabla_{\vec{\nu}}\log\lambda.
\end{equation}
In the aligned regime $s\to 1$ inherited from~\Cref{sec:zero}, the Gibbs distribution~\eqref{eq:gibbs} is concentrated about the director $\widehat{\vec{\nu}}$, and~\eqref{eq:ce_streaming_lambda} reduces to
\begin{equation}
	\mathcal{D} \log(\lambda) \approx \frac{s}{k_B\theta}\left[(\mat{B}^{-1}\vec{\varsigma})\cdot\widehat{\vec{\nu}}
	+ \vec{\nu}\cdot\bigl(\partial_{t_0}\widehat{\vec{\nu}}+(\nabla_{\vec{x}}\widehat{\vec{\nu}})\vec v\bigr)\right],
\end{equation}
up to derivatives of the scalar concentration parameter, which are higher order in the strongly aligned closure. At leading order $\vec{\nu}=\widehat{\vec{\nu}}$, the first term vanishes because $\mat{B}^{-1}\vec{\varsigma}\in T_{\vec{\nu}}\mathbb{S}^2$ and the second vanishes by the unit-length constraint on $\widehat{\vec{\nu}}$.
Hence, in the aligned regime, the orientational contribution to $\mathcal{D}\log f^{(0)}$ vanishes, and we are left only with the contributions in~\eqref{eq:ce_streaming_rho},~\eqref{eq:ce_streaming_V2},~\eqref{eq:ce_streaming_theta} and~\eqref{eq:ce_streaming_Sigma2}.

\subsection{First-order correction}
\label{sec:ce_f1}
Combining the contributions of~\Cref{sec:ce_streaming} via $f^{(1)} = -\bar\tau\,f^{(0)}\,\mathcal{D}\log f^{(0)}$ and regrouping the contributions according to the macroscopic gradients on which they are linear, we arrive at the explicit form of $f^{(1)}$.

\begin{theorem}\label{thm:ce_first_order}
\cla{In the aligned regime $s\to 1$, the first-order Chapman--Enskog correction to the local Maxwellian~\eqref{eq:maxwellian} for the single-relaxation BGK operator~\eqref{eq:BGK} is}
\begin{align}
\label{eq:ce_f1_explicit}
	f^{(1)} &= -\cla{\bar\tau}\,f^{(0)}\,\biggl\{\,\underbrace{\frac{m\,(\vec V\!\otimes\!\vec V)^d}{k_B\theta}\!:\!\mathbb D}_{\text{shear}} \;+\; \underbrace{\Bigl(\frac{m\abs{\vec V}^2}{2k_B\theta} - \frac{5}{2}\Bigr)\frac{\vec V\cdot\nabla_{\vec x}\theta}{\theta}}_{\text{translational heat}}+\;\underbrace{\Bigl(\frac{\abs{\vec\Sigma}^2}{2C_{m,\ell}k_B\theta} - 1\Bigr)\frac{\vec V\cdot\nabla_{\vec x}\theta}{\theta}}_{\text{rotational heat}} \nonumber \\
	&+\;\underbrace{\frac{2}{5}\Bigl(\frac{m\abs{\vec V}^2}{3k_B\theta} - \frac{\abs{\vec\Sigma}^2}{2C_{m,\ell}k_B\theta}\Bigr)\nabla_{\vec x}\!\cdot\!\vec u}_{\text{isotropic bulk}}\nonumber\\
	&-\;\underbrace{\uz{\frac{2}{5}(K/p_{\mathrm{tr,rot}})\,\bigl[(\nabla_{\vec x}\widehat{\vec{\nu}}\odot\nabla_{\vec x}\widehat{\vec{\nu}})\!:\!\nabla_{\vec x}\vec u\bigr]\Bigl(\frac{m\abs{\vec V}^2}{2k_B\theta} + \frac{\abs{\vec\Sigma}^2}{2C_{m,\ell}k_B\theta} - \frac{5}{2}\Bigr)}}_{\text{anisotropic bulk}}\nonumber\\
	&\;\fix{+}\;\underbrace{\frac{2\,\cla{K}}{5\,\uz{C_{m,\ell}}\,R_s\,\rho\,\theta}\,(\widehat{\vec{\nu}}\times\nabla_{\vec x}\widehat{\vec{\nu}})\!:\!\nabla_{\vec x}\vec\mu\,\Bigl(\frac{m\abs{\vec V}^2}{2k_B\theta} + \frac{\abs{\vec\Sigma}^2}{2C_{m,\ell}k_B\theta} - \frac{5}{2}\Bigr)}_{\text{couple-stress reactive}}\nonumber\\
	&\;+\;\underbrace{\frac{\Sigma_a\,(\nabla_{\vec x}\vec\sigma)_{ab}\,V_b}{C_{m,\ell}\,k_B\theta}}_{\text{angular advection}}\,\biggr\},
\end{align}
where $(\vec V\!\otimes\!\vec V)^d$ denotes the deviatoric part of the velocity dyad, $\mathbb D = \frac{1}{2}(\nabla_{\vec x}\vec u + (\nabla_{\vec x}\vec u)^T)$ is the symmetric part of the velocity gradient, \fix{$K = K_{\mathrm{kin}} = C_{m,\ell}\,p_{\mathrm{tr,rot}}/m$ is the one-constant Frank modulus from \Cref{rmk:Kkin_emerges}}\uz{, and $p_{\mathrm{tr,rot}} = R_s\rho\theta$ is the translational--rotational (thermal) pressure. The $\vec{V}$-linear source \fix{is removed by the solvability condition, see \eqref{eq:ce_W_constraint} and \Cref{rmk:xiE_solvability}}, so $f^{(1)}$ is consistent with \eqref{eq:ce_sol_conditions}}.
\end{theorem}

\begin{proof}
We start from~\eqref{eq:ce_first_order_linearised}, which expresses $f^{(1)}$ as the streaming derivative of the local Maxwellian,
\begin{equation}
	\label{eq:ce_f1_proof_start}
	f^{(1)} = -\bar\tau\,\mathcal{D}f^{(0)} = -\bar\tau\,f^{(0)}\,\mathcal{D}\log f^{(0)}.
\end{equation}
By~\eqref{eq:ce_log_maxwellian} and the linearity of $\mathcal{D}$,
\begin{align}
	\label{eq:ce_Dlog_decomposition}
	\mathcal{D}\log f^{(0)} &= \mathcal{D}\log\rho - A\,\mathcal{D}\log\theta - \frac{m}{2k_B\theta}\,\mathcal{D}\abs{\vec V}^2 - \frac{1}{2C_{m,\ell}k_B\theta}\,\mathcal{D}\abs{\vec\Sigma}^2 + \mathcal{D}\log\lambda,
	\\
	A &= \frac{5}{2} - \frac{m}{2k_B\theta}\abs{\vec{V}}^2 - \frac{1}{2C_{m,\ell}k_B\theta}\abs{\vec{\Sigma}}^2.
\end{align}
where we used $\mathcal{D}(k_B\theta)^{-1} = (k_B\theta)^{-1}\mathcal{D}\log\theta\cdot(-1)$ and absorbed both temperature contributions inside~\eqref{eq:ce_streaming_theta} below. The orientational term $\mathcal{D}\log\lambda$ vanishes in the regime $s\to 1$ as discussed above. The remaining four terms are given by~\eqref{eq:ce_streaming_rho},~\eqref{eq:ce_streaming_V2},~\eqref{eq:ce_streaming_theta} and~\eqref{eq:ce_streaming_Sigma2}. We now substitute them into~\eqref{eq:ce_Dlog_decomposition} and group by the macroscopic gradient on which each contribution depends.

By~\eqref{eq:ce_streaming_rho},
\begin{equation*}
	\mathcal{D}\log\rho = -\nabla_{\vec x}\!\cdot\!\vec u + \rho^{-1}\vec V\cdot\nabla_{\vec x}\rho.
\end{equation*}
The first term combines with the trace part of $\mathcal{D}\abs{\vec V}^2$ and with $-\frac{5}{2}\mathcal{D}\log\theta$ to give the bulk contribution analysed below. The second term contains the density gradient. \fix{Recalling the thermal equation of state $p = p_{\mathrm{tr,rot}} = R_s\rho\theta$, we have $\rho^{-1}\nabla_{\vec x}p = R_s\theta\,\rho^{-1}\nabla_{\vec x}\rho + R_s\nabla_{\vec x}\theta$, so that the pressure-gradient term in~\eqref{eq:ce_streaming_V2} contributes}
\begin{equation*}
	-\frac{m}{2k_B\theta}\cdot\frac{2}{\rho}\vec V\cdot\nabla_{\vec x}p = -\frac{1}{\theta}\vec V\cdot\nabla_{\vec x}\theta - \rho^{-1}\vec V\cdot\nabla_{\vec x}\rho\fix{,}
\end{equation*}
where we used $m R_s = k_B$. The contribution $-\rho^{-1}\vec V\cdot\nabla_{\vec x}\rho$ exactly cancels the $+\rho^{-1}\vec V\cdot\nabla_{\vec x}\rho$ from $\mathcal{D}\log\rho$, leaving the temperature-gradient piece $-\theta^{-1}\vec V\cdot\nabla_{\vec x}\theta$. The remaining $\vec{V}$-linear source, $-(m/(k_B\theta\rho))\,\vec{V}\cdot\vec{W}$, is now the Ericksen divergence \eqref{eq:ce_W_def} alone, \fix{which is removed by the solvability condition, see \Cref{rmk:xiE_solvability}}.

The advective term $-2\,\vec V\cdot(\nabla_{\vec x}\vec u)\vec V$ in~\eqref{eq:ce_streaming_V2}, multiplied by the prefactor $-\tfrac{m}{2k_B\theta}$ in~\eqref{eq:ce_Dlog_decomposition}, contributes
\begin{equation*}
	\frac{m}{k_B\theta}\,V_iV_j\,\partial_{x_i}u_j = \frac{m}{k_B\theta}\,(\vec V\otimes\vec V):\nabla_{\vec x}\vec u.
\end{equation*}
Splitting $\nabla_{\vec x}\vec u$ into its symmetric part $\mathbb D$ and its antisymmetric part, which vanishes against the symmetric tensor $\vec V\otimes\vec V$, and decomposing $\vec V\otimes\vec V$ into its deviatoric and trace parts, the trace part $\tfrac{1}{3}\abs{\vec V}^2\mat{I}$ couples to $\nabla_{\vec x}\!\cdot\!\vec u$ and is reabsorbed into the bulk term, while the deviatoric piece yields the shear contribution
\begin{equation}
	\frac{m\,(\vec V\otimes\vec V)^d}{k_B\theta}:\mathbb D,
\end{equation}
which is precisely the first term in~\eqref{eq:ce_f1_explicit}.

Combining the temperature-gradient piece $-\theta^{-1}\vec V\cdot\nabla_{\vec x}\theta$ obtained above with the trace part of the shear contribution gives
\begin{equation}
	-\frac{1}{\theta}\vec V\cdot\nabla_{\vec x}\theta + \frac{m\abs{\vec V}^2}{3k_B\theta}\,(\nabla_{\vec x}\!\cdot\!\vec u);
\end{equation}
the second piece is reabsorbed into the bulk contribution below.

Combining the last term in \eqref{eq:ce_streaming_theta} with the $-\theta^{-1}\vec{V}\cdot \nabla_{\vec{x}}\theta$ term obtained above gives
\begin{equation}
	-\left[ \frac{5}{2} - \frac{m}{2k_B\theta}\abs{\vec{V}}^2 + 1 - \frac{1}{2C_{m,\ell}k_B\theta}\abs{\vec{\Sigma}}^2 \right] \vec{V}\cdot \frac{\nabla_{\vec{x}} \theta}{\theta},
\end{equation}
which yields the translational and rotational heat contributions.

It remains to evaluate the rotational-gradient contribution:
	\begin{align}
		\mathcal{D}\,\abs{\vec{\Sigma}}^2 = 2\vec{\Sigma}\cdot \bigl(\partial_t \vec{\Sigma} + (\nabla_{\vec{x}} \vec{\Sigma}) \vec{v}\bigr) = \fix{-\,\frac{2mK}{\rho}\,\vec\Sigma\cdot(\mat{I}-\widehat{\vec{\nu}}\otimes\widehat{\vec{\nu}})\Delta\widehat{\vec{\nu}} - 2\,\vec\Sigma\cdot(\vec{\mu}\times\dot{\widehat{\vec{\nu}}})} - 2\,\vec{\Sigma}\cdot (\nabla_{\vec{x}} \vec{\sigma})\, \vec{V},
	\end{align}
where we have used \eqref{eq:ce_zeroth_balance} \cla{with the couple stress $\mathbb{M}^{(0)} = \fix{mK}(\widehat{\vec{\nu}}\times\nabla_{\vec x}\widehat{\vec{\nu}})$} \fix{and the micro--macro identification $\vec{\sigma} = \vec{\mu}\times\widehat{\vec{\nu}}$}.

\fix{The first term above is a second director gradient, $\Delta\widehat{\vec{\nu}}$, and the second is a product of rates, $\vec{\mu}\times\dot{\widehat{\vec{\nu}}}$. Both are of Burnett order, i.e.\ second order in the Knudsen number, and are dropped at the order of the closure.}
The remaining term, i.e.\ $-2\,\vec\Sigma\cdot(\nabla_{\vec x}\vec\sigma)\vec V = -2\,\Sigma_a(\nabla_{\vec x}\vec\sigma)_{ab}V_b$ is bilinear in $\vec V$ and $\vec\Sigma$ and does not assemble into a polyatomic bulk-polynomial mode. Multiplied by the prefactor $-\tfrac{1}{2C_{m,\ell}k_B\theta}$ it contributes to $\mathcal{D}\log f^{(0)}$ the term
\begin{equation}
	\frac{\Sigma_a\,(\nabla_{\vec x}\vec\sigma)_{ab}\,V_b}{C_{m,\ell}\,k_B\theta},
\end{equation}
which is the sixth term in~\eqref{eq:ce_f1_explicit}. 

We now collect all contributions to $\mathcal{D}\log f^{(0)}$ that are proportional to $\nabla_{\vec x}\!\cdot\!\vec u$. Three sources contribute:
\begin{enumerate}
	\item From $\mathcal{D}\log\rho$ via the continuity equation, $-\nabla_{\vec x}\!\cdot\!\vec u$.
	\item From the trace part of $\mathcal{D}\abs{\vec V}^2$, the term $\tfrac{m\abs{\vec V}^2}{3k_B\theta}\,\nabla_{\vec x}\!\cdot\!\vec u$.
	\item From the first term of \eqref{eq:ce_streaming_theta} multiplied by $A$.
\end{enumerate}
This yields exactly the bulk term in \eqref{eq:ce_f1_explicit}.

The second term in \eqref{eq:ce_streaming_theta}, namely
\begin{equation*}
	\frac{2\,\cla{K}}{5\,\uz{C_{m,\ell}}\,R_s\,\rho\,\theta}\,(\vec{\nu}\times\nabla_{\vec x}\widehat{\vec{\nu}})\!:\!\nabla_{\vec x}\vec\mu,
\end{equation*}
multiplied by $A$ becomes the fifth (couple-stress reactive) summand in~\eqref{eq:ce_f1_explicit}\uz{, and the tensorial Ericksen contraction in the first bracket of \eqref{eq:ce_streaming_theta}, multiplied by $A$, becomes the second (anisotropic) bulk summand}.

Substituting all \uz{six} contributions into~\eqref{eq:ce_f1_proof_start} and pulling out the common factor $-\bar\tau\,f^{(0)}$ yields~\eqref{eq:ce_f1_explicit}.
\end{proof}

\subsection{First-order entropy flux and entropy production}\label{sec:ce_entropy}

Throughout this subsection we evaluate the kinetic definitions of the entropy flux vector and the entropy production rate given in \eqref{defn:entropy_flux_and_production}, namely
\begin{equation}
    \label{eq:ce_entropy_definitions}
    \vec{\Phi} = -\frac{\rho k_B}{m}\cchevrons{\vec{V}\log(f)}, \qquad \xi = - k_B\int \mathcal{Q}(f,f)\log(f)\,d\vec{\Xi},
\end{equation}
on the Chapman--Enskog ansatz $f = f^{(0)} + \epsilon\,f^{(1)} + O(\epsilon^2)$ of~\eqref{eq:ce_expansion}, with $f^{(1)}$ given explicitly by~\eqref{eq:ce_f1_explicit}. \cla{The single-relaxation BGK collision operator~\eqref{eq:BGK},
\begin{equation}
    \label{eq:ce_BGK_single}
    \mathcal{Q}_{\tau}[f] = -\frac{1}{\tau}\bigl(f - f^{(0)}\bigr),
\end{equation}
is the point of entry exploited below. The equality case of \Cref{lemma:BGK_entropy} ensures $\xi_\tau^{(0)} = 0$, as already discussed in~\eqref{eq:zeroth_entropy_production_zero}.}

We shall repeatedly use the following Gaussian moments against the local Maxwellian $f^{(0)}$.

\begin{lemma}[Maxwellian moments]\label{lemma:maxwellian_moments}
	With $\cchevrons{\cdot}_{\sim f^{(0)}}$ denoting the average with respect to $f^{(0)}$, and writing $B \coloneqq m\abs{\vec V}^2/(2k_B\theta)$, $C \coloneqq \abs{\vec\Sigma}^2/(2C_{m,\ell}k_B\theta)$, one has
\begin{subequations}
    \label{eq:maxwellian_moments}
    \begin{align}
	    &\cchevrons{V_i V_j}_{\sim f^{(0)}} = \frac{k_B\theta}{m}\,\delta_{ij}, \qquad \cchevrons{\Sigma_a\Sigma_c}_0 = C_{m,\ell}k_B\theta\,(\delta_{ac} - \nu_a\nu_c),\\
	    &\cchevrons{V_iV_jV_kV_l}_{\sim f^{(0)}} = \Bigl(\frac{k_B\theta}{m}\Bigr)^2\,\bigl(\delta_{ij}\delta_{kl}+\delta_{ik}\delta_{jl}+\delta_{il}\delta_{jk}\bigr), \\
	    &\cchevrons{1}_{\sim f^{(0)}} = 1,\quad \cchevrons{B}_{\sim f^{(0)}} = \tfrac{3}{2}, \quad \cchevrons{B^2}_{\sim f^{(0)}} = \tfrac{15}{4}, \quad \cchevrons{B^3}_{\sim f^{(0)}} = \tfrac{105}{8}, \\
	    &\cchevrons{C}_{\sim f^{(0)}} = 1, \; \cchevrons{C^2}_{\sim f^{(0)}} = 2,\\
	    &\cchevrons{(B-\tfrac{5}{2})^2}_{\sim f^{(0)}} = \tfrac{5}{2},\quad \cchevrons{(C-1)^2}_{\sim f^{(0)}} = 1,\quad \cchevrons{(B-\tfrac{5}{2})B}_{\sim f^{(0)}} = 0, \\
	    &\cchevrons{(\tfrac{2B}{3}-C)^2}_{\sim f^{(0)}} = \tfrac{5}{3},\quad \cchevrons{(B+C-\tfrac{5}{2})^2}_{\sim f^{(0)}} = \tfrac{5}{2},\quad \cchevrons{(\tfrac{2B}{3}-C)(B+C-\tfrac{5}{2})}_{\sim f^{(0)}} = 0,
    \end{align}
\end{subequations}
	together with the cross-independence $\cchevrons{g(\abs{\vec V}^2)\,h(\abs{\vec\Sigma}^2)}_{\sim f^{(0)}} = \cchevrons{g}_{\sim f^{(0)}}\cchevrons{h}_{\sim f^{(0)}}$ and all odd moments of $\vec V$ and $\vec\Sigma$ vanish.
\end{lemma}

The moments in \eqref{eq:maxwellian_moments} are standard Gaussian integrals. Note that for $\vec\Sigma$ we use that the rigid-rod angular momentum lives in the two-dimensional plane orthogonal to $\vec{\nu}$, whence the projector $\delta_{ac}-\nu_a\nu_c$.

\begin{theorem}[First-order entropy flux]\label{thm:ce_entropy_flux}
\cla{In the aligned regime $s\to 1$,} the first-order Chapman--Enskog entropy flux vector obtained by substituting the expansion~\eqref{eq:ce_expansion} into~\eqref{eq:ce_entropy_definitions} is
\begin{equation}
    \label{eq:ce_entropy_flux}
    \vec\Phi \;=\; \frac{\vec Q^{(1)}}{\theta},
\end{equation}
where $\vec Q^{(1)}$ is the first-order Chapman--Enskog heat flux, obtained by evaluating the kinetic definition of $\vec Q$ given in~\eqref{eq:balance_laws} on $f^{(1)}$, namely
\begin{equation}
    \label{eq:ce_heat_flux}
    \vec Q^{(1)} \;=\; m\!\int \vec V\,\tfrac{1}{2}\abs{\vec V}^2\,f^{(1)}\,\dd\vec\Xi \;+\; \int \vec V\,\tfrac{1}{2C_{m,\ell}}\abs{\vec\Sigma}^2\,f^{(1)}\,\dd\vec\Xi.
\end{equation}
\end{theorem}

\begin{proof}
Expanding $\log f = \log f^{(0)} + \epsilon\,f^{(1)}/f^{(0)} + O(\epsilon^2)$ and $f = f^{(0)} + \epsilon f^{(1)} + O(\epsilon^2)$, a direct computation gives
\begin{equation}
    \label{eq:ce_flux_expansion}
    \int \vec V\,\log(f)\,f\,\dd\vec\Xi \;=\; \int \vec V\,\log(f^{(0)})\,f^{(0)}\,\dd\vec\Xi \;+\; \epsilon\!\int \vec V\,\bigl[\log(f^{(0)}) + 1\bigr] f^{(1)}\,\dd\vec\Xi \;+\; O(\epsilon^2).
\end{equation}
The zeroth-order integral vanishes in fact since $\log f^{(0)}$ is an even polynomial in $\vec V$ thus $\vec V\log f^{(0)}\,f^{(0)}$ is odd in $\vec V$ and integrates to zero. The additive constant $1$ and the $\vec V,\vec\Sigma$-independent part of $\log f^{(0)}$ multiply $\int\vec V f^{(1)}\dd\vec\Xi = 0$, by the solvability condition~\eqref{eq:ce_sol_conditions}. The orientational term $\log\lambda(\vec{\nu})$ multiplies
\begin{equation*}
    \int \vec V f^{(1)}\,\dd\vec V\,\dd\vec\Sigma \quad\text{at fixed } \vec{\nu},
\end{equation*}
which in the aligned regime $s\to 1$ where $\lambda$ is concentrated about the director $\widehat{\vec{\nu}}$ vanishes by solvability.

We are left with the two $\vec V,\vec\Sigma$-dependent contributions of $\log f^{(0)}$:
\begin{equation*}
    \int \vec V\,\log(f)\,f\,\dd\vec\Xi \;=\; \epsilon\!\int \vec V\!\left[-\frac{m\abs{\vec V}^2}{2k_B\theta} - \frac{\abs{\vec\Sigma}^2}{2C_{m,\ell}k_B\theta}\right]\! f^{(1)}\,\dd\vec\Xi \;+\; O(\epsilon^2).
\end{equation*}
Multiplying by $-k_B$ and recalling that $\vec\Phi = -k_B\int\vec V \log(f) f\,\dd\vec\Xi$,
\begin{equation*}
    \vec\Phi \;=\; \frac{\epsilon}{\theta}\!\int \vec V\!\left[\frac{m\abs{\vec V}^2}{2} + \frac{\abs{\vec\Sigma}^2}{2C_{m,\ell}}\right]\! f^{(1)}\,\dd\vec\Xi \;+\; O(\epsilon^2).
\end{equation*}
The right-hand side is precisely $\vec Q^{(1)}/\theta$ since absorbing $\epsilon$ back into the physical rescaling \cla{$\tau = \epsilon\,\bar\tau$} gives~\eqref{eq:ce_entropy_flux}.
\end{proof}

The first-order entropy flux \eqref{eq:ce_entropy_flux} satisfies the same expansion-independence statement as the equation of state of \Cref{cor:eos_first_order}. In Clausius--Duhem form $\vec\Phi^{(1)} = \vec\Phi^{(0)} + \theta^{-1}\vec Q^{(1)} + O(\epsilon^2)$ its zeroth-order flux vanishes, $\vec\Phi^{(0)} = \vec 0$, the analogue of \cite[Proposition~3]{malek}.

\begin{theorem}[First-order entropy production]\label{thm:ce_entropy_production}
\cla{In the aligned regime $s\to 1$,} the first-order Chapman--Enskog entropy production rate obtained by substituting the expansion~\eqref{eq:ce_expansion} into the kinetic definition~\eqref{eq:ce_entropy_definitions} with the single-relaxation BGK operator~\eqref{eq:ce_BGK_single} is
\begin{equation}
    \label{eq:ce_entropy_production}
    \xi \;=\; \frac{k_B\,\cla{\tau}\,\rho}{m}\,\cchevrons{D^2}_{\sim f^{(0)}},
\end{equation}
where $D$ is the bracketed factor of~\eqref{eq:ce_f1_explicit}, i.e.\ $f^{(1)} = -\cla{\bar\tau}\,f^{(0)}\,D$, and $\cchevrons{\cdot}_{\sim f^{(0)}}$ denotes the average with respect to $f^{(0)}$. Explicitly, in terms of the thermodynamic forces $\mathbb D = \tfrac{1}{2}(\nabla_{\vec x}\vec u+(\nabla_{\vec x}\vec u)^T)$, $\mathbb D^d = \mathbb D - \tfrac{1}{3}(\nabla_{\vec x}\!\cdot\!\vec u)\,\mat{I}$, the temperature gradient $\nabla_{\vec x}\theta$, and the orientational forces $\nabla_{\vec x}\widehat{\vec{\nu}}$, $\nabla_{\vec x}\vec\mu$, the entropy production admits the manifestly non-negative decomposition
\begin{equation}
    \label{eq:ce_entropy_production_explicit}
    \xi \;=\; \xi_{\mathrm{shear}} + \xi_{\mathrm{heat}} + \xi_{\mathrm{bulk}}^{\mathrm{iso}} + \xi_{\mathrm{bulk}}^{\mathrm{aniso}} + \xi_{\mathrm{ang}} \;\geq\; 0,
\end{equation}
\cla{with each summand given by
\begin{subequations}
\label{eq:ce_entropy_production_terms}
\begin{align}
    \xi_{\mathrm{shear}}       &= \frac{2\,k_B\,\tau\,\rho}{m}\,\mathbb D^d\!:\!\mathbb D^d,\\
    \xi_{\mathrm{heat}}        &= \frac{7}{2}\,\frac{k_B^2\,\tau\,\rho\,\theta}{m^2}\,\frac{\abs{\nabla_{\vec x}\theta}^2}{\theta^2},\\
    \xi_{\mathrm{bulk}}^{\mathrm{iso}}    &= \frac{4}{15}\,\frac{k_B\,\tau\,\rho}{m}\,(\nabla_{\vec x}\!\cdot\!\vec u)^2,\\
    \xi_{\mathrm{bulk}}^{\mathrm{aniso}}  &= \uz{\frac{2\,(K/p_{\mathrm{tr,rot}})^2}{5}\,\frac{k_B\,\tau\,\rho}{m}\Bigl[(\nabla_{\vec x}\widehat{\vec{\nu}}\odot\nabla_{\vec x}\widehat{\vec{\nu}})\!:\!\nabla_{\vec x}\vec u \fix{-} \frac{1}{C_{m,\ell}}(\widehat{\vec{\nu}}\times\nabla_{\vec x}\widehat{\vec{\nu}})\!:\!\nabla_{\vec x}\vec\mu\Bigr]^{\!2}},\\
    \xi_{\mathrm{ang}}          &= \uz{\frac{k_B\,\tau\,\rho}{m^2\,C_{m,\ell}}\,(\delta_{ac}-\widehat{\nu}_a\widehat{\nu}_c)(\partial_b\mu_a)(\partial_b\mu_c)},
\end{align}
\end{subequations}
corresponding respectively to the Newtonian shear, polyatomic Fourier heat conduction, the translational--rotational imbalance bulk viscosity, the anisotropic coupling of \uz{the Ericksen tensor with the velocity gradient} and couple-stress reactivity, and the angular-advection dissipation. \uz{No Ericksen-stress dissipation appears: the would-be summand $\xi_{\mathrm{E}} = (\tau/(\theta\rho))\abs{\vec{W}}^2$ vanishes because $\vec{W}$ is removed by the solvability condition, see \eqref{eq:ce_W_constraint} and \Cref{rmk:xiE_solvability}.} \uz{The two sub-modes of the bulk component, $(2B/3-C)$ and $(B+C-5/2)$ in the reduced variables of \eqref{eq:maxwellian_moments}, are $L^2(f^{(0)})$-orthogonal, so $\xi_{\mathrm{bulk}}^{\mathrm{iso}}$ and $\xi_{\mathrm{bulk}}^{\mathrm{aniso}}$ appear as separate perfect squares with no cross term.} \fix{The Frank modulus $K = K_{\mathrm{kin}}$ appears in $\xi_{\mathrm{bulk}}^{\mathrm{aniso}}$ via the ratio $K/p_{\mathrm{tr,rot}} = C_{m,\ell}/m$, recovering the form of \cite{farrell}.}}
\end{theorem}

\begin{proof}
With the single-temperature BGK operator~\eqref{eq:ce_BGK_single} the kinetic definition~\eqref{eq:ce_entropy_definitions} reads
\begin{equation}
    \label{eq:ce_xi_starting}
    \xi \;=\; \frac{k_B}{\tau}\!\int (f - f^{(0)})\,\log(f)\,\dd\vec\Xi.
\end{equation}
We split $\log f = \log f^{(0)} + \log(f/f^{(0)})$ and note that
\begin{equation*}
    \int (f - f^{(0)})\,\log f^{(0)}\,\dd\vec\Xi \;=\; 0,
\end{equation*}
because $\log f^{(0)}$, given by~\eqref{eq:ce_log_maxwellian}, is a linear combination of the collision invariants $1$, $m\abs{\vec V}^2$, $\abs{\vec\Sigma}^2$ and of $\log\lambda(\vec{\nu})$, all of which are annihilated by the solvability conditions~\eqref{eq:ce_sol_conditions}, thus $\int (f-f^{(0)})\,\dd\vec V\,\dd\vec\Sigma = 0$. Hence
\begin{equation}
    \label{eq:ce_xi_relative}
    \xi \;=\; \frac{k_B}{\tau}\!\int (f - f^{(0)})\,\log\!\bigl(f/f^{(0)}\bigr)\,\dd\vec\Xi.
\end{equation}
Substituting $f - f^{(0)} = \epsilon\,f^{(1)} + O(\epsilon^2)$ and $\log(f/f^{(0)}) = \epsilon\,f^{(1)}/f^{(0)} + O(\epsilon^2)$ yields
\begin{equation}
    \label{eq:ce_xi_leading}
    \xi \;=\; \frac{k_B\,\epsilon^2}{\tau}\!\int \frac{(f^{(1)})^2}{f^{(0)}}\,\dd\vec\Xi \;+\; O(\epsilon^3).
\end{equation}
Writing $f^{(1)} = -\bar\tau\,f^{(0)}\,D$ from~\eqref{eq:ce_f1_explicit} and using $\tau = \epsilon\,\bar\tau$,
\begin{equation*}
	\xi \;=\; k_B\,\epsilon\,\bar\tau\!\int f^{(0)}\,D^2\,\dd\vec\Xi \;=\; k_B\,\tau\,\frac{\rho}{m}\,\cchevrons{D^2}_{\sim f^{(0)}},
\end{equation*}
which is~\eqref{eq:ce_entropy_production} and is manifestly non-negative.

It remains to evaluate $\cchevrons{D^2}_{\sim f^{(0)}}$ using \Cref{lemma:maxwellian_moments}. We label the seven bracketed summands of~\eqref{eq:ce_f1_explicit} as $D_s$ for shear, $D_{\mathrm h}^t$ for translational heat, $D_{\mathrm h}^r$ for rotational heat, $D_b$ for bulk, $D_{\mathrm{cs}}$ for the couple-stress reactive term, $D_{\mathrm{aa}}$ for angular advection, and $D_{\mathrm E}$ for the Ericksen term. Expanding $D^2 = \sum_i D_i^2 + 2\sum_{i<j}D_iD_j$ and using the parity properties of \Cref{lemma:maxwellian_moments}, all cross-terms $\cchevrons{D_i D_j}_{\sim f^{(0)}}$ vanish except for $\cchevrons{D_bD_{\mathrm{cs}}}_{\sim f^{(0)}}$, which we compute explicitly below.

\emph{Shear.} Using the fourth-order Gaussian moment,
\begin{equation*}
	\cchevrons{\bigl((\vec V\otimes\vec V)^d\!:\!\mathbb D\bigr)^2}_{\sim f^{(0)}} \;=\; 2\Bigl(\tfrac{k_B\theta}{m}\Bigr)^{\!2}\,\mathbb D^d\!:\!\mathbb D^d,
\end{equation*}
so $\cchevrons{D_s^2}_{\sim f^{(0)}} = 2\,\mathbb D^d\!:\!\mathbb D^d$.

\emph{Heat.} By isotropy of the Gaussian in $\vec V$ and by cross-independence of $\vec V$ and $\vec\Sigma$,
\begin{align*}
	\cchevrons{(D_{\mathrm h}^t)^2}_{\sim f^{(0)}} &= \tfrac{1}{3}\cchevrons{(B-\tfrac{5}{2})^2 \abs{\vec V}^2}_{\sim f^{(0)}} \tfrac{\abs{\nabla_{\vec x}\theta}^2}{\theta^2} = \tfrac{5}{2}\tfrac{k_B\theta}{m}\tfrac{\abs{\nabla_{\vec x}\theta}^2}{\theta^2},\\
	\cchevrons{(D_{\mathrm h}^r)^2}_{\sim f^{(0)}} &= \cchevrons{(C-1)^2}_{\sim f^{(0)}}\cchevrons{(\vec V\!\cdot\!\nabla_{\vec x}\theta/\theta)^2}_{\sim f^{(0)}} = \tfrac{k_B\theta}{m}\tfrac{\abs{\nabla_{\vec x}\theta}^2}{\theta^2},
\end{align*}
and the cross $\cchevrons{D_{\mathrm h}^tD_{\mathrm h}^r}_{\sim f^{(0)}}$ vanishes because $\cchevrons{C-1}_{\sim f^{(0)}} = 0$ and $\vec\Sigma\perp\vec V$ under $f^{(0)}$. Summing, $\cchevrons{(D_{\mathrm h}^t+D_{\mathrm h}^r)^2}_{\sim f^{(0)}} = \tfrac{7}{2}\tfrac{k_B\theta}{m}\abs{\nabla_{\vec x}\theta}^2/\theta^2$.

\uz{\emph{Bulk and couple-stress reactive.} With $\alpha := \nabla_{\vec x}\!\cdot\!\vec u$ and the combined anisotropic affinity
\begin{equation*}
	\mathcal{B} \;:=\; (\nabla_{\vec x}\widehat{\vec{\nu}}\odot\nabla_{\vec x}\widehat{\vec{\nu}})\!:\!\nabla_{\vec x}\vec u \;\fix{-}\; \frac{1}{C_{m,\ell}}\,(\widehat{\vec{\nu}}\times\nabla_{\vec x}\widehat{\vec{\nu}})\!:\!\nabla_{\vec x}\vec\mu,
\end{equation*}
the anisotropic bulk summand and the couple-stress reactive summand of \eqref{eq:ce_f1_explicit} carry the same Gaussian factor $(B+C-\tfrac{5}{2})$ and assemble into
\begin{equation*}
    D_b + D_{\mathrm{cs}} \;=\; \tfrac{2\alpha}{5}(\tfrac{2B}{3}-C) \;-\; \tfrac{2}{5}\,(K/p_{\mathrm{tr,rot}})\,\mathcal{B}\,(B+C-\tfrac{5}{2}).
\end{equation*}
Using the moments of \eqref{eq:maxwellian_moments},
\begin{equation*}
	\cchevrons{(\tfrac{2B}{3}-C)^2}_{\sim f^{(0)}} = \tfrac{5}{3}, \qquad
	\cchevrons{(B+C-\tfrac{5}{2})^2}_{\sim f^{(0)}} = \tfrac{5}{2}, \qquad
	\cchevrons{(\tfrac{2B}{3}-C)(B+C-\tfrac{5}{2})}_{\sim f^{(0)}} = 0,
\end{equation*}
we obtain
\begin{equation*}
	\cchevrons{(D_b + D_{\mathrm{cs}})^2}_{\sim f^{(0)}} \;=\; \tfrac{4}{15}\alpha^2 \;+\; \tfrac{2}{5}\,(K/p_{\mathrm{tr,rot}})^2\,\mathcal{B}^2.
\end{equation*}
	The isotropic and anisotropic bulk components are orthogonal Gaussian modes and appear as two separate perfect squares, with no cross term to complete.}

\emph{Angular advection.} By the rigid-rod projector identity $\cchevrons{\Sigma_a\Sigma_c}_0 = C_{m,\ell}k_B\theta\,(\delta_{ac}-\nu_a\nu_c)$ discussed in \Cref{sec:vanishing_girth} and $\cchevrons{V_bV_d}_0 = (k_B\theta/m)\delta_{bd}$, and taking the strong alignment limit $s\to 1$ for the $\vec{\nu}$ integration,
\begin{equation*}
	\cchevrons{D_{\mathrm{ang}}^2}_{\sim f^{(0)}} \;=\; \frac{1}{(C_{m,\ell}k_B\theta)^2}\,(C_{m,\ell}k_B\theta)\,\tfrac{k_B\theta}{m}\,(\delta_{ac}-\widehat{\nu}_a\widehat{\nu}_c)(\partial_b\sigma_a)(\partial_b\sigma_c) \;=\; \frac{(\delta_{ac}-\widehat{\nu}_a\widehat{\nu}_c)(\partial_b\sigma_a)(\partial_b\sigma_c)}{m\,C_{m,\ell}}.
\end{equation*}
\uz{We express the result in terms of the intrinsic angular momentum: by the strong-alignment identity $\vec\mu = \cchevrons{\vec{\nu}\times\vec\varsigma} = \widehat{\vec{\nu}}\times\vec\sigma$, equivalently $\vec\sigma = \vec\mu\times\widehat{\vec{\nu}}$, the product rule gives $\partial_b\vec\sigma = (\partial_b\vec\mu)\times\widehat{\vec{\nu}} + \vec\mu\times\partial_b\widehat{\vec{\nu}}$, the second term being of higher order in the alignment expansion since $\vec\mu$ itself vanishes at the aligned equilibrium. Using $\abs{\vec a\times\widehat{\vec{\nu}}}^2 = \abs{\vec a}^2 - (\vec a\cdot\widehat{\vec{\nu}})^2$, the projected quadratic form transfers verbatim,
\begin{equation*}
	(\delta_{ac}-\widehat{\nu}_a\widehat{\nu}_c)(\partial_b\sigma_a)(\partial_b\sigma_c) \;=\; (\delta_{ac}-\widehat{\nu}_a\widehat{\nu}_c)(\partial_b\mu_a)(\partial_b\mu_c) + \mathcal{O}(\epsilon^4),
\end{equation*}
which is the form stated in \eqref{eq:ce_entropy_production_terms}, written directly in the couple-stress affinity $\nabla_{\vec x}\vec\mu$ in view of the pairing of \Cref{sec:hybrid}.}

\uz{\emph{Remaining cross-terms.} Every remaining cross $\cchevrons{D_iD_j}_{\sim f^{(0)}}$ vanishes: either by $\vec V$-parity (when one factor carries an odd power of $\vec V$ and the other an even power, e.g.\ heat against shear, bulk and couple-stress), by $\vec\Sigma$-parity (angular-advection against any non-$\vec\Sigma$-carrying piece), or by tracelessness of $(\vec V\otimes\vec V)^d$ (shear against bulk and couple-stress).}

Summing all surviving contributions and multiplying by the overall prefactor $k_B\tau\rho/m$ yields the \uz{five-term} decomposition~\eqref{eq:ce_entropy_production_terms}.
\end{proof}

\uz{Reading the transport coefficients of the model from~\eqref{eq:ce_entropy_production_terms}, we find the shear viscosity $\zeta_{\mathrm{shear}} = p_{\mathrm{tr,rot}}\,\tau$, the thermal conductivity $\kappa = \tfrac{7}{2}(k_B/m)\,p_{\mathrm{tr,rot}}\,\tau = c_p\,p_{\mathrm{tr,rot}}\,\tau$, where $c_p = \tfrac{7}{2}(k_B/m)$ is the specific heat at constant pressure of a diatomic gas, the bulk viscosity $\zeta_{\mathrm{iso}} = \tfrac{4}{15}\,p_{\mathrm{tr,rot}}\,\tau$, the anisotropic coefficient $\zeta_{\mathrm{aniso}} = \tfrac{2}{5}\,(K/p_{\mathrm{tr,rot}})^2\,p_{\mathrm{tr,rot}}\,\tau$, and the spin-diffusion coefficient $\zeta_{\mathrm{ang}} = p_{\mathrm{tr,rot}}\,\tau/(m\,C_{m,\ell})$. The ratio $\kappa/\zeta_{\mathrm{shear}} = c_p$ recovers the standard Eucken relation.} \cla{In contrast with a monatomic BGK derivation, $\zeta_{\mathrm{iso}}$ does not vanish: Stokes' hypothesis is not an issue in the present single-relaxation BGK framework because the Ericksen-stress contribution to the zeroth-order pressure carries through the Chapman--Enskog substitution rule.} \uz{The two summands $\xi_{\mathrm{bulk}}^{\mathrm{aniso}}$ and $\xi_{\mathrm{ang}}$} have no counterpart in a standard polyatomic BGK derivation and are structural consequences of the Ericksen and couple-stress contributions to the zeroth-order balance laws~\eqref{eq:ce_zeroth_balance}.

\uz{With these coefficients, the first-order entropy production takes the compact final form
\begin{equation}
	\label{eq:ce_entropy_production_final}
	\begin{aligned}
	\theta\,\xi \;=\;& 2\,\zeta_{\mathrm{shear}}\,\mathbb D^d\!:\!\mathbb D^d
	\;+\; \zeta_{\mathrm{iso}}\,(\nabla_{\vec x}\!\cdot\!\vec u)^2
	\;+\; \zeta_{\mathrm{aniso}}\Bigl[(\nabla_{\vec x}\widehat{\vec{\nu}}\odot\nabla_{\vec x}\widehat{\vec{\nu}})\!:\!\nabla_{\vec x}\vec u \fix{-} \tfrac{1}{C_{m,\ell}}(\widehat{\vec{\nu}}\times\nabla_{\vec x}\widehat{\vec{\nu}})\!:\!\nabla_{\vec x}\vec\mu\Bigr]^{2}\\
	&+\; \zeta_{\mathrm{ang}}\,(\delta_{ac}-\widehat{\nu}_a\widehat{\nu}_c)(\partial_b\mu_a)(\partial_b\mu_c)
	\;+\; \kappa\,\frac{\abs{\nabla_{\vec x}\theta}^{2}}{\theta},
	\end{aligned}
\end{equation}
a manifestly non-negative quadratic form in the affinities $(\mathbb D^d,\ \nabla_{\vec x}\!\cdot\!\vec u,\ \nabla_{\vec x}\widehat{\vec{\nu}},\ \nabla_{\vec x}\vec\mu,\ \nabla_{\vec x}\theta)$. This is the form of the entropy production used by the Rajagopal--Srinivasa closure of \Cref{sec:hybrid}.}

\os{
\section{Rajagopal--Srinivasa--Chapman--Enskog closure}\label{sec:hybrid}

The first-order Chapman--Enskog expansion of~\Cref{sec:ce} produced a non-negative expression \eqref{eq:ce_entropy_production_final} for the entropy production rate. In addition, the Chapman--Enskog procedure yielded an explicit expression \eqref{eq:ce_entropy_flux} for the entropy flux.

A complementary thermodynamic constraint is provided by the Chapman--Enskog-based constitutive relation for the Helmholtz free energy derived in \Cref{sec:entropy_balance_and_rate}. This relation allows us to infer the structure of the entropy balance from purely thermodynamic arguments, using only the balance equations. Combining this thermodynamic constraint with the kinetic-theory constraint on the entropy flux, the following corollary gives an independent expression for the entropy production rate and exposes its flux--affinity structure.

\begin{corollary}[Entropy production inferred from the first-order Helmholtz free-energy structure]\label{cor:first_order_entropy_production}
Substituting the entropy balance \eqref{eq:entropy_balance_law},
\[
\rho\,\frac{d\eta}{dt} = \xi - \nabla_{\vec{x}}\!\cdot\!\vec{\Phi},
\]
into \uz{the entropy-production identity \eqref{eq:entropy_production_constitutive}}, and using the first-order entropy flux
\[
\vec{\Phi} = \frac{\vec{Q}^{(1)}}{\theta}
\]
from \Cref{thm:ce_entropy_flux}, yields
\begin{align}
\nonumber
\theta\xi^{(1)}
&= \left[\mathbb{T}^{(1)} + p\mat{I} + \rho \,(\nabla_{\vec{x}}\widehat{\vec{\nu}})^T \partial_{\nabla_{\vec{x}}\widehat{\vec{\nu}}} \psi\right]\!:\!\nabla_{\vec{x}}\vec{u} \\
\nonumber
&\quad + \fix{\frac{1}{mC_{m,\ell}}}\left[\mathbb{M}^{(1)}-\fix{m\rho}\!\left(\widehat{\vec{\nu}}\times \partial_{\nabla_{\vec{x}}\widehat{\vec{\nu}}} \psi\right)\right]\!:\!\nabla_{\vec{x}}\vec{\mu}
-\fix{\vec{Q}^{(1)}}\cdot\frac{\nabla_{\vec{x}}\theta}{\theta} \fix{.}
\end{align}

\fix{Since the constitutive free energy carries no interaction component, the heat flux is the bare kinetic flux $\vec{Q}$ and requires no interaction correction, so the final divergence term is absent.} The tensor
\begin{equation}
	\left[\mathbb{T}^{(1)} + p\mat{I} + \rho \,(\nabla_{\vec{x}}\widehat{\vec{\nu}})^T \partial_{\nabla_{\vec{x}}\widehat{\vec{\nu}}} \psi\right]
\end{equation}
is symmetric, by the symmetry of $\mathbb{T}$ and by the assumed structure of $\psi$. Consequently, this tensor contracts only with the symmetric part of the velocity gradient, $\mathbb{D}$. Splitting $\mathbb{D}$ into its spherical and traceless parts, we obtain the following bilinear form for the entropy production rate:
\begin{equation}\label{eq:hybrid_constraint}
\theta\,\xi(J,A)
= J_{\mathrm{tr}\mathbb{T}}A_{\mathrm{tr}\mathbb{T}}
+ \mathbb{J}_{\mathbb{T}^d}\!:\!\mathbb{A}_{\mathbb{T}^d}
+ \mathbb{J}_{\mathbb{M}}\!:\!\mathbb{A}_{\mathbb{M}}
+ \vec{J}_{\vec{Q}}\!\cdot\!\vec{A}_{\vec{Q}} .
\end{equation}
The corresponding thermodynamic fluxes and affinities are defined by
\begin{subequations}\label{eq:hybrid_JA}
\begin{align}
J_{\mathrm{tr}\mathbb{T}}
&= \frac{1}{3}\mathrm{tr}(\mathbb{T}^{(1)}) + p
+ \frac{1}{3}\rho\,(\nabla_{\vec x}\widehat{\vec{\nu}})^T :\partial_{\nabla_{\vec{x}}\widehat{\vec{\nu}}} \psi,
&
A_{\mathrm{tr}\mathbb{T}}
&= \nabla_{\vec{x}}\cdot\vec{u}, \\
\mathbb{J}_{\mathbb{T}^d}
&= \left(\mathbb{T}^{(1)}  + 
\rho \,(\nabla_{\vec{x}}\widehat{\vec{\nu}})^T \partial_{\nabla_{\vec{x}}\widehat{\vec{\nu}}} \psi
\right)^{d},
&
\mathbb{A}_{\mathbb{T}^d}
&= \mathbb{D}^d, \\
\mathbb{J}_{\mathbb{M}}
&= \mathbb{M}^{(1)} - \fix{m\rho}\bigl(\widehat{\vec{\nu}}\times\partial_{\nabla_{\vec{x}}\widehat{\vec{\nu}}}\psi\bigr),
&
\mathbb{A}_{\mathbb{M}}
&= \fix{\frac{1}{mC_{m,\ell}}}\nabla_{\vec{x}}\vec{\mu}, \\
\vec{J}_{\vec{Q}}
&= \vec{Q}^{(1)},
&
\vec{A}_{\vec{Q}}
&= -\frac{\nabla_{\vec{x}}\theta}{\theta}.
\end{align}
\end{subequations}
\end{corollary}

\subsection{Rajagopal--Srinivasa maximisation principle}\label{sec:hybrid_RSP}

\uz{The final form \eqref{eq:ce_entropy_production_final} of the first-order entropy production is a positive-definite and convex quadratic form in the affinities
\begin{equation}
	(\nabla_{\vec x}\!\cdot\!\vec u,\mathbb D^d,\nabla_{\vec x }\vec{\mu}, -\frac{\nabla_{\vec x}\theta}{\theta}).
\end{equation}
Introducing the shorthand $\mathcal{B}(A)$ for the anisotropic affinity
\begin{equation}
\label{eq:hybrid_B}
\mathcal B(A)
= (\nabla_{\vec x}\widehat{\vec{\nu}}\odot\nabla_{\vec x}\widehat{\vec{\nu}})\!:\!\nabla_{\vec x}\vec u
\fix{-} \frac{1}{C_{m,\ell}}\,(\widehat{\vec{\nu}}\times\nabla_{\vec x}\widehat{\vec{\nu}})\!:\!\nabla_{\vec x}\vec\mu,
\end{equation}
the entropy production rate reads
\begin{equation}
\label{eq:hybrid_iota}
\begin{aligned}
\theta\,\xi(A)
&= 2\zeta_{\mathrm{shear}}\,\mathbb D^d\!:\!\mathbb D^d
+ \zeta_{\mathrm{iso}}(\nabla_{\vec x}\!\cdot\!\vec u)^2
+ \zeta_{\mathrm{aniso}}\,\mathcal{B}(A)^{2}
+ \zeta_{\mathrm{ang}}\,(\delta_{ac}-\widehat{\nu}_a\widehat{\nu}_c)(\partial_b\mu_a)(\partial_b\mu_c)
+ \kappa\,\frac{\abs{\nabla_{\vec x}\theta}^{2}}{\theta},\\
\mathcal B(A)
&= (\nabla_{\vec x}\widehat{\vec{\nu}}\odot\nabla_{\vec x}\widehat{\vec{\nu}})\!:\!\nabla_{\vec x}\vec u
\fix{-} \frac{1}{C_{m,\ell}}(\widehat{\vec{\nu}}\times\nabla_{\vec x}\widehat{\vec{\nu}})\!:\!\nabla_{\vec x}\vec\mu,
\end{aligned}
\end{equation}
with the transport coefficients
\begin{equation}
\label{eq:hybrid_coefficients}
\begin{gathered}
\zeta_{\mathrm{shear}} = p_{\mathrm{tr,rot}}\,\tau, \qquad
\zeta_{\mathrm{iso}} = \tfrac{4}{15}\,p_{\mathrm{tr,rot}}\,\tau, \qquad
\zeta_{\mathrm{aniso}} = \tfrac{2}{5}\,(K/p_{\mathrm{tr,rot}})^{2}\,p_{\mathrm{tr,rot}}\,\tau,\\
\zeta_{\mathrm{ang}} = \frac{p_{\mathrm{tr,rot}}\,\tau}{m\,C_{m,\ell}}, \qquad
\kappa = \tfrac{7}{2}\frac{k_B}{m}\,p_{\mathrm{tr,rot}}\,\tau,
\end{gathered}
\end{equation}
and $\zeta_{\mathrm{aniso}} = \tfrac{2}{5}\fix{(C_{m,\ell}/m)}^2\,p_{\mathrm{tr,rot}}\tau$ \fix{with $K/p_{\mathrm{tr,rot}} = C_{m,\ell}/m$}. Recall that the would-be Ericksen-stress dissipation never enters: \fix{$\vec{W}$ is removed by the solvability condition, see \eqref{eq:ce_W_constraint} and \Cref{rmk:xiE_solvability}}.}

Following~\cite[Section 7]{malek}, we identify the closure relations, i.e.\ the constitutive equations for the thermodynamic fluxes, by means of the maximisation of the entropy production rate, the key idea of Rajagopal \& Srinivasa \cite{vitMalek,rajagopalSrinivasa}. This amounts to solving the constrained optimisation problem
\begin{equation}
\mathrm{maximize}_{A}\ \theta\xi(A) + \Lambda\bigl(\theta\xi(A)-\theta\xi(J,A)\bigr),
\end{equation}
where $\theta\xi(A)$ is given by \eqref{eq:hybrid_iota}, $\theta\xi(J,A)$ is given by \eqref{eq:hybrid_constraint}, and the maximisation is performed over the set of affinities
\begin{equation}
(A_{\mathrm{tr}\mathbb{T}},\mathbb{A}_{\mathbb{T}^d},\mathbb{A}_{\mathbb{M}},\vec{A}_{\vec{Q}}).
\end{equation}
The corresponding stationarity condition is
\begin{equation}
\frac{\partial\theta\xi(A)}{\partial A_i}
= \frac{\Lambda}{1+\Lambda}J_i,
\qquad
A_i\in\{A_{\mathrm{tr}\mathbb{T}},\mathbb{A}_{\mathbb{T}^d},\mathbb{A}_{\mathbb{M}},\vec{A}_{\vec{Q}}\}.
\end{equation}
Multiplying by $A_i$ and summing over all affinities gives
\begin{equation}
\sum_i A_i \frac{\partial\theta\xi(A)}{\partial A_i}
= \frac{\Lambda}{1+\Lambda}\sum_i A_i J_i
= \frac{\Lambda}{1+\Lambda}\xi.
\end{equation}
Since $\xi$ is homogeneous of degree two in the affinities, the left-hand side is equal to $2\theta\xi(A)$. We therefore obtain $\Lambda=-2$ and hence $(1+\Lambda)/\Lambda=1/2$. The resulting constitutive relations are
\begin{align}
J_{\mathrm{tr}{\mathbb{T}}}
&= \uz{\zeta_{\mathrm{iso}}\,A_{\mathrm{tr}\mathbb{T}} + \tfrac{1}{3}\,\zeta_{\mathrm{aniso}}\,\mathcal{B}(A)\,(\nabla_{\vec x}\widehat{\vec{\nu}}\!:\!\nabla_{\vec x}\widehat{\vec{\nu}})}, \\
\mathbb{J}_{\mathbb{T}^d}
&= 2\uz{\zeta_{\mathrm{shear}}}\,\mathbb{A}_{\mathbb{T}^d} \uz{+ \zeta_{\mathrm{aniso}}\,\mathcal{B}(A)\,\bigl(\nabla_{\vec x}\widehat{\vec{\nu}}\odot\nabla_{\vec x}\widehat{\vec{\nu}}\bigr)^{d}}, \\
\mathbb{J}_{\mathbb{M}}
&= \fix{-\,m\,\zeta_{\mathrm{aniso}}}\,\mathcal{B}(A)(\widehat{\vec{\nu}}\times\nabla_{\vec x}\widehat{\vec{\nu}}) \fix{+ mC_{m,\ell}\,\zeta_{\mathrm{ang}}\,(\nabla_{\vec x}\vec\mu)\bigl(\mat{I}-\widehat{\vec{\nu}}\otimes\widehat{\vec{\nu}}\bigr)}, \\
\vec{J}_{\vec{Q}}
&= -\kappa\,\nabla_{\vec x}\theta.
\end{align}
In the trace and deviatoric fluxes, the Ericksen tensor splits as $\nabla_{\vec x}\widehat{\vec{\nu}}\odot\nabla_{\vec x}\widehat{\vec{\nu}} = (\nabla_{\vec x}\widehat{\vec{\nu}}\odot\nabla_{\vec x}\widehat{\vec{\nu}})^d + \tfrac{1}{3}(\nabla_{\vec x}\widehat{\vec{\nu}}\!:\!\nabla_{\vec x}\widehat{\vec{\nu}})\,\mat{I}$, the trace part contributing to $J_{\mathrm{tr}\mathbb{T}}$ and the deviatoric part contributing to $\mathbb{J}_{\mathbb{T}^d}$.
Substituting the definitions of the fluxes and affinities from~\eqref{eq:hybrid_JA} into the above relations yields the following constitutive relations for the stress, couple-stress, and heat flux.

\begin{theorem}[Rajagopal--Srinivasa--Chapman--Enskog constitutive relations]\label{thm:hybrid_closure}
Under the hypotheses of~\Cref{thm:ce_first_order}, the dissipative stress, couple-stress, and heat flux of the single-relaxation BGK model are closed by
\begin{subequations}
\label{eq:hybrid_closure}
\begin{align}
\mathbb{T}^{(1)} + p\,\mat{I}
+ \rho \,(\nabla_{\vec{x}}\widehat{\vec{\nu}})^T \partial_{\nabla_{\vec{x}}\widehat{\vec{\nu}}} \psi
&= 2\uz{\zeta_{\mathrm{shear}}}\,\mathbb D^{d}
+ \uz{\zeta_{\mathrm{iso}}\,(\nabla_{\vec x}\!\cdot\!\vec u)\,\mat{I}}\nonumber\\
&\qquad+ \uz{\zeta_{\mathrm{aniso}}\,\mathcal B(A)\,\bigl(\nabla_{\vec x}\widehat{\vec{\nu}}\odot\nabla_{\vec x}\widehat{\vec{\nu}}\bigr)},
\label{eq:hybrid_T}\\
\mathbb{M}^{(1)} - \fix{m\rho}\,\bigl(\vec{\nu}\times\partial_{\nabla_{\vec x}\widehat{\vec{\nu}}} \psi\bigr)
&= \fix{-\,m\,\zeta_{\mathrm{aniso}}}\,\mathcal B(A)\,\bigl(\vec{\nu}\times\nabla_{\vec x}\widehat{\vec{\nu}}\bigr) \fix{+ mC_{m,\ell}\,\zeta_{\mathrm{ang}}\,(\nabla_{\vec x}\vec\mu)\bigl(\mat{I}-\widehat{\vec{\nu}}\otimes\widehat{\vec{\nu}}\bigr)},
\label{eq:hybrid_M}\\
\vec{Q}^{(1)}
&= -\kappa\,\nabla_{\vec x}\theta.
\label{eq:hybrid_Q}
\end{align}
\end{subequations}
Here $\mathcal B(A)$ is defined in \eqref{eq:hybrid_B}
and the transport coefficients \uz{$\zeta_{\mathrm{shear}},\zeta_{\mathrm{iso}},\zeta_{\mathrm{aniso}},\zeta_{\mathrm{ang}},\kappa$} are defined in~\eqref{eq:hybrid_coefficients}. The entropy flux is consistent with~\eqref{eq:ce_entropy_flux}, namely
\[
\vec\Phi^{(1)} = \frac{\vec Q^{(1)}}{\theta} = -\frac{\kappa}{\theta}\,\nabla_{\vec x}\theta,
\]
and the entropy production~\eqref{eq:hybrid_iota} is non-negative.
\end{theorem}

Relation~\eqref{eq:hybrid_Q} is Fourier's law, with Eucken-type conductivity satisfying
\[
\uz{\frac{\kappa}{\zeta_{\mathrm{shear}}}} = \tfrac{7}{2}\frac{k_B}{m} = c_p.
\]
Relation~\eqref{eq:hybrid_T} combines the Newtonian shear response with an isotropic bulk viscosity $\zeta_{\mathrm{iso}}$, generated by the translational--rotational equipartition imbalance of~\Cref{sec:zero}, and \uz{an Ericksen-aligned dissipative stress $\zeta_{\mathrm{aniso}}\,\mathcal{B}(A)\,(\nabla_{\vec x}\widehat{\vec{\nu}}\odot\nabla_{\vec x}\widehat{\vec{\nu}})$, whose trace part is the anisotropic bulk correction and whose deviatoric part is a texture-aligned contribution to the viscous stress}. Relation~\eqref{eq:hybrid_M} provides the couple-stress closure, including coupling both to \uz{the velocity gradient through the Ericksen contraction in $\mathcal{B}(A)$} and to the angular-momentum gradient through $\nabla_{\vec x}\vec\mu$\uz{; the angular-advection component contributes the spin-diffusion term $\fix{mC_{m,\ell}\,}\zeta_{\mathrm{ang}}\,\fix{(\nabla_{\vec x}\vec\mu)(\mat{I}-\widehat{\vec{\nu}}\otimes\widehat{\vec{\nu}})}$, an Eringen-type micropolar couple stress projected on the plane transverse to the director, with the kinetic prediction $\zeta_{\mathrm{ang}} = p_{\mathrm{tr,rot}}\tau/(m C_{m,\ell})$ for the spin viscosity}. \fix{The sign with which this term enters the director equation \eqref{eq:director_equation} makes spin diffusion dissipative, damping director oscillations.} In the isotropic limit $\nabla_{\vec x}\widehat{\vec{\nu}}\to 0$, $C_{m,\ell}\to 0$, the closure reduces to the monatomic-gas result
\[
\mathbb{T}^{(1)} = -p\,\mat{I} + 2\uz{\zeta_{\mathrm{shear}}}\,\mathbb D^d,
\qquad
\vec Q^{(1)} = -\kappa\,\nabla_{\vec x}\theta
\]
of e.g.~\cite[eq.~(98)]{malek}. The Ericksen stress and couple-stress contributions encoded in~\eqref{eq:hybrid_T}--\eqref{eq:hybrid_M} have no counterpart in that theory and constitute the structural novelty of the present kinetic theory of ordered fluids.

}

\uz{\subsection{The closed system and the director form of the angular-momentum balance}\label{sec:hybrid_system}
We now recast the angular-momentum balance as an evolution equation for the nematic director, following the discussion around \cite[eq.~(6.9)]{farrell}, and summarise the closed system obtained in this work. Recall the zeroth-order identification $\vec{\mu} = C_{m,\ell}\,\widehat{\vec{\nu}}\times\dot{\widehat{\vec{\nu}}}$ of \eqref{eq:balance_laws_zeroth}, where the dot denotes the material derivative. Since $\dot{\widehat{\vec{\nu}}}\times\dot{\widehat{\vec{\nu}}} = \vec{0}$,
\begin{equation}
	\rho\,\frac{d\vec{\mu}}{dt} \;=\; \rho\,C_{m,\ell}\,\widehat{\vec{\nu}}\times\ddot{\widehat{\vec{\nu}}}.
\end{equation}
Moreover, every contribution to the closed couple stress is, up to the higher-order terms already discarded in \Cref{sec:ce}, of the form $g\,(\widehat{\vec{\nu}}\times\nabla_{\vec x}\widehat{\vec{\nu}})$ or $\fix{mC_{m,\ell}^{2}\,}\zeta_{\mathrm{ang}}(\widehat{\vec{\nu}}\times\nabla_{\vec x}\dot{\widehat{\vec{\nu}}})$, and for any scalar field $g$ the identity
\begin{equation}
	\label{eq:cross_divergence_identity}
	\nabla_{\vec x}\!\cdot\!\bigl[g\,(\widehat{\vec{\nu}}\times\nabla_{\vec x}\widehat{\vec{\nu}})\bigr] \;=\; \widehat{\vec{\nu}}\times\nabla_{\vec x}\!\cdot\!\bigl(g\,\nabla_{\vec x}\widehat{\vec{\nu}}\bigr)
\end{equation}
holds, because $\varepsilon_{ikl}\,\partial_j\hat{\nu}_k\,\partial_j\hat{\nu}_l = 0$ by antisymmetry. The angular-momentum balance $\rho\,d\vec{\mu}/dt \fix{-} \nabla_{\vec x}\!\cdot\!\mathbb{M} = 0$ therefore takes the form
\begin{equation}
	\label{eq:director_cross_form}
	\widehat{\vec{\nu}}\times\Bigl[\,\rho\,C_{m,\ell}\,\ddot{\widehat{\vec{\nu}}} \;-\; \fix{m}\,\nabla_{\vec x}\!\cdot\!\Bigl(\bigl(K - \zeta_{\mathrm{aniso}}\,\mathcal{B}(A)\bigr)\,\nabla_{\vec x}\widehat{\vec{\nu}}\Bigr) \;\fix{-\; mC_{m,\ell}^{2}}\,\nabla_{\vec x}\!\cdot\!\bigl(\zeta_{\mathrm{ang}}\,\nabla_{\vec x}\dot{\widehat{\vec{\nu}}}\bigr)\,\Bigr] \;=\; \vec{0}.
\end{equation}
This prescribes that the component orthogonal to $\widehat{\vec{\nu}}$ of the bracketed director operator vanish. Exactly as in \cite[eq.~(6.14)]{farrell}, this constraint is expressed through a Lagrange multiplier $\mathfrak{t}(\vec{x},t)$, determined by the unit-length constraint on the director, and the angular-momentum balance becomes the \emph{director equation}
\begin{equation}
	\label{eq:director_equation}
	\rho\,C_{m,\ell}\,\ddot{\widehat{\vec{\nu}}} \;-\; \fix{m}\,\nabla_{\vec x}\!\cdot\!\Bigl[\bigl(K - \zeta_{\mathrm{aniso}}\,\mathcal{B}(A)\bigr)\,\nabla_{\vec x}\widehat{\vec{\nu}}\Bigr] \;\fix{-\; mC_{m,\ell}^{2}}\,\nabla_{\vec x}\!\cdot\!\bigl(\zeta_{\mathrm{ang}}\,\nabla_{\vec x}\dot{\widehat{\vec{\nu}}}\bigr) \;=\; \mathfrak{t}\,\widehat{\vec{\nu}}.
\end{equation}
The three terms are, respectively, the rotational inertia of the rods, the elastic restoring force weakened by the anisotropic dissipative coupling, and the spin diffusion of \eqref{eq:hybrid_M}. In the inviscid limit $\tau \to 0$ all dissipative coefficients vanish and \eqref{eq:director_equation} reduces to the inertial director equation of \cite{farrell}.

For ease of reference we collect the closed macroscopic system derived in this work for the unknowns $(\rho, \vec{u}, \widehat{\vec{\nu}}, \theta, \mathfrak{t})$:
\begin{subequations}
\label{eq:final_system}
\begin{align}
	&\partial_t\rho + \nabla_{\vec{x}}\!\cdot\!(\rho\vec{u}) = 0, \\
	&\rho\bigl[\partial_t\vec{u} + (\nabla_{\vec{x}}\vec{u})\vec{u}\bigr] = \nabla_{\vec{x}}\!\cdot\!\Bigl[-p\,\mat{I} - K\,(\nabla_{\vec x}\widehat{\vec{\nu}}\odot\nabla_{\vec x}\widehat{\vec{\nu}}) + 2\zeta_{\mathrm{shear}}\,\mathbb{D}^d\nonumber\\
	&\hspace{11em}+ \zeta_{\mathrm{iso}}\,(\nabla_{\vec{x}}\!\cdot\!\vec{u})\,\mat{I} + \zeta_{\mathrm{aniso}}\,\mathcal{B}(A)\,(\nabla_{\vec x}\widehat{\vec{\nu}}\odot\nabla_{\vec x}\widehat{\vec{\nu}})\Bigr], \\
	&\rho\,C_{m,\ell}\,\ddot{\widehat{\vec{\nu}}} - \fix{m}\,\nabla_{\vec x}\!\cdot\!\Bigl[\bigl(K - \zeta_{\mathrm{aniso}}\,\mathcal{B}(A)\bigr)\nabla_{\vec x}\widehat{\vec{\nu}}\Bigr] \fix{-\, mC_{m,\ell}^{2}}\,\nabla_{\vec x}\!\cdot\!\bigl(\zeta_{\mathrm{ang}}\,\nabla_{\vec x}\dot{\widehat{\vec{\nu}}}\bigr) = \mathfrak{t}\,\widehat{\vec{\nu}}, \\
	&\rho\bigl[\partial_t \widetilde{e} + \vec{u}\cdot\nabla_{\vec{x}} \widetilde{e}\bigr] - \mathbb{T}\!:\!\nabla_{\vec{x}}\vec{u} - \fix{\tfrac{1}{mC_{m,\ell}}}\mathbb{M}\!:\!\nabla_{\vec{x}}\vec{\mu} - \nabla_{\vec{x}}\!\cdot\!\bigl(\kappa\,\nabla_{\vec{x}}\theta\bigr) = \fix{0}, \\
	&\fix{p = R_s\,\rho\,\theta}, \qquad \theta = \tfrac{2}{5R_s}\,\widetilde{e}, \qquad \abs{\widehat{\vec{\nu}}} = 1,
\end{align}
\end{subequations}
where $\mathbb{T}$ and $\mathbb{M}$ in the energy balance are the closed stresses of \eqref{eq:hybrid_T}--\eqref{eq:hybrid_M}, $\vec{\mu} = C_{m,\ell}\widehat{\vec{\nu}}\times\dot{\widehat{\vec{\nu}}}$, $\widetilde{e} = e_{\mathrm{tr}} + e_{\mathrm{rot}}$ the kinetic internal energy, \fix{$K = K_{\mathrm{kin}}$ the Frank modulus} of \Cref{rmk:Kkin_emerges}, and $\mathcal{B}(A)$ as in \eqref{eq:hybrid_B}. \fix{The energy balance is written for the kinetic internal energy $\widetilde{e}$ and is source-free by \eqref{eq:energy_balance_law}: under the micro--macro identification the entire Vlasov power is absorbed by the budget of the bulk intrinsic angular momentum, see \Cref{rmk:mu_source_in_energy}.} The system \eqref{eq:final_system} is a compressible, viscous, heat-conducting and spin-diffusive variant of the Leslie--Ericksen equations, reducing to the inviscid compressible system of \cite{farrell} when $\tau \to 0$.}

\uz{\subsection{The incompressible model}\label{sec:hybrid_incompressible}
The Rajagopal--Srinivasa procedure delivers the incompressible counterpart of \eqref{eq:final_system} with no additional effort: it suffices to append to the constrained maximisation of \Cref{sec:hybrid_RSP} a further Lagrange multiplier $\mathfrak{p}(\vec{x},t)$ enforcing the divergence-free constraint on the affinities,
\begin{equation}
	\mathrm{maximize}_{A}\ \theta\xi(A) + \Lambda\bigl(\theta\xi(A)-\theta\xi(J,A)\bigr) + \mathfrak{p}\,A_{\mathrm{tr}\mathbb{T}},
	\qquad
	A_{\mathrm{tr}\mathbb{T}} = \nabla_{\vec{x}}\!\cdot\!\vec{u} = 0.
\end{equation}
The stationarity conditions for the deviatoric, couple-stress and heat affinities are unchanged, so the closures for $\mathbb{J}_{\mathbb{T}^d}$, $\mathbb{J}_{\mathbb{M}}$ and $\vec{J}_{\vec{Q}}$ carry over verbatim, with $\mathcal{B}(A)$ evaluated on divergence-free fields. The stationarity condition in the trace direction instead reads
\begin{equation}
	\frac{\partial\,\theta\xi(A)}{\partial A_{\mathrm{tr}\mathbb{T}}} = \frac{\Lambda}{1+\Lambda}\,J_{\mathrm{tr}\mathbb{T}} - \frac{\mathfrak{p}}{1+\Lambda}.
\end{equation}
The spherical part of the stress is no longer determined constitutively but acts as the \emph{reaction} to the constraint, exactly as the pressure of the incompressible Navier--Stokes equations. Writing $\widetilde{\mathfrak{p}}$ for the resulting reaction pressure, which absorbs $p$, $\zeta_{\mathrm{iso}}\nabla_{\vec x}\!\cdot\!\vec{u}$ and the trace of the anisotropic component, and is determined by the incompressibility constraint rather than by the equation of state, the incompressible \emph{inhomogeneous} fluid model reads
\begin{subequations}
\label{eq:incompressible_system}
\begin{align}
	&\nabla_{\vec{x}}\!\cdot\!\vec{u} = 0, \qquad \partial_t\rho + \vec{u}\cdot\nabla_{\vec{x}}\rho = 0,\\
	&\rho\bigl[\partial_t\vec{u} + (\nabla_{\vec{x}}\vec{u})\vec{u}\bigr] = -\nabla_{\vec{x}}\widetilde{\mathfrak{p}} + \nabla_{\vec{x}}\!\cdot\!\Bigl[- K\,(\nabla_{\vec x}\widehat{\vec{\nu}}\odot\nabla_{\vec x}\widehat{\vec{\nu}}) + 2\zeta_{\mathrm{shear}}\,\mathbb{D} + \zeta_{\mathrm{aniso}}\,\mathcal{B}(A)\,(\nabla_{\vec x}\widehat{\vec{\nu}}\odot\nabla_{\vec x}\widehat{\vec{\nu}})^{d}\Bigr], \\
	&\rho\,C_{m,\ell}\,\ddot{\widehat{\vec{\nu}}} - \fix{m}\,\nabla_{\vec x}\!\cdot\!\Bigl[\bigl(K - \zeta_{\mathrm{aniso}}\,\mathcal{B}(A)\bigr)\nabla_{\vec x}\widehat{\vec{\nu}}\Bigr] \fix{-\, mC_{m,\ell}^{2}}\,\nabla_{\vec x}\!\cdot\!\bigl(\zeta_{\mathrm{ang}}\,\nabla_{\vec x}\dot{\widehat{\vec{\nu}}}\bigr) = \mathfrak{t}\,\widehat{\vec{\nu}}, \\
	&\rho\bigl[\partial_t \widetilde{e} + \vec{u}\cdot\nabla_{\vec{x}} \widetilde{e}\bigr] - \mathbb{T}\!:\!\nabla_{\vec{x}}\vec{u} - \fix{\tfrac{1}{mC_{m,\ell}}}\mathbb{M}\!:\!\nabla_{\vec{x}}\vec{\mu} - \nabla_{\vec{x}}\!\cdot\!\bigl(\kappa\,\nabla_{\vec{x}}\theta\bigr) = \fix{0}, \qquad \abs{\widehat{\vec{\nu}}} = 1,
\end{align}
\end{subequations}
in which $\mathcal{B}(A) = (\nabla_{\vec x}\widehat{\vec{\nu}}\odot\nabla_{\vec x}\widehat{\vec{\nu}})^{d}\!:\!\mathbb{D} \fix{-} C_{m,\ell}^{-1}(\widehat{\vec{\nu}}\times\nabla_{\vec x}\widehat{\vec{\nu}})\!:\!\nabla_{\vec x}\vec{\mu}$ on divergence-free fields. For homogeneous initial data $\rho \equiv \mathrm{const}$, the first line reduces to the usual incompressibility statement, and \eqref{eq:incompressible_system} is an incompressible, viscous, spin-diffusive Leslie--Ericksen-type model with one-constant Frank elasticity.}
\fix{Equation~\eqref{eq:incompressible_system} is the incompressible specialisation of the present viscous theory, which generalises the inviscid Leslie--Ericksen system of Farrell, Russo and Zerbinati \cite{farrell} to include viscous, heat-conducting and spin-diffusive effects. The transport coefficients $\zeta_{\mathrm{shear}}$, $\zeta_{\mathrm{aniso}}$, $\zeta_{\mathrm{ang}}$ and $\kappa$ are all proportional to the relaxation time $\tau$, so in the inviscid limit $\tau\to 0$ the dissipative terms vanish and the director dynamics reduces to the Ericksen elasticity of \cite{farrell} with the one-constant Frank modulus $K = K_{\mathrm{kin}} = C_{m,\ell}\,p/m$ of \Cref{rmk:Kkin_emerges}. The exact term-by-term dictionary between the conjugate-momentum variables used here and the angular-velocity variables of \cite{farrell} is recorded in \Cref{app:energy_balance_xi}, see \eqref{eq:app_frz_dictionary}.}

\uz{\begin{remark}[\fix{Recovery of the Clausius--Duhem identity}]\label{rmk:no_interaction_limit}
\fix{Because the constitutive free energy carries no interaction term, the heat flux is the bare kinetic flux $\vec{Q}$ and} \Cref{cor:first_order_entropy_production} \fix{is} the local Clausius--Duhem identity $\theta\,\xi = (\mathbb{T}+p\mat{I})\!:\!\nabla_{\vec{x}}\vec{u} + \mathbb{M}\!:\!\nabla_{\vec{x}}\vec{\mu} - \vec{Q}^{(1)}\!\cdot\!\nabla_{\vec{x}}\theta/\theta$, the natural extension of \cite[Theorem~1]{farrell} to the ordered fluid setting, and the Rajagopal--Srinivasa procedure of \uz{\Cref{sec:hybrid_RSP}} coincides pointwise with \cite[Principle~1]{farrell}.
\end{remark}}

\begin{remark}[Role of the alignment hypothesis]\label{rmk:alignment_in_CE}
	The closure of \Cref{thm:hybrid_closure} inherits two distinct uses of the aligned regime $s\to 1$ from~\Cref{sec:ce}. The first is structural: the explicit form of $f^{(1)}$ in~\eqref{eq:ce_f1_explicit} is computed under the assumption that the orientational distribution $\lambda$ is concentrated about the director $\widehat{\vec{\nu}}$, which makes $\mathcal{D}\log\lambda$ vanish and renders the streaming-derivative calculation tractable. The second is that the Vlasov power $\rho\cchevrons{\vec{\mathcal V}\cdot\dot{\vec{\nu}}}$ that appears on the right-hand side of \uz{the entropy-production identity \eqref{eq:entropy_production_constitutive}} could be cancelled \emph{pointwise} via the orthogonality identity $\vec{\mathcal V}\cdot\dot{\vec{\nu}} = s(\widehat{\vec{\nu}}\cdot\dot{\widehat{\vec{\nu}}}) = 0$ at $s\to 1$. \uz{The Helmholtz free-energy formulation of \Cref{sec:helmholtz}} resolves the second use \emph{globally} via \Cref{lemma:vlasov_power}, regardless of the value of $s$, by absorbing the integrated Vlasov power into $d\mathcal{U}_{\mathrm{int}}/dt$. The first use, by contrast, remains necessary for the Chapman--Enskog hypothesis about the structure of $f^{(1)}$ and is not lifted by the present reformulation. Relaxing this hypothesis would require a more general Chapman--Enskog expansion in which $\mathcal{D}\log\lambda$ is retained and a constitutive relation for $\lambda$-relaxation is derived from the Vlasov self-consistency equation. We leave this generalisation for future work.
\end{remark}

\fix{\begin{remark}[Frank constant]\label{rmk:frank_constant_augmented}
The Ericksen-stress identification produced by the closure \eqref{eq:hybrid_T}, when read against the zeroth-order rotational specific energy of \Cref{rmk:Kkin_emerges}, gives the one-constant Frank modulus $K = K_{\mathrm{kin}} = C_{m,\ell}\,p_{\mathrm{tr,rot}}/m$, the inviscid modulus of \cite{farrell}.
\end{remark}}

\ack{This work was funded by 
the Engineering and Physical Sciences Research Council [grant number EP/W026163/1],
the Science and Technology Facilities Council [grant number UKRI/ST/B000495/1],
the Donatio Universitatis Carolinae Chair ``Mathematical modelling of multicomponent systems'', 
the UKRI Digital Research Infrastructure Programme through the Science and Technology Facilities Council's Computational Science Centre for Research Communities (CoSeC),
and
the Swedish Research Council under grant no.~Z2021-06594 while in residence at Institut Mittag-Leffler in Djursholm, Sweden.
The authors gratefully acknowledge the hospitality of the Erwin Schrödinger International Institute for Mathematics and Physics, Vienna, where part of this work was carried out.
J.~M\'{a}lek acknowledges the support of the project No.~25-16592S financed by the Czech Science Foundation (GA\v{C}R). P.~E.~Farrell, J.~M\'{a}lek and O.~Sou\v{c}ek are members of the Nečas Center for Mathematical Modeling. For the purpose of open access, the authors have applied a CC BY public copyright licence to any author accepted manuscript arising from this submission.
No new data were generated or analysed during this study.}

\appendix
\makeatletter\gdef\theequation{\thesection\hskip2pt\@arabic\c@equation}\makeatother
\cla{
\section{Derivation of the Maxwellian}\label{app:maxwellian_derivation}
The Maxwellian distribution \eqref{eq:maxwellian} defining the BGK equilibrium of \Cref{sec:bgk} arises as the unique distribution annihilating the full collision integral \eqref{eq:collision} on its collision invariants. In this appendix we give a self-contained derivation, following the lines of \cite[Section~3.4 and Theorem~3.13]{carrillo}.

By the Boltzmann inequality \cite[Theorem~3.12]{carrillo}, which for the rigid-rod rule \eqref{eq:binary_rule} can be proved via the detailed balance condition \cite{curtissV} or, for more general molecular shapes, via the reciprocity balance of \cite{cercignaniLampis}, $\int \mathcal{Q}(f,f)\log f\,d\vec\Xi = 0$ if and only if $\log f$ is a collision invariant for the operator $\mathcal{Q}$ of \eqref{eq:collision}. The collision invariants identified in \eqref{eq:collinv} generate the affine span of
\begin{equation}
	\bigl\{ 1,\ \vec p,\ \tfrac{1}{2}\bigl(\abs{\vec p}^2/m + \vec\varsigma^T \mat{B}^{-1}\vec\varsigma\bigr),\ \vec{\nu}\times\vec\varsigma + \vec x\times\vec p \bigr\},
\end{equation}
together with arbitrary functions of $\vec{\nu}$ alone, which are annihilated by the collision integral because $\vec{\nu}$ is preserved by the binary rule \eqref{eq:binary_rule}. Equality in the Boltzmann inequality therefore forces $\log f$ into this span,
\begin{equation}
	\label{eq:maxwellian_logf}
	\log f \;=\; a(\vec x, t) \;+\; \vec b(\vec x, t)\cdot\vec p \;+\; \vec d(\vec x, t)\cdot(\vec{\nu}\times\vec\varsigma) \;-\; c(\vec x, t)\,\tfrac{1}{2}\bigl(\abs{\vec p}^2/m + \vec\varsigma^T \mat{B}^{-1}\vec\varsigma\bigr) \;+\; g\bigl(\vec{\nu}, t\bigr),
\end{equation}
for scalar fields $a, c$, vector fields $\vec b, \vec d$ depending on $(\vec x, t)$ alone, and an arbitrary function $g(\vec{\nu}, t)$. Completing the square in $\vec p$ produces a Gaussian in the peculiar velocity $\vec V = \vec v - \vec u$ with $\vec u$ fixed by $\vec b$. Completing the square in $\vec\varsigma$ couples the angular-momentum term $\vec d\cdot(\vec{\nu}\times\vec\varsigma) = (\vec d\times\vec{\nu})\cdot\vec\varsigma$ to the rotational quadratic, recentring the rotational Gaussian at the rigid co-rotation $C_{m,\ell}(\vec{\Omega}\times\vec{\nu})$ with $\vec{\Omega} = \vec d/c$, so the peculiar conjugate momentum is $\vec\Sigma = \vec\varsigma - C_{m,\ell}(\vec{\Omega}\times\vec{\nu})$ and the Gaussian width is fixed by $c$. The drift is the rigid co-rotation rather than a uniform shift of $\vec\varsigma$ precisely because the conserved rotational invariant is $\vec{\nu}\times\vec\varsigma$ and not $\vec\varsigma$ itself. Identifying the macroscopic density $\rho = m\int f\,d\vec p\,d\vec{\nu}\,d\vec\varsigma$, the bulk velocity $\vec u = \cchevrons{\vec v}$, the bulk intrinsic angular momentum $\vec\mu = \cchevrons{\vec{\nu}\times\vec\varsigma}$, which fixes $\vec{\Omega}$ through $\vec\mu = C_{m,\ell}(\mat{I} - \mat{S})\vec{\Omega}$ with $\mat{S} = \int_{\mathcal{M}}\lambda\,\vec{\nu}\otimes\vec{\nu}\,d\vec{\nu}$, the orientational distribution $\lambda(\vec{\nu}, t) \propto \exp(g(\vec{\nu}, t))$, and the temperature parameter $\theta = 1/(k_B c)$, the equality case \eqref{eq:maxwellian_logf} reproduces the Maxwellian \eqref{eq:maxwellian}. At this stage $\theta$ is merely the Lagrange multiplier dual to the kinetic energy; it acquires its thermodynamic meaning on evaluating the total internal energy density at the Maxwellian,
\begin{equation}
	\rho \widetilde{e}^{(0)} \;=\; \frac{1}{2}\rho\cchevrons{\abs{\vec V}^2} + \frac{1}{2 \fix{m} C_{m,\ell}}\rho\cchevrons{\abs{\vec\Sigma}^2} \;=\; \frac{3}{2}\rho k_B\theta/m \;+\; \rho k_B\theta/m \;=\; \frac{5}{2}\,\rho k_B\theta/m,
\end{equation}
whence the equipartition relation $\theta = (2/(5R_s))\,\widetilde{e}$ with $R_s = k_B/m$ identifies $\theta$ with the thermodynamic temperature $\theta = \partial \widetilde{e}/\partial\eta^{(0)}$ of \eqref{def:temperature}. The temperature appearing in \eqref{eq:maxwellian} is therefore, by construction, the unique solution of the caloric identity $\widetilde{e} = (5/2)R_s\theta$ with $\widetilde{e}$ the total internal energy of the distribution $f$.}

\fix{\section{The balance laws of angular momentum and energy in the conjugate-momentum variables}\label{app:energy_balance_xi}
In this appendix we derive the angular-momentum and energy balance laws of \Cref{sec:hyd} from the moment equation \eqref{eq:enskog_moment_equation}, without invoking the micro--macro identification, and we record the exact correspondence with the angular-velocity formulation of \cite{farrell}. The balance \eqref{eq:balance_law_angular_momentum} carries the mean-field torque, absent in \cite{farrell}, so the reduction of the rotational kinetic energy to its internal part, which multiplies \eqref{eq:balance_law_angular_momentum} by $\vec{\mu}/(mC_{m,\ell})$, generates the work exerted by the mean-field torque, $\rho\,\vec{\mu}\cdot\cchevrons{\vec{\nu}\times\vec{\mathcal{V}}}/(mC_{m,\ell})$, and the derivation below exhibits where this term ends up in the energy budget.

Throughout we write $\vec{\lambda} := \vec{\nu}\times\vec{\varsigma}$ for the specific intrinsic angular momentum, $\vec{\mu} = \cchevrons{\vec{\lambda}}$, $\vec{\Lambda} = \vec{\lambda} - \vec{\mu}$ for its peculiar part as in \eqref{eq:peculiar_angular_momentum}, and $\mathbb{M} = -\rho\cchevrons{\vec{V}\otimes\vec{\Lambda}}$ for the couple stress, which under the micro--macro identification coincides with \eqref{eq:couple_stress} by \Cref{rmk:peculiar_angular_momentum}, with the flux index first, so that $(\nabla_{\vec{x}}\cdot\mathbb{M})_j = \partial_i \mathbb{M}_{ij}$ and $\mathbb{M}:\nabla_{\vec{x}}\vec{\mu} = \mathbb{M}_{ij}\partial_i\mu_j$. Two pointwise identities, both consequences of the rigid-rod constraint $\vec{\varsigma}\cdot\vec{\nu} = 0$, are used repeatedly,
\begin{equation}
	\label{eq:app_pointwise}
	\abs{\vec{\lambda}}^2 = \abs{\vec{\varsigma}}^2,
	\qquad
	\vec{\mathcal{V}}\cdot\vec{\varsigma} = \vec{\lambda}\cdot(\vec{\nu}\times\vec{\mathcal{V}}).
\end{equation}

\subsection{The angular-momentum balance}
Substituting $\psi = \vec{\lambda}$ in the moment equation \eqref{eq:enskog_moment_equation} and using $\cchevrons{\dot{\vec{\nu}}\cdot\nabla_{\vec{\nu}}\vec{\lambda}} = \cchevrons{\dot{\vec{\nu}}\times\vec{\varsigma}} = \vec{0}$ together with $\cchevrons{\vec{\mathcal{V}}\cdot\nabla_{\vec{\varsigma}}\vec{\lambda}} = \cchevrons{\vec{\nu}\times\vec{\mathcal{V}}}$ gives
\begin{equation}
	\label{eq:app_mu_pre}
	\partial_t(\rho\vec{\mu}) + \nabla_{\vec{x}}\cdot\bigl(\rho\cchevrons{\vec{v}\otimes\vec{\lambda}}\bigr) = \rho\cchevrons{\vec{\nu}\times\vec{\mathcal{V}}}.
\end{equation}
The flux decomposes exactly, $\cchevrons{\vec{v}\otimes\vec{\lambda}} = \vec{u}\otimes\vec{\mu} + \cchevrons{\vec{V}\otimes\vec{\lambda}}$ and $\cchevrons{\vec{V}\otimes\vec{\lambda}} = \cchevrons{\vec{V}\otimes\vec{\Lambda}}$, because $\cchevrons{\vec{V}} = \vec{0}$ and $\vec{\mu}$ is independent of the phase variables. No decorrelation hypothesis between $\vec{V}$ and $\vec{\nu}$ is needed, which implies we are not allowed to move $\vec{V}$ out of the above averages. Notice that far from an aligned state such decorrelation would be unphysical hence why in this appendix we have dropped this hypothesis. With the continuity equation \eqref{eq:continuity_eq},
\begin{equation}
	\label{eq:app_mu_balance}
	\rho\bigl[\partial_t\vec{\mu} + (\nabla_{\vec{x}}\vec{\mu})\vec{u}\bigr] - \nabla_{\vec{x}}\cdot\mathbb{M} = \rho\cchevrons{\vec{\nu}\times\vec{\mathcal{V}}},
\end{equation}
which is \eqref{eq:balance_law_angular_momentum}.

\subsection{The rotational, kinetic and internal energy}
Substituting the rotational kinetic energy $\psi = \abs{\vec{\varsigma}}^2/(2C_{m,\ell})$ in \eqref{eq:enskog_moment_equation}, dividing by $m$ to work with specific energies, and using \eqref{eq:app_pointwise} together with $\cchevrons{\dot{\vec{\nu}}\cdot\nabla_{\vec{\nu}}\abs{\vec{\lambda}}^2} = 0$ yields the balance of the total specific rotational kinetic energy,
\begin{equation}
	\label{eq:app_rot_total}
	\partial_t\left(\rho\,\frac{\cchevrons{\abs{\vec{\lambda}}^2}}{2mC_{m,\ell}}\right) + \nabla_{\vec{x}}\cdot\left(\rho\,\frac{\cchevrons{\vec{v}\abs{\vec{\lambda}}^2}}{2mC_{m,\ell}}\right) = \frac{\rho}{mC_{m,\ell}}\cchevrons{\vec{\lambda}\cdot(\vec{\nu}\times\vec{\mathcal{V}})}.
\end{equation}
Both sides split along $\vec{\lambda} = \vec{\mu} + \vec{\Lambda}$. For the densities and the source,
\begin{equation}
	\cchevrons{\abs{\vec{\lambda}}^2} = \abs{\vec{\mu}}^2 + \cchevrons{\abs{\vec{\Lambda}}^2},
	\qquad
	\cchevrons{\vec{\lambda}\cdot(\vec{\nu}\times\vec{\mathcal{V}})} = \vec{\mu}\cdot\cchevrons{\vec{\nu}\times\vec{\mathcal{V}}} + \cchevrons{\vec{\Lambda}\cdot(\vec{\nu}\times\vec{\mathcal{V}})},
\end{equation}
while the flux carries a cross term proportional to the couple stress,
\begin{equation}
	\label{eq:app_flux_split}
	\rho\,\frac{\cchevrons{\vec{v}\abs{\vec{\lambda}}^2}}{2mC_{m,\ell}}
	= \rho\vec{u}\,\frac{\abs{\vec{\mu}}^2 + \cchevrons{\abs{\vec{\Lambda}}^2}}{2mC_{m,\ell}}
	- \frac{\mathbb{M}\vec{\mu}}{mC_{m,\ell}}
	+ \rho\,\frac{\cchevrons{\vec{V}\abs{\vec{\Lambda}}^2}}{2mC_{m,\ell}},
\end{equation}
since $\rho\cchevrons{\vec{V}(\vec{\mu}\cdot\vec{\Lambda})} = -\mathbb{M}\vec{\mu}$.

Dotting \eqref{eq:app_mu_balance} with $\vec{\mu}/(mC_{m,\ell})$ and using the continuity equation to pass to conservative form gives the budget of the kinetic energy of the bulk intrinsic angular momentum,
\begin{equation}
	\label{eq:app_bulk_spin}
	\partial_t\left(\rho\,\frac{\abs{\vec{\mu}}^2}{2mC_{m,\ell}}\right) + \nabla_{\vec{x}}\cdot\left(\rho\vec{u}\,\frac{\abs{\vec{\mu}}^2}{2mC_{m,\ell}} - \frac{\mathbb{M}\vec{\mu}}{mC_{m,\ell}}\right) + \frac{\mathbb{M}:\nabla_{\vec{x}}\vec{\mu}}{mC_{m,\ell}} = \frac{\rho}{mC_{m,\ell}}\,\vec{\mu}\cdot\cchevrons{\vec{\nu}\times\vec{\mathcal{V}}}.
\end{equation}
This is where the product of $\vec{\mu}$ with the source of the angular-momentum balance lives: the work exerted by the mean-field torque contributes to the kinetic energy $\abs{\vec{\mu}}^2/(2mC_{m,\ell})$ of the bulk intrinsic angular momentum, not the internal energy. Subtracting \eqref{eq:app_bulk_spin} from \eqref{eq:app_rot_total}, the cross-flux $\mathbb{M}\vec{\mu}/(mC_{m,\ell})$ cancels and the rotational internal energy $e_{\mathrm{rot}} = \cchevrons{\abs{\vec{\Lambda}}^2}/(2mC_{m,\ell})$ obeys
\begin{equation}
	\label{eq:app_erot_balance}
	\rho\bigl[\partial_t e_{\mathrm{rot}} + \vec{u}\cdot\nabla_{\vec{x}} e_{\mathrm{rot}}\bigr] - \frac{\mathbb{M}:\nabla_{\vec{x}}\vec{\mu}}{mC_{m,\ell}} + \nabla_{\vec{x}}\cdot\left(\rho\,\frac{\cchevrons{\vec{V}\abs{\vec{\Lambda}}^2}}{2mC_{m,\ell}}\right) = \frac{\rho}{mC_{m,\ell}}\cchevrons{\vec{\Lambda}\cdot(\vec{\nu}\times\vec{\mathcal{V}})}.
\end{equation}
Notice that the peculiar Vlasov power survives on the right-hand side. Under the micro--macro identification the torque factors out of the average, $\cchevrons{\vec{\Lambda}\cdot(\vec{\nu}\times\vec{\mathcal{V}})} = \cchevrons{\vec{\Lambda}}\cdot(\widehat{\vec{\nu}}\times\vec{\mathcal{V}}) = 0$, while $e_{\mathrm{rot}} = \cchevrons{\abs{\vec{\Sigma}}^2}/(2mC_{m,\ell})$ by \Cref{rmk:peculiar_angular_momentum}, and \eqref{eq:app_erot_balance} reduces to the source-free balance \eqref{eq:energy_balance_law_final}.

The translational component is handled exactly as in \cite[Section 3]{malek}, subtracting $\vec{u}\cdot\eqref{eq:linear_momentum_balance_law}$ from the moment equation for $\abs{\vec{p}}^2/(2m)$. Adding the result to \eqref{eq:app_erot_balance} produces the internal energy balance
\begin{equation}
	\label{eq:app_etilde_general}
	\rho\bigl[\partial_t \widetilde{e} + \vec{u}\cdot\nabla_{\vec{x}} \widetilde{e}\bigr] - \mathbb{T}:\nabla_{\vec{x}}\vec{u} - \frac{\mathbb{M}:\nabla_{\vec{x}}\vec{\mu}}{mC_{m,\ell}} + \nabla_{\vec{x}}\cdot\vec{Q} = \frac{\rho}{m}\cchevrons{\vec{\mathcal{V}}\cdot\dot{\vec{\nu}}} - \frac{\rho}{mC_{m,\ell}}\,\vec{\mu}\cdot\cchevrons{\vec{\nu}\times\vec{\mathcal{V}}},
\end{equation}
with $\widetilde{e}$ and $\vec{Q}$ as in \Cref{sec:hyd}, which reduces to the source-free \eqref{eq:energy_balance_law} under the micro--macro identification by \eqref{eq:Vlasov_power}.

\subsection{Correspondence with the angular-velocity formulation}
Under the micro--macro identification the dictionary between the conjugate-momentum variables used here and the angular-velocity variables of \cite{farrell} is exact. With $\vec{\omega} = \widehat{\vec{\nu}}\times\dot{\vec{\nu}}$ the angular velocity of a vanishing-girth rod, whose axial spin is suppressed, and $\mathbb{I} = C_{m,\ell}$ the transverse moment of inertia of \Cref{sec:vanishing_girth},
\begin{equation}
	\label{eq:app_frz_dictionary}
	\vec{\varsigma} = C_{m,\ell}\dot{\vec{\nu}}, \qquad
	\vec{\lambda} = \mathbb{I}\vec{\omega}, \qquad
	\vec{\mu} = \cchevrons{\mathbb{I}\vec{\omega}}, \qquad
	\vec{\Lambda} = \mathbb{I}\vec{\Omega}, \qquad
	\vec{\Omega} := \vec{\omega} - \cchevrons{\vec{\omega}},
\end{equation}
so the intrinsic angular momentum $\vec{\eta} = \cchevrons{\mathbb{I}\vec{\omega}}$ of \cite[eq.~2.19]{farrell} coincides with $\vec{\mu}$, their angular-momentum flux $\cchevrons{\vec{V}\otimes\mathbb{I}\vec{\omega}} = \cchevrons{\vec{V}\otimes\mathbb{I}\vec{\Omega}}$ coincides with $-\mathbb{M}/\rho$, and their balance \cite[eq.~2.22c]{farrell},
\begin{equation}
	\rho\bigl[\partial_t\vec{\eta} + (\nabla_{\vec{x}}\vec{\eta})\vec{u}\bigr] + \nabla_{\vec{x}}\cdot\bigl(\rho\cchevrons{\vec{V}\otimes\mathbb{I}\vec{\omega}}\bigr) = \vec{0},
\end{equation}
is \eqref{eq:app_mu_balance} with vanishing right-hand side: the mean-field torque $\rho\cchevrons{\vec{\nu}\times\vec{\mathcal{V}}} = \rho\,\widehat{\vec{\nu}}\times\vec{\mathcal{V}}$ is the only structural difference between the two angular-momentum balances. The energy reduction of \cite[eqs.~3.7--3.9]{farrell} converts term by term as well: their rotational internal-energy density $n\cchevrons{\tfrac{1}{2}\vec{\Omega}\cdot\mathbb{I}\vec{\Omega}}$ equals $\rho\,e_{\mathrm{rot}}$, their work term $n\cchevrons{\vec{V}\otimes\mathbb{I}\vec{\Omega}}:\nabla_{\vec{x}}\cchevrons{\vec{\omega}}$ equals $-\mathbb{M}:\nabla_{\vec{x}}\vec{\mu}/(mC_{m,\ell})$, and their rotational heat flux $n\cchevrons{\vec{V}\tfrac{1}{2}\vec{\Omega}\cdot\mathbb{I}\vec{\Omega}}$ equals $\rho\cchevrons{\vec{V}\abs{\vec{\Sigma}}^2/(2mC_{m,\ell})}$. Their reduction step, the scalar product of \cite[eq.~2.22c]{farrell} with $\cchevrons{\vec{\omega}}$, here produces in addition the work exerted by the torque on the right-hand side of \eqref{eq:app_bulk_spin}, which by \eqref{eq:Vlasov_power} absorbs the entire Vlasov power. At every aligned Gibbs equilibrium the Vlasov torque is parallel to the director, see \Cref{corollary:kuramoto,corollary:onsager}, and every torque term vanishes, so the closure of \Cref{sec:hybrid} is common to the two formulations.}


\bibliographystyle{RS}

\bibliography{sample}

\end{document}